\documentclass[a4paper,10pt]{article}
\usepackage{amscd}
\usepackage{amsmath}
\usepackage{bbding}
\usepackage{amsfonts}
\usepackage{cite}
\usepackage{mathrsfs}
\usepackage{graphicx}
\usepackage{subfigure}
\usepackage{amssymb}
\usepackage{float}
\usepackage{graphicx,latexsym,amsmath,amssymb,amsthm,enumerate}
\usepackage{color}
\allowdisplaybreaks[2]   
9.4in\textwidth 6.4in \hoffset -2.0cm \DeclareFontEncoding{OT2}{}{}
\DeclareTextSymbol{\cyrsftsn}{OT2}{126}
\DeclareTextSymbol{\textnumero}{OT2}{125}

\renewcommand{\theequation}{\thesection.\arabic{equation}}
\theoremstyle{definition}
\newtheorem{theorem}{Theorem}[section]
\newtheorem{corollary}{Corollary}[section]

\newtheorem{remark}{Remark}[section]

\newtheorem{lemma}{Lemma}[section]
\begin{document}
\title{{\LARGE\bf{ Asset Prices with Investor Protection and Past Information}\thanks{This work was supported by  the National Natural Science Foundation of China (11471230, 11671282, 11801462).}}}
\author{Jia Yue$^a$, Ben-Zhang Yang$^b$, Ming-Hui Wang$^b$, Nan-Jing Huang$^b$
\footnote{Corresponding author.  E-mail addresses: nanjinghuang@hotmail.com, njhuang@scu.edu.cn} \\
{\small\it a. Department of Economic Mathematics, Southwestern University of Finance and Economics,}\\
{\small\it Chengdu, Sichuan 610074, P.R. China}\\{\small\it b. Department of Mathematics, Sichuan University, Chengdu,
Sichuan 610064, P.R. China}}
\date{ }
\maketitle
\begin{flushleft}
\hrulefill\\
\end{flushleft}
 {\bf Abstract}.
In this paper, we consider a dynamic asset pricing model in an approximate fractional economy to address empirical regularities related to both investor protection and past information. Our newly developed model features not only in terms with a controlling shareholder who diverts a fraction of the output, but also good (or bad) memory in his budget dynamics which can be well-calibrated by a pathwise way from the historical data. We find that poorer investor protection leads to higher stock holdings of controlling holders, lower gross stock returns, lower interest rates, and lower modified stock volatilities if the ownership concentration is sufficiently high. More importantly, by establishing an approximation scheme for good (bad) memory of investors on the historical market information, we conclude that good (bad) memory would increase (decrease) aforementioned dynamics and reveal that good (bad) memory strengthens (weakens) investor protection for minority shareholder when the ownership concentration is sufficiently high, while good (bad) memory inversely weakens (strengthens) investor protection for minority shareholder when the ownership concentration is sufficiently low. Our model's implications are consistent with a number of interesting facts documented in the recent literature.

\section{Introduction}\label{section1} \noindent
\setcounter{equation}{0}
The corporation owners are often classified into two types of shareholders---controlling shareholders and minority shareholders. The controlling shareholders dominate the corporation (e.g., \cite{Gorton,La}) and then they can divert resources for their private benefits. Therefore, it is very demanding to protect minority shareholders from expropriation by controlling shareholders. Investor protection has multitudinous effects on the economy (such as asset prices, consumption, etc.) and this may shed light on how investor protection exactly affects shareholders'  stock holdings, stock returns and volatilities, interest rates, etc. Furthermore, the market state often depends on the analysis of historical data with memory properties, and then historical information may have a significant effect on the economy. In this paper, we adopt the effects of  both investor protection and past information to develop a theory in an approximate fractional economy with a general equilibrium setting, and then try to emphasize key features of the equilibrium and extensional phenomena with empirical evidence for the economy.

Investor protection in related literature has been studied with different forms under various economy settings. Shleifer and Wolfenzon (\cite{Shleifer}, 2002) assume that the entrepreneur is caught for diverting revenue with a probability and then adopt such a probability as the measure of the legal protection of investors. Under the static two-date economy setting, they reveal that better investor protection leads to lower ownership concentration, larger external capital markets and larger firms in equilibrium, and their model additionally makes a number of general equilibrium predictions which are consistent with developed empirical evidence. Albuquerue and  Wang (\cite{Albuquerue}, 2008) introduce a parameter to measure investor protection, which captures the role of laws and law enforcement protection of minority investors. With their dynamic stochastic general equilibrium model (driven by Brownian motion), they conclude that imperfect investor protection would decrease investor welfare and reduce market value, and the model additionally predicts that weaker investor protection conduces to more incentives to overinvest, higher return volatility, larger risk premia and higher interest rate. Basak et al. (\cite{Basak}, 2019) adopt the parameter limiting the controlling shareholder's power over the firm as the protection of minority shareholders. The model contains the dynamics of output, bond, stock and investors' self-financing budget constraints. By solving both controlling shareholders' and minority shareholders' optimal wealth problems in equilibrium, they draw a conclusion that better investor protection results in lower stock holdings of controlling shareholders, higher stock returns and volatilities, and higher interest rates. For more literature related to investor protection, readers are referred to \cite{Bebchuk,Doidge,Giannetti,La2} for static models, and \cite{Dow,Gompers} for dynamic models.

In reality, empirical evidence shows that asset prices behave differently depending on the market state such as good information on the past (\cite{Perez-Quiros}, 2000), market regime (\cite{Kole}, 2016), the regime of returns and/or volatilities (\cite{Kwon}, 2016 and \cite{Wahab}, 2011) and so on. Thus, it is still of much necessity to consider how to recognize the market state in the case of investor protection. For instance, how past information affects investors can be a remarkable candidate. To our best knowledge, it has been widely examined by theoretical and empirical evidence that fractional Brownian motions (fBms) are most suggested to capture memory properties of historical financial data. The antipersistent fBm (Hurst index $H<\frac{1}{2}$) has intermediate memory, whereas the persistent one (Hurst index $H>\frac{1}{2}$) has long memory (see, e.g. \cite{Rostek}). Mandelbrot and van Ness (\cite{Mandelbrot}, 1968) suggest fBms as alternative models for assets' dynamics which allow for dependence between returns over time. Kloeden et al. (\cite{Kloeden}, 2011) adopt the multilevel Monte Carlo method to SDEs with additive fractional noise of Hurst index $H>\frac{1}{2}$. Shokrollahi and K{\i}l{\i}\c{c}man (\cite{Shokrollahi}, 2016) provide a new framework for pricing currency options in accordance with the fBm model to capture long-memory property of the spot exchange rate. Fouque and Hu consider optimal portfolio problems under fractional stochastic environment in their recent papers.  A first-order approximation of the optimal value is established under the condition that the return and volatility of the asset are functions of a stationary fast mean-reverting fractional Ornstein-Uhlenbeck process (\cite{Fouque}, 2018) as well as a stationary slowly varying fractional Ornstein-Uhlenbeck process (\cite{Fouque2}, 2018) respectively. One can refer to \cite{Biagini,Li,Lo,Mishura} for more researches on fBms. As far as we know, few literature related to investor protection contains memory properties driven by fBms, and hence, we intend to develop a new model which not only contains investor protection but also relates to memory properties.

In fact, fBms can be hardly introduced to the complex economy in \cite{Basak} which is the most related to our paper. The main reason is that stochastic calculus for fBms is mainly the integral theory based on Wick-product but not usual It\^{o} integrals (see \cite{Alos,Biagini,Mishura}). For one thing, the model based on Wick-product for fBms is much less analytical and tractable for the complex economy. Hence, we may fail to reach the equilibrium in the economy and study the effects of investor protection. For another, integrals based on Wick-product for fBms are not pathwise, which implies that ``the general state of a company does not really have a noted stock price a priori, but it brings out a number (price) when confronted with a market observer (the stock market)'' (see Page 175 in \cite{Biagini}). Consequently, we could not deal with historical realized financial data in a complex economy in a pathwise way. Therefore, we choose an approximate fBm instead of fBm in the model of our economy:
\begin{equation*}
w_t^{H}=\sqrt{2H}\int_0^t{(t-s+\varepsilon)^{H-\frac{1}{2}}}d{w}_s
\end{equation*}
for standard Brownian motion $w$, the Hurst index $H\in(0,1)$ and $\varepsilon>0$, where the parameter $\sqrt{2H}$ is introduced to make the following $\mathbb{B}_t^H$  has same distribution with the associated fBm for each $t>0$. There are several reasons why we introduce the process $w_t^{H}$.  First, the process $w_t^{H}$ is indeed an approximation of fBm and then inherits the memory properties of fBm. It can be proved (see \cite{Thao,Thao02})  that when $\varepsilon$ tends to $0$, the process $w_t^H$ converges in $L^2$ to the process
\begin{equation*}
\mathbb{B}_t^H=\sqrt{2H}\int_0^t{(t-s)^{H-\frac{1}{2}}}dw_s
\end{equation*}
which is the main memory part of fBm and can be viewed as a good approximation of the fBm for large times (see Theorem 17 in \cite{Picard}). Second, the process $w_t^{H}$ is a classic It\^{o} integral with analytical expression (see (\ref{fBm})), and hence in such an economy we are able to study the effects of investor protection and apply the historical realized financial data in the economy in a pathwise way. Third, it is known that the process $w_t^{H}$ degenerates to Brownian motion for $H=\frac{1}{2}$ and hence such a fractional economy contains the usual economy. Moreover, the basic dynamic of output used in our paper can be viewed as a modified geometric Brownian motion whose increments are approximately affected by a linear combination of its past increments (see Remark \ref{remark_intro}).

We note that the process $w_t^{H}$ has been used in theoretical study as well as practical applications. Dung (\cite{Dung}, 2011) applies the approximate fBm for an approximate fractional Black-Scholes model and then the European option pricing formula is found. Because of the semimartingale property, the effect of the approximate fBm on such pricing system is reflected in the Hurst index and is independent of the historical realized stock price. Xu and Li (\cite{Xu}, 2015) use the approximate fBm to study a class of doubly perturbed neutral stochastic functional equations driven by fBm. Yue and Huang (\cite{Yue}, 2018) introduce the approximate fBm to a general Wishart process and establish $\varepsilon$-fractional Wishart process. The new process possesses the approximate memory property inheriting from fBm. Besides, applying the new process to the financial volatility theory, it extends both the one-dimension CIR process (\cite{Cox}) and the high-dimension Wishart process (see \cite{Bru},\cite{Fonseca}). However, it is also difficult to price derivatives by the historical realized stock price. Therefore, it is of much meaning to make use of the historical realized data to obtain the equilibrium and then study the effects of investor protection in such a fractional economy.

In this paper, we consider a dynamic asset pricing model in an approximate fractional economy to address empirical regularities related to both investor protection for minority shareholders and memory properties of financial data. Our model features not only a controlling shareholder who diverts a fraction of output but also good (or bad) memory which investors obtain from the historical realized output data of the economy in a pathwise way. In line with \cite{Basak}, we find that poorer investor protection leads to higher stock holdings of controlling holders, lower gross stock returns, lower interest rates, and lower modified stock volatilities if the ownership concentration is sufficiently high. More importantly, we reveal the economic behavior of the approximate fractional economy in the following aspects: (i) describe good (or bad) memory of investors for the past information and then establish an approximation scheme with pathwise convergence for the memory; (ii) conclude that good (bad) memory would increase/decrease stock holdings of controlling holders, modified stock returns and volatilities, gross stock returns, and interest rates; (iii) show that  good (bad) memory would strengthen (weaken) investor protection for minority shareholder when the ownership concentration is sufficiently high, while good (bad) memory would inversely weaken (strengthen) investor protection for minority shareholder when the ownership concentration is sufficiently low.

The rest of the paper is organized as follows. Section \ref{section2} sketches the financial model with investor protection and approximate fBms. In Section \ref{section3} we derive the dynamics of asset prices in equilibrium for the approximate fractional economy.  Before we summarizing this paper in Section \ref{section6}, the approximation scheme with pathwise convergence is established for past information by using the historical realized data of the output and some numerical results are given to address empirical regularities related to both investor protection for minority shareholders and memory properties of financial data in Section \ref{section5}.

\section{Models driven by approximate fBms}\label{section2}\noindent
\setcounter{equation}{0}
Let $(\mathbf{\Omega},\mathcal{F},(\mathcal{F}_t)_{t\geq 0}, \mathbb{P})$ be a filtered probability space satisfying the usual conditions throughout where $\mathcal{F}_t$ is generated by a standard Brownian motion $w$.

Define an approximate fBm $w_t^H$ with parameter $\varepsilon>0$ and $H\in(0,1)$ as
\[
w_t^H=\sqrt{2H}\int_0^t(t-s+\varepsilon)^h dw_s,
\]
where $t\geq 0,\ h=H-\frac{1}{2}$ and we omit the parameter $\varepsilon$ in $w_t^H$. By It\^{o}'s Lemma, we can express $w_t^H$ in an equivalent way (or one may refer to \cite{Thao} for the case $\frac{1}{2}<H<1$ and \cite{Thao02} for the case $0<H<\frac{1}{2}$):
\begin{equation}\label{fBm}
dw_t^H=\Lambda_tdt+\sqrt{2H}\varepsilon^h dw_t
\end{equation}
where $\Lambda_t=\sqrt{2H}h\int_0^t(t-s+\varepsilon)^{h-1}dw_s$.

In above approximate fractional probability space, we construct a financial market and consider its economy with investor protection (for more details, we refer to \cite{Basak}). We focus on one representative firm standing for amounts of identical firms. The exogenous stream of the output associated with the firm is captured by
\begin{equation}\label{output}
d\widehat{D}_t=\widehat{D}_t\left[\mu_Ddt+\sigma_Ddw_t^H\right],
\end{equation}
where constants $\mu_D$ and $\sigma_D>0$ represent the mean-growth rate (or return) and volatility of the output, respectively.

\begin{remark}\label{remark_intro}{\it
Here we explain why we introduce the process $w_t^{H}$. In a financial market, the traditional dynamic for the output of a product or the price of a stock (where we use the same mark $\hat{D}_t$ for simplicity) is the geometric Brownian motion (\cite{Basak})
\[
d\ln\hat{D}_t=\left(\mu_D-\frac{\sigma_D^2}{2}\right)dt+\sigma_Ddw_t.
\]
However, empirical evidence (see, for example, Section 11.6 in \cite{Franke}) shows that it may contradict the unit root tests, which implies we should not ignore correlation between two increments of the process $\ln\hat{D}_t$. To be specific, if we consider the increment of $\ln\hat{D}_t$ at any fixed time $t=t_n>0$ for the time nodes in Section \ref{subsection1}, then besides the future uncertain term $\Delta w_{t_{n+1}}=w_{t_{n+1}}-w_{t_{n}}$ contained in the geometric Brownian motion, it may be necessary to take all pathwise uncertain terms $\Delta w_{t_{k+1}}(k=0,1,\cdots n-1)$ into account. With this aim, the simplest method is to regard the increment approximately as a linear combination of aforementioned uncertain terms:
\[
\ln\hat{D}_{t_{n+1}}-\ln\hat{D}_{t_{n}}\approx a_{n,0}+\sum_{k=0}^na_{n,k+1}\Delta w_{t_{k+1}}.
\]
Specially, putting
\[
\left\{
\begin{aligned}
a_{n,0}&=(\mu_D-H\varepsilon^{2h}\sigma_D^2)\Delta t;\\
a_{n,k+1}&=\sqrt{2H}h(t_{n}-t_k+\varepsilon)^{h-1}\sigma_D\Delta t,\quad k=0,1,\cdots n-1;\\
a_{n,n+1}&=\sqrt{2H}\varepsilon^{h}\sigma_D,
\end{aligned}
\right.
\]
then Lemma \ref{lastlemma} shows that the process (\ref{output}) introduced in our paper can be viewed a good candidate for the modification of the geometric Brownian motion. Moreover, the absolute value of each parameter $a_{n,k+1}(k=0,1,\cdots,n)$ represents the effect of the uncertain term $\Delta w_{t_{k+1}}$ on the increment. Hence, the uncertain information closer to the time $t_{n+1}$ plays a greater role in the information of the increment, which is consistent with our common sense.}
\end{remark}

\begin{remark}\label{remark0}
{\it
In the approximate fractional economy, new characteristics brought by approximate fBms could be explained in two aspects--$\Lambda_t$ and $H$. It is seen that  $\Lambda_t$ is the essential difference between the approximate fBm $w_t^H$ (in the case of $H\neq \frac{1}{2}$) and the Brownian motion $w_t^{\frac{1}{2}}=w_t$. In the case $H=\frac{1}{2}$, expressing (\ref{output}) equivalently as
\begin{equation*}
d\widehat{D}_t=\widehat{D}_t\left[\mu_Ddt+\sigma_Ddw_t\right],
\end{equation*}
it is well-known that $\widehat{D}_t$ is a Markov process whose increments are independent of past information.  In the case $H\neq\frac{1}{2}$, rewriting (\ref{output}) equivalently as
\begin{equation}\label{eqH}
d\widehat{D}_t=\widehat{D}_t\left[(\mu_D+\sigma_D\Lambda_t)dt+\sqrt{2H}\varepsilon^h\sigma_Ddw_t\right],
\end{equation}
it is telling that the return $\mu_D$  and volatility $\sigma_D$ are modified as $\mu_D+\sigma_D\Lambda_t$ and $\sqrt{2H}\varepsilon^h\sigma_D$ respectively.\\
(1) As the modified return $\mu_D+\sigma_D\Lambda_t$ depends on the past information $\Lambda_t$ which is the past cumulate information of $w_t$, some functions of $\Lambda_t$ could be regarded as the measure of the approximate fBm's memory in $[0,t]$. The changing return can be viewed as a continuous version of the regime switching return in \cite{Wahab}.  As a result, in the case of $H=\frac{1}{2}$ the increment of $\widehat{D}_t$ depends only on the information of $\widehat{D}_t$, while in the case of $H\neq\frac{1}{2}$ the increment of $\widehat{D}_t$ is correlated to the information of $\{\widehat{D}_s\}_{0\leq s\leq t}$. However, we shall see in Section \ref{section5} that $\Lambda_t$ can be known in theory and be estimated in practice by using the information of $\{\widehat{D}_s\}_{0\leq s\leq t}$, which is crucial in our fractional models.\\
(2) The Hurst index $H$ also changes the modified volatility $\sqrt{2H}\varepsilon^h\sigma_D$ of the output by (\ref{eqH}). It can be verified that that $g(H)=\sqrt{2H}\varepsilon^h$ is an increasing function of $H$ in the interval $(0,-\frac{1}{2\ln\varepsilon})\cap (0,1)$ and a decreasing function of $H$ in the interval $(-\frac{1}{2\ln\varepsilon},1)\cap (0,1)$. Consequently, higher $H$ in general leads to higher volatility of the output in the former case, while higher $H$ in general leads to lower volatility of the output in the later case. If $\varepsilon>0$ is extremely small, then the output will be extremely volatile for $H<\frac{1}{2}$ and be quite stable on the contrary for $H>\frac{1}{2}$, which barely exists in our true financial market. Hence the parameter $\varepsilon>0$ in our paper are not supposed to be extremely small, and we choose it from the interval $(0.01,1).$
}
\end{remark}

The shareholders in the corporation are sorted into two types---a controlling shareholders $C$ and a representative minority shareholders $M$. A bond $B_t$ and a stock $S_t$ traded by the shareholders are respectively normalized to one unit and follow stochastic differential equations:
\begin{align}
dB_t&=B_tr_tdt,\quad B_0=1;\label{bond}\\
dS_t&=S_t\left[\mu_tdt+\sigma_tdw_t^H\right],\quad S_0>0,\label{stock}
\end{align}
where $r_t$ is the interest rate of the bond; $\mu_t$ is the mean-return (or return) of the stock; $\sigma_t$ is the volatility of the stock, and all processes $r_t,\mu_t,\sigma_t$ are $\mathcal{F}_t$-adapt and are endogenously determined in equilibrium.

Furthermore, we assume that the wealth of shareholders $C$ and $M$ at time $t$, denoted by $W_{Ct}$ and $W_{Mt}$, are made up of portfolios of $B_t$ and $S_t$ and expressed as following equations:
\begin{align}
W_{it}=&b_{it}B_t+n_{it}S_t,\label{budgetB}\\
dW_{it}=&\left(W_{it}r_t+n_{it}(S_t(\mu_t-r_t)+(1-x_t)D_t)-c_{it}+l_i\widehat{D}_t\right)dt\nonumber\\
&+\left(\mathbf{1}_{\{i=C\}}(x_tD_t-f(x_t,D_t))+\mathbf{1}_{\{i=M\}}f(x_t,D_t)\right)dt+n_{it}S_t\sigma_tdw_t^H,\label{budgetW}
\end{align}
where $i\in\{C,M\}$;
$b_{it}$ is the number of units of bounds in the shareholder's portfolio;
$n_{it}$ is the number of shares of stocks in the shareholder's portfolio;
$x_t$ is the fraction of diverted output which satisfies the investor protection constraint with a parameter $p\in[0,1]$ (interpreted as the protection of minority shareholders)
    \begin{equation}\label{protectC}
    x_t\leq (1-p)n_{Ct};
    \end{equation}
$c_{it}$ is the consumption of the shareholder $i$;
$l_i$ is the fraction of the output paid to the shareholder $i$ as labor incomes;
$D_t=(1-l_C-l_M)\widehat{D}_t$ represents the net output;
$f(x,D)=\frac{kx^2D}{2}$ is assumed to be a pecuniary cost from diverting output with a constant $k$ and the parameter $k$ captures the magnitude of the cost.

We now turn to constructing shareholders' optimization problems with investor protection and past information. For shareholders' objective functions (utility functions) in our approximate fractional economy,  not only preferences over their consumption and wealth but also affections of the past information should be taken into account. We assume shareholders' preferences  over current consumption $c$ and wealth $W$ are determined by some similar myopic preferences used in \cite{Basak}. Moreover, inspired by \cite{Kraft}, we introduce some Cobb-Douglas functions to measure how past information affects shareholders' preferences. Summarizing, shareholders' objective functions are assumed to be
\begin{equation}\label{wealth}
V_i(c_{it},W_{it},W_{i,t+dt},\Lambda_t)=\rho u_i(c_{it},\Lambda_t)dt+(1-\rho dt)\mathbb{E}_t\left[U_i(W_{i,t+dt},\Lambda_t)\right],\quad i\in\{C,M\},
\end{equation}
with
\begin{align*}
u_i(c,\Lambda)&=\frac{1}{1-\gamma_i}\left[c^{\alpha_i}\phi(\Lambda)^{1-\alpha_i}\right]^{1-\gamma_i},\\
U_i(W,\Lambda)&=\frac{1}{1-\gamma_i}\left[W^{\alpha_i}\psi(\Lambda)^{1-\alpha_i}\right]^{1-\gamma_i},
\end{align*}
where $\rho>0$ is a time-preference parameter; $\gamma_i>1$ are parameters in the power utility functions and following \cite{Basak} we assume that the controlling shareholder is less risk averse than the minority shareholder, i.e. $\gamma_M\geq\gamma_C$; $\phi$ and $\psi$ are some adjustment functions of the measure of past information $\Lambda_t$ associated with the consumption and the wealth respectively, and we further assume that $\phi$ and $\psi$ are increasing and twice continuously differential with $\phi(0)=\psi(0)=1$; $0<\alpha_i\leq 1$ are trade-off parameters between the past information and the current information.

Both functions $\phi,\psi$ and parameters $\alpha_i$ reflect how investors are affected by the past information.  The adjusted functions $\phi$ and $\psi$ describe the measure of the past information estimated by general investors. Lower trade-off parameter $\alpha_C$ ($\alpha_M$, respectively) depicts higher sensitivity of the controlling (minority) shareholder to the past information, which implies the past information has greater effects on the controlling (minority) shareholder. On the contrary, higher trade-off parameter $\alpha_C$ ($\alpha_M$, respectively) depict lower sensitivity of the controlling (minority) shareholder to the past information, which implies the past information has fewer effects on the controlling (minority) shareholder.  If $\alpha_M\leq\alpha_C$ ($\alpha_M\geq\alpha_C$), then the minority (controlling) shareholder is more sensitive to the past information and he would attach more importance to the past information obtained at the present time.

Here we give an example of functions $\phi,\psi$. It is seen by (\ref{eqH}) that
\[
\hat{D}_{t+dt}=\hat{D}_{t}\exp\left\{(\mu_D-H\varepsilon^{2h}\sigma_D^2)dt+\sqrt{2H}\varepsilon^h\sigma_Ddw_t\right\}
\times e^{\sigma_Ddt\cdot \Lambda_t},
\]
which implies, for the constant time increment $dt$ at time $t$, the increment of output related to the past information $\Lambda_t$ is $e^{\sigma_Ddt\cdot \Lambda_t}$. Hence the adjustment function $\psi$ can be chosen in the form of
\begin{equation}\label{psi}
\psi(\Lambda)=e^{\beta_1\Lambda}
\end{equation}
for some constant $\beta_1>0$. As the consumption contains the necessary part (for example, the infrastructure development of the corporation) which is mainly associated with the present wealth wether or not the historical market behaves badly, the past information should have fewer effects on the consumption and we can simply assume that
\begin{equation}\label{phi}
\phi(\Lambda)=e^{\beta_2\Lambda}
\end{equation}
with some constant $0\leq\beta_2<\beta_1$. In this paper, we shall derive equilibrium results for general functions $\phi$ and $\psi$, and then give results for (\ref{psi}) and (\ref{phi}) as a special example.

The controlling shareholder $C$ maximizes $V_C$ through strategies of the investment $n_{Ct}$, the consumption $c_{Ct}$ and the diverting fraction $x_t$:
\begin{equation}\label{problemC}
\max\limits_{n_{Ct},c_{Ct},x_t}V_C(c_{Ct},W_{Ct},W_{C,t+dt},\Lambda_t),
\end{equation}
where $V_C$ is given by (\ref{wealth}) in the case $i=C$, subject to constraints (\ref{budgetW}) and (\ref{protectC}), and maximum share constraint $n_{Ct}\leq 1$. Similarly, the minority shareholder $M$ maximizes $V_M$ through strategies of the investment $n_{Mt}$ and the consumption $c_{Mt}$:
\begin{equation}\label{problemM}
\max\limits_{n_{Mt},c_{Mt}}V_M(c_{Mt},W_{Mt},W_{M,t+dt},\Lambda_t),
\end{equation}
where $V_M$ is given by (\ref{wealth}) in the case $i=M$, subject to the constraint (\ref{budgetW}) and maximum share constraint $n_{Mt}\leq 1$.

\begin{remark}\label{remark1}{\it
(1) Similar to Remark \ref{remark0}, characteristics brought by approximate fBms could be discussed for the stock $S$ and the wealth $W_i$. Furthermore, if $H=\frac{1}{2}$, then the parameter $\varepsilon$ disappears in $w_t^H$ and our models (\ref{output})-(\ref{budgetW}) degenerate to ones studied in \cite{Basak}.\\
(2) In the case of $H=\frac{1}{2}$, the past information $\Lambda\equiv0$ and we can additionally request the hypothesis $\alpha_i=1$ in which case shareholders' objective functions are equivalent to those used in \cite{Basak} (except a constant).
}
\end{remark}

\section{Equilibrium with investor protection and past information}\label{section3}\noindent
\setcounter{equation}{0}
In this section, we solve shareholders' optimization problems (\ref{problemC}) and (\ref{problemM}) in the models driven by approximate fBms and derive investors' optimal strategies and asset price dynamics in equilibrium. To this end, we start from a partial equilibrium setting where processes $r_t,\mu_t$ and $\sigma_t$ are regarded as given processes. Then, making use of market clearing conditions, $r_t,\mu_t$ and $\sigma_t$ are truly obtained in equilibrium.

Using (\ref{fBm}) and It\^{o}'s Lemma we first decompose problems (\ref{problemC}) and (\ref{problemM}) into some equivalent but more precise problems. To be brief, we denote for $i\in\{C,M\}$,
\begin{align*}
\delta_i=&1-\alpha_i(1-\gamma_i)> 0,\\
\theta_i=&1-\frac{\gamma_i}{\delta_i}\leq 0,\\
\varphi=&\frac{\phi}{\psi}.
\end{align*}
Noting that $\mathbb{E}_t[\sim dw_t]=0$ with some $\mathcal{F}_t$-measurable expressions $\sim$ and for $i\in\{C,M\}$,
\begin{equation*}
\mathbb{E}_t[U_i(W_{i,t+dt},\Lambda_t)]=U_i(W_{it},\Lambda_t)+\frac{1}{1-\gamma_i}
\psi(\Lambda_t)^{(1-\alpha_i)(1-\gamma_i)}\mathbb{E}_t[d(W_{it})^{\alpha_i(1-\gamma_i)}],
\end{equation*}
it is seen that
\begin{align*}
&V_i(c_{it},W_{it},W_{i,t+dt},\Lambda_t)\\
=&\left[\frac{\rho}{1-\gamma_i}\phi(\Lambda_t)^{(1-\alpha_i)(1-\gamma_i)}c_{it}^{\alpha_i(1-\gamma_i)}
-\alpha_i\psi(\Lambda_t)^{(1-\alpha_i)(1-\gamma_i)}W_{it}^{-\delta_i}c_{it}\right]dt\\
&+(1-\rho dt)U_i(W_{it},\Lambda_t)
+\alpha_i\psi(\Lambda_t)^{(1-\alpha_i)(1-\gamma_i)}W_{it}^{-\delta_i}l_i\widehat{D}_tdt\\
&+\alpha_i\psi(\Lambda_t)^{(1-\alpha_i)(1-\gamma_i)}W_{it}^{\alpha_i(1-\gamma_i)}
\left\{r_t+\frac{n_{it}S_t}{W_{it}}\left(\mu_t-r_t+(1-x_t)\frac{D_t}{S_t}+\sigma_t\Lambda_t\right)\right.\\
&\left.+\mathbf{1}_{\{i=C\}}\left(x_t\frac{D_t}{W_{it}}-\frac{kx_t^2}{2}\frac{D_t}{W_{it}}\right)+\mathbf{1}_{\{i=M\}}\frac{kx_t^2}{2}\frac{D_t}{W_{it}}
-H\varepsilon^{2h}\delta_i\left(\frac{S_t\sigma_t}{W_{it}}\right)^2n_{it}^2\right\}dt.
\end{align*}
Hence, the problem (\ref{problemC}) for shareholder $C$ is equivalent to
\begin{equation}\label{problemCc}
\max\limits_{c_{Ct}}\left\{\frac{\rho}{1-\gamma_C}\phi(\Lambda_t)^{(1-\alpha_C)(1-\gamma_C)}
c_{Ct}^{\alpha_C(1-\gamma_C)}
-\alpha_C\psi(\Lambda_t)^{(1-\alpha_C)(1-\gamma_C)}W_{Ct}^{-\delta_C}c_{Ct}\right\};
\end{equation}
\begin{equation}\label{problemCn}
\max\limits_{n_{Ct},x_t}J_C(n_{Ct};x_t),
\end{equation}
where
\begin{align}
J_C(n_{Ct};x_t)=&\frac{n_{Ct}S_t}{W_{Ct}}\left(\mu_t-r_t+(1-x_t)\frac{D_t}{S_t}+\sigma_t\Lambda_t\right)\nonumber\\
&+x_t\frac{D_t}{W_{Ct}}-\frac{kx_t^2}{2}\frac{D_t}{W_{Ct}}
-H\varepsilon^{2h}\delta_C\left(\frac{S_t\sigma_t}{W_{Ct}}\right)^2n_{Ct}^2.\label{JCC}
\end{align}
And the problem (\ref{problemM}) for shareholder $M$ is equivalent to
\begin{equation}\label{problemMc}
\max\limits_{c_{Mt}}\left\{\frac{\rho}{1-\gamma_M}\phi(\Lambda_t)^{(1-\alpha_M)(1-\gamma_M)}
c_{Mt}^{\alpha_M(1-\gamma_M)}
-\alpha_M\psi(\Lambda_t)^{(1-\alpha_M)(1-\gamma_M)}W_{Mt}^{-\delta_M}c_{Mt}\right\};
\end{equation}
\begin{equation}\label{problemMn}
\max\limits_{n_{Mt}}J_M(n_{Mt}),
\end{equation}
where
\begin{equation*}
J_M(n_{Mt})=\frac{n_{Mt}S_t}{W_{Mt}}\left(\mu_t-r_t+(1-x_t)\frac{D_t}{S_t}+\sigma_t\Lambda_t\right)\\
-H\varepsilon^{2h}\delta_M\left(\frac{S_t\sigma_t}{W_{Mt}}\right)^2n_{Mt}^2.
\end{equation*}

Under a partial equilibrium setting, we first obtain investors' optimal strategies which is the following theorem.
\begin{theorem}\label{th1}{\it
In the model driven by the approximate fBm, the optimal consumptions $c_{it}^*$, the fraction of diverted output $x_t^*$ and the optimal stock holding $n_{it}^*$ for $i\in\{C,M\}$ are given as follows:
\begin{align}
c_{it}^*=&\rho^{\frac{1}{\delta_i}}\varphi(\Lambda_t)^{\theta_i}W_{it},\quad i\in\{C,M\};\label{ps_c}\\
x_t^*=&x_t^*(n_{Ct}^*)=\min\left\{\frac{1-n_{Ct}^*}{k},(1-p)n_{Ct}^*\right\};\label{xx}\\
n_{Ct}^*=&n_{Ct,i}^*,\quad \text{if}\, J_C^*=J_C(n_{Ct,i}^*;x_t^*)(\text{Region}\ i),\, i\in\{1,2,3,4\};\label{ps_nC}\\
n_{Mt}^*=&\frac{\mu_t-r_t+(1-x_t^*)\frac{D_t}{S_t}+\sigma_t\Lambda_t}{2H\varepsilon^{2h}\delta_M\frac{S_t}{W_{Mt}}\sigma_t^2},\label{ps_nM}
\end{align}
where
\begin{align}
n_{Ct,1}^*=&\frac{\mu_t-r_t+(2-p)\frac{D_t}{S_t}+\sigma_t\Lambda_t}{2H\varepsilon^{2h}\delta_C\frac{S_t}{W_{Ct}}\sigma_t^2
+(2(1-p)+k(1-p)^2)\frac{D_t}{S_t}};\label{nC1}\\
n_{Ct,2}^*=&\frac{\mu_t-r_t+\left(1-\frac{1}{k}\right)\frac{D_t}{S_t}+\sigma_t\Lambda_t}{2H\varepsilon^{2h}\delta_C\frac{S_t}{W_{Ct}}\sigma_t^2
-\frac{1}{k}\frac{D_t}{S_t}};\label{nC2}\\
n_{Ct,3}^*=&\frac{1}{1+(1-p)k};\label{nC3}\\
n_{Ct,4}^*=&1;\label{nC4}
\end{align}
and
\begin{equation}\label{J_C}
J_C^*=\max\left\{J_C(n_{Ct,1}^*;x_t^*),\, J_C(n_{Ct,2}^*;x_t^*),
J_C(n_{Ct,3}^*;x_t^*),\, J_C(n_{Ct,4}^*;x_t^*)\right\}.
\end{equation}
}
\end{theorem}

\begin{remark}\label{memory}{\it
Compared with Proposition 1 in \cite{Basak}, the past information $\Lambda_t$ plays an important and different role under the partial equilibrium setting. It seems that both the controlling shareholder and the minority shareholder can benefit from the past information $\Lambda_t$. If $\Lambda_t=0$, then by (\ref{ps_nM}) and (\ref{nC1})-(\ref{nC2}), the past information $\Lambda_t$ contributes nothing to investors' shares, which implies that shareholders have ``no memory" for the past information $\Lambda_t$.
If $\Lambda_t>0$, then by (\ref{ps_nM}) and (\ref{nC1})-(\ref{nC2}), shareholders would acquire more shares through the past information $\Lambda_t$, which implies that the past information $\Lambda_t>0$ is a kind of ``good memory" for shareholders. Similarly, if $\Lambda_t<0$, then by (\ref{ps_nM}) and (\ref{nC1})-(\ref{nC2}), shareholders would acquire fewer shares through the past information $\Lambda_t$, which implies that the past information $\Lambda_t<0$ is a kind of ``bad memory" for shareholders. Moreover, if there exists a variety of stocks in the market, it is worthwhile to note that we can find counterparts for above financial phenomena: in a period when ``good memory" always stays, the market turns out to be a bull market, while in a period when ``bad memory" always lasts, the market turns out to be a bear market.
}
\end{remark}

In the following part of this section, we obtain equilibrium dynamics $\mu_t,\sigma_t,r_t,n_{it}^*,b^*_{it}$ and $c_{it}^*$ ($i\in\{C,M\}$) in (\ref{bond})-(\ref{budgetW}). We define a standard equilibrium as a set of $\mu_t,\sigma_t,r_t,n_{it}^*,b_{it}^*$ and $c^*_{it}$ ($i\in\{C,M\}$) satisfying the following market clearing conditions (\cite{Basak})
\begin{align}
n_{Ct}^*+n_{Mt}^*&=1,\label{ec-1}\\
b_{Ct}^*+b_{Mt}^*&=0,\label{ec-2}\\
c_{Ct}^*+c_{Mt}^*&=\widehat{D}_t.\label{ec-3}
\end{align}
To this end, we introduce a new process
\[\lambda_t=\sqrt{2H}h(h-1)\int_0^t(t-s+\varepsilon)^{h-2}dw_s.\]
It is obvious that $\lambda_t$ is an $\mathcal{F}_t$-measurable process for which we will also construct an approximation scheme with pathwise convergence in Section \ref{section5}.  Similar to $w_t^H$, we have
\[d\Lambda_t=\lambda_tdt+\sqrt{2H}h\varepsilon^{h-1}dw_t.\]

Denote the minority shareholder's share in the aggregate consumption by
$$y_t=\frac{c_{Mt}^*}{\widehat{D}_t}.$$
The process $y_t$ is an important state process to derive equilibrium dynamics and following \cite{Basak}, we assume and then verify that $y_t$ satisfies a form of
\begin{equation}
dy_t=\mu_{yt}dt+\sigma_{yt}dw_t^H,
\end{equation}
where $\mu_{yt}$ and $\sigma_{yt}$ are functions of processes $y_t$, $\Lambda_t$ and $\lambda_t$ to be determined in equilibrium. Moreover, we also define a similar Sharpe ratio $\kappa_t$ as
\begin{equation}\label{kt}
\kappa_t=\frac{\mu_t-r_t+(1-x_t^*)\frac{D_t}{S_t}+\sigma_t\Lambda_t}{\sqrt{2H}\varepsilon^h \sigma_t}.
\end{equation}

\begin{theorem}\label{th2}{\it
In the model driven by the approximate fBm and under the equilibrium conditions (\ref{ec-1})-(\ref{ec-3}) with $p\in[0,1]$ and $y_t\in(0,1)$, the shareholders' optimal consumptions $c_{it}^*$ ($i\in\{C,M\}$) are given by (\ref{ps_c}), and the interest rate $r_t$, the stock mean-return $\mu_t$, the stock volatility $\sigma_t$, the Sharpe ratio $\kappa_t$, the parameters $\mu_{yt}$ and $\sigma_{yt}$ are given by:
\begin{align}
r_t=&\mu_D+\sigma_D\Lambda_t-l_C\rho^{\frac{1}{\delta_C}}\varphi(\Lambda_t)^{\theta_C}-l_M\rho^{\frac{1}{\delta_M}}\varphi(\Lambda_t)^{\theta_M}\nonumber\\
&-n_{Ct}^*\sigma_t\left[\sqrt{2H}\varepsilon^h\kappa_t+2Hh\varepsilon^{2h-1}\theta_C\frac{\varphi'(\Lambda_t)}{\varphi(\Lambda_t)}\right]
\left[1-y_t+y_t\rho^{\frac{1}{\delta_C}-\frac{1}{\delta_M}}\varphi(\Lambda_t)^{\theta_C-\theta_M}\right]\nonumber\\
&-(1-n_{Ct}^*)\sigma_t\left[\sqrt{2H}\varepsilon^h\kappa_t+2Hh\varepsilon^{2h-1}\theta_M\frac{\varphi'(\Lambda_t)}{\varphi(\Lambda_t)}\right]
\left[y_t+(1-y_t)\rho^{\frac{1}{\delta_M}-\frac{1}{\delta_C}}\varphi(\Lambda_t)^{\theta_M-\theta_C}\right]\nonumber\\
&+(1-y_t)\left[\rho^{\frac{1}{\delta_C}}\varphi(\Lambda_t)^{\theta_C}-\theta_C\lambda_t\frac{\varphi'(\Lambda_t)}{\varphi(\Lambda_t)}
-Hh^2\varepsilon^{2h-2}\theta_C\left((\theta_C-1)\left(\frac{\varphi'(\Lambda_t)}{\varphi(\Lambda_t)}\right)^2
+\frac{\varphi''(\Lambda_t)}{\varphi(\Lambda_t)}\right)\right]\nonumber\\
&+y_t\left[\rho^{\frac{1}{\delta_M}}\varphi(\Lambda_t)^{\theta_M}-\theta_M\lambda_t\frac{\varphi'(\Lambda_t)}{\varphi(\Lambda_t)}
-Hh^2\varepsilon^{2h-2}\theta_M\left((\theta_M-1)\left(\frac{\varphi'(\Lambda_t)}{\varphi(\Lambda_t)}\right)^2
+\frac{\varphi''(\Lambda_t)}{\varphi(\Lambda_t)}\right)\right]\nonumber\\
&-\left(x_t^*-\frac{kx_t^{*2}}{2}\right)(1-l_C-l_M)\rho^{\frac{1}{\delta_C}}\varphi(\Lambda_t)^{\theta_C}
-\frac{kx_t^{*2}}{2}(1-l_C-l_M)\rho^{\frac{1}{\delta_M}}\varphi(\Lambda_t)^{\theta_M},\label{rt}\\
\mu_t=&r_t-\sigma_t\Lambda_t+\sqrt{2H}\varepsilon^h\kappa_t\sigma_t
-\frac{(1-x_t^*)(1-l_C-l_M)}{(1-y_t)\rho^{-\frac{1}{\delta_C}}\varphi(\Lambda_t)^{-\theta_C}+y_t\rho^{-\frac{1}{\delta_M}}\varphi(\Lambda_t)^{-\theta_M}},\label{miut}\\
\sigma_t=&\frac{\sigma_D-\frac{h}{\varepsilon}
\frac{\varphi'(\Lambda_t)}{\varphi(\Lambda_t)}\left[\theta_C(1-y_t)+\theta_My_t\right]}
{\left[n_{Ct}^*\rho^{\frac{1}{\delta_C}}\varphi(\Lambda_t)^{\theta_C}+(1-n_{Ct}^*)\rho^{\frac{1}{\delta_M}}\varphi(\Lambda_t)^{\theta_M}\right]
\left[(1-y_t)\rho^{-\frac{1}{\delta_C}}\varphi(\Lambda_t)^{-\theta_C}+y_t\rho^{-\frac{1}{\delta_M}}\varphi(\Lambda_t)^{-\theta_M}\right]},\label{sigmat}\\
\kappa_t=&\frac{\sqrt{2H}\varepsilon^h\delta_M\rho^{\frac{1}{\delta_M}}\varphi(\Lambda_t)^{\theta_M}(1-n_{Ct}^*)
\left[\sigma_D-\frac{h}{\varepsilon}
\frac{\varphi'(\Lambda_t)}{\varphi(\Lambda_t)}\left(\theta_C(1-y_t)+\theta_My_t\right)\right]}
{y_t\left[n_{Ct}^*\rho^{\frac{1}{\delta_C}}\varphi(\Lambda_t)^{\theta_C}+(1-n_{Ct}^*)\rho^{\frac{1}{\delta_M}}\varphi(\Lambda_t)^{\theta_M}\right]},\label{ktt}\\
\sigma_{yt}=&y_t\left(\frac{\kappa_t}{\sqrt{2H}\varepsilon^h\delta_M}+\theta_M\frac{h}{\varepsilon}
\frac{\varphi'(\Lambda_t)}{\varphi(\Lambda_t)}-\sigma_D\right),\label{muy}\\
\mu_{yt}=&-\sigma_{yt}(\Lambda_t+2H\varepsilon^{2h}\sigma_D)+l_M\rho^{\frac{1}{\delta_M}}\varphi(\Lambda_t)^{\theta_M}
+\frac{kx_t^{*2}}{2}(1-l_C-l_M)\rho^{\frac{1}{\delta_M}}\varphi(\Lambda_t)^{\theta_M}\nonumber\\
&+y_t\bigg[r_t-\mu_D-\sigma_D\Lambda_t+\frac{\kappa_t^2}{\delta_M}+\sqrt{2H}h\varepsilon^{h-1}\theta_M
\frac{\varphi'(\Lambda_t)}{\varphi(\Lambda_t)}\frac{\kappa_t}{\delta_M}-\rho^{\frac{1}{\delta_M}}\varphi(\Lambda_t)^{\theta_M}\nonumber\\
&+\theta_M\lambda_t\frac{\varphi'(\Lambda_t)}{\varphi(\Lambda_t)}
+Hh^2\varepsilon^{2h-2}\theta_M\left((\theta_M-1)\left(\frac{\varphi'(\Lambda_t)}{\varphi(\Lambda_t)}\right)^2
+\frac{\varphi''(\Lambda_t)}{\varphi(\Lambda_t)}\right)\bigg].\label{sy}
\end{align}
The shareholders' optimal bond holdings $b_{it}^*$ and stock holdings $n_{it}^*$ ($i\in\{C,M\}$) could be expressed as follows:
\begin{align}
b_{Ct}^*=&\exp\left(-\int_0^tr_sds\right)(W_{Ct}-n_{Ct}^*S_t),\label{bcc}\\
b_{Mt}^*=&-b_{Ct}^*,\\
n_{Ct}^*=&n_{Ct,j}^*,\quad \text{if}\, J_C^*=J_C(n_{Ct,j}^*;x_t^*)(\text{Region}\, j),\, j\in\{1,2,3,4\},\nonumber\\
n_{Mt}^*=&1-n_{Ct}^*,\label{nmt}
\end{align}
where $n_{Ct,1}^*$ satisfies the following fixed point problem
\begin{align}
n_{Ct,1}^*=&n_{Ct,2}^*+n_{Ct,2}^*\left(1-\frac{n_{Ct,1}^*}{n_{Ct,3}^*}\right)\frac{(1-p)(1-l_C-l_M)}{2H\varepsilon^{2h}}
\frac{y_t}{{\delta_M}\rho^{\frac{1}{\delta_M}}\varphi(\Lambda_t)^{\theta_M}}\nonumber\\
&\times\frac{\left[n_{Ct,1}^*\rho^{\frac{1}{\delta_C}}\varphi(\Lambda_t)^{\theta_C}+(1-n_{Ct,1}^*)\rho^{\frac{1}{\delta_M}}\varphi(\Lambda_t)^{\theta_M}\right]^2
}{\left[\sigma_D-\frac{h}{\varepsilon}\frac{\varphi'(\Lambda_t)}{\varphi(\Lambda_t)}\left(\theta_C(1-y_t)+\theta_My_t\right)\right]^2};\label{thnC1}
\end{align}
$n_{Ct,i}^*$ ($i\in\{2,3,4\}$) satisfy that
\begin{align}
n_{Ct,2}^*=&\frac{\frac{1-y_t}{{\delta_C}\rho^{\frac{1}{\delta_C}}\varphi(\Lambda_t)^{\theta_C}}}
{\frac{1-y_t}{{\delta_C}\rho^{\frac{1}{\delta_C}}\varphi(\Lambda_t)^{\theta_C}}+
\frac{y_t}{{\delta_M}\rho^{\frac{1}{\delta_M}}\varphi(\Lambda_t)^{\theta_M}}},\label{thnC2}\\
n_{Ct,3}^*=&\frac{1}{1+(1-p)k},\label{thnC3}\\
n_{Ct,4}^*=&1;\label{thnC4}
\end{align}
and $J_C^*$ is given by (\ref{JCC}) and (\ref{J_C}) with $x_t^*,r_t,\mu_t,\sigma_t$ given by (\ref{xx}),(\ref{rt})-(\ref{sigmat}) and
\begin{align}
\frac{{D}_t}{S_t}=&\frac{1-l_C-l_M}{(1-y_t)\rho^{-\frac{1}{\delta_C}}\varphi(\Lambda_t)^{-\theta_C}+y_t\rho^{-\frac{1}{\delta_M}}\varphi(\Lambda_t)^{-\theta_M}},\label{rato4}\\
\frac{{D}_t}{W_{Ct}}=&\frac{1-l_C-l_M}{(1-y_t)\rho^{-\frac{1}{\delta_C}}\varphi(\Lambda_t)^{-\theta_C}},\label{rato5}\\
\frac{S_t}{W_{Ct}}=&\frac{(1-y_t)\rho^{-\frac{1}{\delta_C}}\varphi(\Lambda_t)^{-\theta_C}+y_t\rho^{-\frac{1}{\delta_M}}\varphi(\Lambda_t)^{-\theta_M}}{(1-y_t)\rho^{-\frac{1}{\delta_C}}\varphi(\Lambda_t)^{-\theta_C}}.\label{rato6}
\end{align}
}
\end{theorem}

\begin{remark}\label{zero}{\it
We note here that there always exists a number of stock holdings $n_{Ct,1}^*\in[0,1]$ for (\ref{thnC1}). Indeed, it suffices to show that the associated function
\begin{align*}
g(n)=&-n+n_{Ct,2}^*+n_{Ct,2}^*\left(1-\frac{n}{n_{Ct,3}^*}\right)\frac{(1-p)(1-l_C-l_M)}{2H\varepsilon^{2h}}
\frac{y_t}{{\delta_M}\rho^{\frac{1}{\delta_M}}\varphi(\Lambda_t)^{\theta_M}}\nonumber\\
&\times\frac{\left[n\rho^{\frac{1}{\delta_C}}\varphi(\Lambda_t)^{\theta_C}+(1-n)\rho^{\frac{1}{\delta_M}}\varphi(\Lambda_t)^{\theta_M}\right]^2}
{\left[\sigma_D-\frac{h}{\varepsilon}\frac{\varphi'(\Lambda_t)}{\varphi(\Lambda_t)}\left(\theta_C(1-y_t)+\theta_My_t\right)\right]^2}
\end{align*}
has at least one zero point in $[0,1]$. To this end, using $n_{Ct,2}^*\in[0,1]$ and $n_{Ct,3}^*\in(0,1]$ we could verify that
\begin{align*}
g(0)=&n_{Ct,2}^*+n_{Ct,2}^*\frac{(1-p)(1-l_C-l_M)}{2H\varepsilon^{2h}}\frac{y_t}{{\delta_M}\rho^{\frac{1}{\delta_M}}\varphi(\Lambda_t)^{\theta_M}}
\frac{\left[\rho^{\frac{1}{\delta_M}}\varphi(\Lambda_t)^{\theta_M}\right]^2}
{\left[\sigma_D-\frac{h}{\varepsilon}\frac{\varphi'(\Lambda_t)}{\varphi(\Lambda_t)}\left(\theta_C(1-y_t)+\theta_My_t\right)\right]^2}
\geq 0,\\
g(1)=&(n_{Ct,2}^*-1)+n_{Ct,2}^*\left(1-\frac{1}{n_{Ct,3}^*}\right)\frac{(1-p)(1-l_C-l_M)}{2H\varepsilon^{2h}\sigma_D^2}
\frac{y_t}{{\delta_M}\rho^{\frac{1}{\delta_M}}\varphi(\Lambda_t)^{\theta_M}}\\
&\times\frac{\left[\rho^{\frac{1}{\delta_C}}\varphi(\Lambda_t)^{\theta_C}\right]^2}
{\left[\sigma_D-\frac{h}{\varepsilon}\frac{\varphi'(\Lambda_t)}{\varphi(\Lambda_t)}\left(\theta_C(1-y_t)+\theta_My_t\right)\right]^2}
\leq 0.
\end{align*}
which implies by the zero point theorem that there exists at least one $n^*\in[0,1]$ such that $g(n^*)=0$.
}
\end{remark}

\begin{remark}{\it
By Theorem \ref{th1} (or Proposition 3 in \cite{Basak}), $n_{Ct}^*$ in Theorem \ref{th2} also solves the fixed-point equation
\begin{align}
n_{Ct}^*=&\mathop{\mathrm{argmax}}\limits_{n_{Ct},0\leq n_{Ct}\leq 1}J_C(n_{Ct})\nonumber\\
=&\mathop{\mathrm{argmax}}\limits_{n_{Ct},0\leq n_{Ct}\leq 1}\left\{\frac{n_{Ct}S_t}{W_{Ct}}\left(\mu_t-r_t+(1-x^*(n_{Ct}))\frac{D_t}{S_t}+\sigma_t\Lambda_t\right)\right.\nonumber\\
&\left.+x^*(n_{Ct})\frac{D_t}{W_{Ct}}-\frac{kx^*(n_{Ct})^2}{2}\frac{D_t}{W_{Ct}}
-H\varepsilon^{2h}\delta_C\left(\frac{S_t\sigma_t}{W_{Ct}}\right)^2n_{Ct}^2\right\}\label{fix}
\end{align}
with
\[
x^*(n_{Ct})=\min\left\{\frac{1-n_{Ct}}{k},(1-p)n_{Ct}\right\}
\]
and other parameters given in Theorem \ref{th2}. The stock holding $n_{Ct}^*$ could be chosen from $n_{Ct,i}^*\;(i=1,2,3,4)$. Hence, to obtain $n_{Ct}^*$, we just need to seek the Region $j\in \{1,2,3,4\}$ such that
\[
J_C(n_{Ct,j}^*)=\max\{J_C(n_{Ct,i}^*),i\in \{1,2,3,4\}\},
\]
where for each $j\in \{1,2,3,4\}$, the parameters in $J_C$ should be related to Region $j$.
}
\end{remark}

It is of much importance to note that the equilibrium in Theorem \ref{th1} may not exist. We give a simple example here. Set $H=p=y_t=\frac{1}{2},\gamma_M=\gamma_C,\alpha_M=\alpha_C=1, k=2,\rho=\sigma_D=0.01$, which is a special scenario in \cite{Basak}. Then it is easy to obtain $n_{Ct,1}^*=n_{Ct,2}^*=n_{Ct,3}^*=\frac{1}{2}$ and $n_{Ct,4}^*=1$.

In Region 1 (or Region 2, Region 3, equivalently), we have $\sigma_t=\sigma_D,\mu_t=\mu_D,r_t=\mu_D-\sigma_D^2+\frac{3}{8}\rho$ and $\frac{{D}_t}{S_t}=\frac{\rho}{2},\frac{{D}_t}{W_{Ct}}=\rho,\frac{S_t}{W_{Ct}}=2$, which gives
\begin{equation*}
J_C(n)=2n\left(\sigma_D^2-\frac{3}{8}\rho+(1-x^*(n))\frac{\rho}{2}\right)+x^*(n)\rho-x^*(n)^2\rho-2\sigma_D^2n^2.
\end{equation*}
Hence, $J_C(n_{Ct,1}^*)=J_C(n_{Ct,2}^*)=J_C(n_{Ct,3}^*)=\sigma_D^2+\frac{3}{16}\rho=0.0019$ and $J_C(n_{Ct,4}^*)=\frac{\rho}{4}=0.0025$. Since $0.0019<0.0025$, the stock holding $n_{Ct}^*$ does not locate in Region $j$ with $j=1,2,3$.

In Region 4, we have $\sigma_t=\sigma_D,\mu_t=\mu_D,r_t=\mu_D+\frac{1}{2}\rho$ and $\frac{{D}_t}{S_t}=\frac{\rho}{2},\frac{{D}_t}{W_{Ct}}=\rho,\frac{S_t}{W_{Ct}}=2$, which gives
\begin{equation*}
J_C(n)=-n\rho x^*(n)+x^*(n)\rho-x^*(n)^2\rho-2\sigma_D^2n^2.
\end{equation*}
Hence, $J_C(n_{Ct,1}^*)=J_C(n_{Ct,2}^*)=J_C(n_{Ct,3}^*)=-\frac{\sigma_D^2}{2}+\frac{\rho}{16}=5.57\times 10^{-4}$ and $J_C(n_{Ct,4}^*)=-2\sigma_D^2=-2\times 10^{-4}$. Since $5.57\times 10^{-4}>-2\times 10^{-4}$, the stock holding $n_{Ct}^*$ does not locate in Region 4.

Summarizing above discussion, the stock holding $n_{Ct}^*$, as well as the equilibrium in Theorem \ref{th1}, does not exist. However, we shall always consider scenarios where the equilibrium in Theorem \ref{th1} exists in our paper.

In a perfect economy, the controlling shareholder can not divert output for his private benefit, which is the case $p=1$ and denoted by the benchmark economy. Hence, we give the benchmark equilibrium by applying Theorem \ref{th1} directly.

\begin{corollary}\label{co1}{\it
In the model driven by the approximate fBm and under the equilibrium conditions (\ref{ec-1})-(\ref{ec-3}) with $p=1$ and $y_t\in(0,1)$, the shareholders' optimal consumptions $c_{it}^{\mathfrak{B}}$ ($i\in\{C,M\}$) are given by (\ref{ps_c}), and the interest rate $r_t^{\mathfrak{B}}$, the stock mean-return $\mu_t^{\mathfrak{B}}$, the stock volatility $\sigma_t^{\mathfrak{B}}$, the Sharpe ratio $\kappa_t^{\mathfrak{B}}$, the parameters $\mu_{yt}^{\mathfrak{B}}$ and $\sigma_{yt}^{\mathfrak{B}}$ are given by:
\begin{align}
r_t^\mathfrak{B}=&\mu_D+\sigma_D\Lambda_t-l_C\rho^{\frac{1}{\delta_C}}\varphi(\Lambda_t)^{\theta_C}-l_M\rho^{\frac{1}{\delta_M}}\varphi(\Lambda_t)^{\theta_M}\nonumber\\
&-n_{Ct}^\mathfrak{B}\sigma_t^\mathfrak{B}\left[\sqrt{2H}\varepsilon^h\kappa_t^\mathfrak{B}+2Hh\varepsilon^{2h-1}\theta_C\frac{\varphi'(\Lambda_t)}{\varphi(\Lambda_t)}\right]
\left[1-y_t+y_t\rho^{\frac{1}{\delta_C}-\frac{1}{\delta_M}}\varphi(\Lambda_t)^{\theta_C-\theta_M}\right]\nonumber\\
&-(1-n_{Ct}^\mathfrak{B})\sigma_t^\mathfrak{B}\left[\sqrt{2H}\varepsilon^h\kappa_t^\mathfrak{B}+2Hh\varepsilon^{2h-1}\theta_M\frac{\varphi'(\Lambda_t)}{\varphi(\Lambda_t)}\right]
\left[y_t+(1-y_t)\rho^{\frac{1}{\delta_M}-\frac{1}{\delta_C}}\varphi(\Lambda_t)^{\theta_M-\theta_C}\right]\nonumber\\
&+(1-y_t)\left[\rho^{\frac{1}{\delta_C}}\varphi(\Lambda_t)^{\theta_C}-\theta_C\lambda_t\frac{\varphi'(\Lambda_t)}{\varphi(\Lambda_t)}
-Hh^2\varepsilon^{2h-2}\theta_C\left((\theta_C-1)\left(\frac{\varphi'(\Lambda_t)}{\varphi(\Lambda_t)}\right)^2
+\frac{\varphi''(\Lambda_t)}{\varphi(\Lambda_t)}\right)\right]\nonumber\\
&+y_t\left[\rho^{\frac{1}{\delta_M}}\varphi(\Lambda_t)^{\theta_M}-\theta_M\lambda_t\frac{\varphi'(\Lambda_t)}{\varphi(\Lambda_t)}
-Hh^2\varepsilon^{2h-2}\theta_M\left((\theta_M-1)\left(\frac{\varphi'(\Lambda_t)}{\varphi(\Lambda_t)}\right)^2
+\frac{\varphi''(\Lambda_t)}{\varphi(\Lambda_t)}\right)\right]\nonumber\\
\mu_t^\mathfrak{B}=&r_t^\mathfrak{B}-\sigma_t^\mathfrak{B}\Lambda_t+\sqrt{2H}\varepsilon^h\kappa_t^\mathfrak{B}\sigma_t^\mathfrak{B}
-\frac{1-l_C-l_M}{(1-y_t)\rho^{-\frac{1}{\delta_C}}\varphi(\Lambda_t)^{-\theta_C}+y_t\rho^{-\frac{1}{\delta_M}}\varphi(\Lambda_t)^{-\theta_M}},\nonumber\\
\sigma_t^\mathfrak{B}=&\frac{\sigma_D-\frac{h}{\varepsilon}
\frac{\varphi'(\Lambda_t)}{\varphi(\Lambda_t)}\left[\theta_C(1-y_t)+\theta_My_t\right]}
{\left[n_{Ct}^\mathfrak{B}\rho^{\frac{1}{\delta_C}}\varphi(\Lambda_t)^{\theta_C}+(1-n_{Ct}^\mathfrak{B})\rho^{\frac{1}{\delta_M}}\varphi(\Lambda_t)^{\theta_M}\right]
\left[(1-y_t)\rho^{-\frac{1}{\delta_C}}\varphi(\Lambda_t)^{-\theta_C}+y_t\rho^{-\frac{1}{\delta_M}}\varphi(\Lambda_t)^{-\theta_M}\right]},\label{sigmaB}\\
\kappa_t^\mathfrak{B}=&\frac{\sqrt{2H}\varepsilon^h\delta_M\rho^{\frac{1}{\delta_M}}\varphi(\Lambda_t)^{\theta_M}(1-n_{Ct}^\mathfrak{B})
\left[\sigma_D-\frac{h}{\varepsilon}
\frac{\varphi'(\Lambda_t)}{\varphi(\Lambda_t)}\left(\theta_C(1-y_t)+\theta_My_t\right)\right]}
{y_t\left[n_{Ct}^\mathfrak{B}\rho^{\frac{1}{\delta_C}}\varphi(\Lambda_t)^{\theta_C}+(1-n_{Ct}^\mathfrak{B})\rho^{\frac{1}{\delta_M}}\varphi(\Lambda_t)^{\theta_M}\right]},\nonumber\\
\sigma_{yt}^\mathfrak{B}=&y_t\left(\frac{\kappa_t^\mathfrak{B}}{\sqrt{2H}\varepsilon^h\delta_M}+\theta_M\frac{h}{\varepsilon}
\frac{\varphi'(\Lambda_t)}{\varphi(\Lambda_t)}-\sigma_D\right),\nonumber\\
\mu_{yt}^\mathfrak{B}=&-\sigma_{yt}^\mathfrak{B}(\Lambda_t+2H\varepsilon^{2h}\sigma_D)+l_M\rho^{\frac{1}{\delta_M}}\varphi(\Lambda_t)^{\theta_M}
+y_t\bigg[r_t^\mathfrak{B}-\mu_D-\sigma_D\Lambda_t+\frac{(\kappa_t^\mathfrak{B})^2}{\delta_M}
+\sqrt{2H}h\varepsilon^{h-1}\theta_M
\frac{\varphi'(\Lambda_t)}{\varphi(\Lambda_t)}\frac{\kappa_t^\mathfrak{B}}{\delta_M}\nonumber\\
&-\rho^{\frac{1}{\delta_M}}\varphi(\Lambda_t)^{\theta_M}+\theta_M\lambda_t\frac{\varphi'(\Lambda_t)}{\varphi(\Lambda_t)}
+Hh^2\varepsilon^{2h-2}\theta_M\left((\theta_M-1)\left(\frac{\varphi'(\Lambda_t)}{\varphi(\Lambda_t)}\right)^2
+\frac{\varphi''(\Lambda_t)}{\varphi(\Lambda_t)}\right)\bigg].\nonumber
\end{align}
The fraction of diverted output $x_t^\mathfrak{B}$, the shareholders' optimal bond holdings $b_{it}^{\mathfrak{B}}$ and stock holdings $n_{it}^{\mathfrak{B}}$ ($i\in\{C,M\}$) could be expressed as follows:
\begin{align}
x_t^\mathfrak{B}=&0,\nonumber\\
b_{Ct}^{\mathfrak{B}}=&\exp\left(-\int_0^tr_s^{\mathfrak{B}}ds\right)(W_{Ct}-n_{Ct}^{\mathfrak{B}}S_t),\nonumber\\
b_{Mt}^{\mathfrak{B}}=&-b_{Ct}^{\mathfrak{B}},\nonumber\\
n_{Ct}^{\mathfrak{B}}=&\frac{\frac{1-y_t}{{\delta_C}\rho^{\frac{1}{\delta_C}}\varphi(\Lambda_t)^{\theta_C}}}
{\frac{1-y_t}{{\delta_C}\rho^{\frac{1}{\delta_C}}\varphi(\Lambda_t)^{\theta_C}}+
\frac{y_t}{{\delta_M}\rho^{\frac{1}{\delta_M}}\varphi(\Lambda_t)^{\theta_M}}},\label{conc}\\
n_{Mt}^{\mathfrak{B}}=&1-n_{Ct}^{\mathfrak{B}}.\nonumber
\end{align}
}
\end{corollary}

Conditions $y_t=0$ and $y_t=1$ implies that $W_{Mt}=0$ and $W_{Ct}=0$ respectively. Such related results are practically meaningful, but we can hardly apply those conditions to Theorem \ref{th1}. Noticing that all equilibrium results in Theorem \ref{th2} and Corollary \ref{co1} are continuous functions of $y_t\in(0,1)$ and $\Lambda_t$, we shall assume that they are continuous at $y_t=0$ and $y_t=1$, which means we take their values at $y_t=0$ and $y_t=1$ by considering $y_t\rightarrow 0^+$ and $y_t\rightarrow 1^-$ respectively. This method is implicitly used in \cite{Basak}. Hence, we extend all equilibrium results to the case $y_t\in[0,1]$ in Theorem \ref{th2} and Corollary \ref{co1}. As $y_t\rightarrow 0$ implies $W_{Mt}\rightarrow 0$, we obtain $n^*_{Mt}\rightarrow 0$ and then $n^*_{Ct}\rightarrow 1$ by $n^*_{Mt}+n^*_{Ct}=1$. Now we assume $p<1$ and  $y_t$ is sufficiently small in some interval $I=(0,\delta)$ with $\delta>0$. The fact that $n_{Ct}^*$ is a continuous function of $y_t$ implies that $n_{Ct}^*$ should locate in Region 2 or Region 4. If $n_{Ct}^*$ locates in Region 2 for $y_t\in I$ (which is the only case for $p=1$), $\kappa_t(y_t)$ converges to $\kappa_t(0)=\sqrt{2H}\varepsilon^h\delta_C[\sigma_D-\frac{h}{\varepsilon}
\frac{\varphi'(\Lambda_t)}{\varphi(\Lambda_t)}\theta_C]$ by substituting $n_{Ct,2}^*$. If $n_{Ct}^*$ locates in Region 4 for $y_t\in I$, $\kappa_t(y_t)$ converges to $\kappa_t(0)=0$. If $p=1$, then $n_{Ct}^*=n_{Ct,2}^*=n_{Ct}^{\mathfrak{B}}$.  Similar discussion could be done for the case $y_t=1$. Finally, by $y_t\rightarrow 0^+$ and $y_t\rightarrow 1^-$, the equilibrium results for the case $y_t=0$ and $y_t=1$ are explicitly given by
\begin{align}
x_t^*(0)=&0,\\
r_t(0)=&\mu_D+\sigma_D\Lambda_t-l_C\rho^{\frac{1}{\delta_C}}\varphi(\Lambda_t)^{\theta_C}-l_M\rho^{\frac{1}{\delta_M}}\varphi(\Lambda_t)^{\theta_M}
-\sigma_t(0)\left[\sqrt{2H}\varepsilon^h\kappa_t(0)+2Hh\varepsilon^{2h-1}\theta_C\frac{\varphi'(\Lambda_t)}{\varphi(\Lambda_t)}\right]\nonumber\\
&+\rho^{\frac{1}{\delta_C}}\varphi(\Lambda_t)^{\theta_C}-\theta_C\lambda_t\frac{\varphi'(\Lambda_t)}{\varphi(\Lambda_t)}
-Hh^2\varepsilon^{2h-2}\theta_C\left((\theta_C-1)\left(\frac{\varphi'(\Lambda_t)}{\varphi(\Lambda_t)}\right)^2
+\frac{\varphi''(\Lambda_t)}{\varphi(\Lambda_t)}\right),\label{rto}\\
\mu_t(0)=&r_t(0)-\sigma_t(0)\Lambda_t+\sqrt{2H}\varepsilon^h\kappa_t(0)\sigma_t(0)
-(1-l_C-l_M)\rho^{\frac{1}{\delta_C}}\varphi(\Lambda_t)^{\theta_C},\\
\sigma_t(0)=&\sigma_D-\frac{h}{\varepsilon}
\frac{\varphi'(\Lambda_t)}{\varphi(\Lambda_t)}\theta_C,\\
\kappa_t(0)=&\left\{
  \begin{array}{*{20}{l}}
{\sqrt{2H}\varepsilon^h\delta_C\left[\sigma_D-\frac{h}{\varepsilon}
\frac{\varphi'(\Lambda_t)}{\varphi(\Lambda_t)}\theta_C\right],}&{n_{Ct}^*=n_{Ct,2}^*\;\text{for $p<1$ and $y_t\in I$; or $p=1$}}\\
{0,}&{n_{Ct}^*=1\;\text{for $p<1$ and $y_t\in I$}}\\
\end{array}
\right.\\
\sigma_{yt}(0)=&0,\\
\mu_{yt}(0)=&l_M\rho^{\frac{1}{\delta_M}}\varphi(\Lambda_t)^{\theta_M},\\
n_{Ct}^*(0)=&1,\\
n_{Mt}^*(0)=&0,
\end{align}
and
\begin{align*}
x_t^*(1)=&0,\\
r_t(1)=&\mu_D+\sigma_D\Lambda_t-l_C\rho^{\frac{1}{\delta_C}}\varphi(\Lambda_t)^{\theta_C}-l_M\rho^{\frac{1}{\delta_M}}\varphi(\Lambda_t)^{\theta_M}
-\sigma_t(1)\left[\sqrt{2H}\varepsilon^h\kappa_t(1)+2Hh\varepsilon^{2h-1}\theta_M\frac{\varphi'(\Lambda_t)}{\varphi(\Lambda_t)}\right]\\
&+\rho^{\frac{1}{\delta_M}}\varphi(\Lambda_t)^{\theta_M}-\theta_M\lambda_t\frac{\varphi'(\Lambda_t)}{\varphi(\Lambda_t)}
-Hh^2\varepsilon^{2h-2}\theta_M\left((\theta_M-1)\left(\frac{\varphi'(\Lambda_t)}{\varphi(\Lambda_t)}\right)^2
+\frac{\varphi''(\Lambda_t)}{\varphi(\Lambda_t)}\right),\nonumber\\
\mu_t(1)=&r_t(1)-\sigma_t(1)\Lambda_t+\sqrt{2H}\varepsilon^h\kappa_t(1)\sigma_t(1)
-(1-l_C-l_M)\rho^{\frac{1}{\delta_M}}\varphi(\Lambda_t)^{\theta_M},\\
\sigma_t(1)=&\sigma_D-\frac{h}{\varepsilon}\frac{\varphi'(\Lambda_t)}{\varphi(\Lambda_t)}\theta_M,\\
\kappa_t(1)=&\sqrt{2H}\varepsilon^h\delta_M\left[\sigma_D-\frac{h}{\varepsilon}
\frac{\varphi'(\Lambda_t)}{\varphi(\Lambda_t)}\theta_M\right],\\
\sigma_{yt}(1)=&0,\\
\mu_{yt}(1)=&r_t(1)-\mu_D-\sigma_D\Lambda_t+\frac{\kappa_t(1)^2}{\delta_M}+l_M\rho^{\frac{1}{\delta_M}}\varphi(\Lambda_t)^{\theta_M}+\sqrt{2H}h\varepsilon^{h-1}\theta_M
\frac{\varphi'(\Lambda_t)}{\varphi(\Lambda_t)}\frac{\kappa_t(1)}{\delta_M}-\rho^{\frac{1}{\delta_M}}\varphi(\Lambda_t)^{\theta_M}\nonumber\\
&+\theta_M\lambda_t\frac{\varphi'(\Lambda_t)}{\varphi(\Lambda_t)}
+Hh^2\varepsilon^{2h-2}\theta_M\left((\theta_M-1)\left(\frac{\varphi'(\Lambda_t)}{\varphi(\Lambda_t)}\right)^2
+\frac{\varphi''(\Lambda_t)}{\varphi(\Lambda_t)}\right),\\
n_{Ct}^*(1)=&0,\\
n_{Mt}^*(1)=&1.
\end{align*}

\begin{remark}{\it
Compared with \cite{Basak}, equilibrium results in Theorems \ref{th1}-\ref{th2} and Corollary \ref{co1} additionally depend on the pathwise past information ($\Lambda_t$ and $\lambda_t$) and the sensitivity of investors to past information ($\alpha_C$ and $\alpha_M$). Hence, the pathwise past information indeed plays an important and different role in our approximate fractional economy.
}\end{remark}

To conclude this section, we give equilibrium results for (\ref{psi}) and (\ref{phi}) as a special example. Denote $\beta=\beta_1-\beta_2$, $\beta_M=\beta\theta_M$ and $\beta_C=\beta\theta_C$ for briefness. Making use of (\ref{psi}) and (\ref{phi}) and then applying
\[\varphi=e^{\beta\Lambda},\;\frac{\varphi'}{\varphi}=\beta,\;\frac{\varphi''}{\varphi}=\beta^2\]
to Theorem \ref{th2} and Corollary \ref{co1}, we can obtain the following corollaries.
\begin{corollary}\label{co2}{\it
Assume the conditions in Theorem \ref{th2} hold, and functions $\phi,\psi$ further satisfy (\ref{psi}) and (\ref{phi}). Then the shareholders' optimal consumptions $c_{it}^*$ ($i\in\{C,M\}$) are given by (\ref{ps_c}), and the interest rate $r_t$, the stock mean-return $\mu_t$, the stock volatility $\sigma_t$, the Sharpe ratio $\kappa_t$, the parameters $\mu_{yt}$ and $\sigma_{yt}$ are given by:
\begin{align}
r_t=&\mu_D+\sigma_D\Lambda_t
-n_{Ct}^*\sigma_t\left[\sqrt{2H}\varepsilon^h\kappa_t+2Hh\varepsilon^{2h-1}\beta_C\right]
\left[1-y_t+y_t\rho^{\frac{1}{\delta_C}-\frac{1}{\delta_M}}e^{(\beta_C-\beta_M)\Lambda_t}\right]
-l_C\rho^{\frac{1}{\delta_C}}e^{\beta_C\Lambda_t}\nonumber\\
&-(1-n_{Ct}^*)\sigma_t\left[\sqrt{2H}\varepsilon^h\kappa_t+2Hh\varepsilon^{2h-1}\beta_M\right]
\left[y_t+(1-y_t)\rho^{\frac{1}{\delta_M}-\frac{1}{\delta_C}}e^{(\beta_M-\beta_C)\Lambda_t}\right]
-l_M\rho^{\frac{1}{\delta_M}}e^{\beta_M\Lambda_t}\nonumber\\
&+(1-y_t)\left[\rho^{\frac{1}{\delta_C}}e^{\beta_C\Lambda_t}-\beta_C\lambda_t
-Hh^2\varepsilon^{2h-2}\beta_C^2\right]
+y_t\left[\rho^{\frac{1}{\delta_M}}e^{\beta_M\Lambda_t}-\beta_M\lambda_t
-Hh^2\varepsilon^{2h-2}\beta_M^2)\right]\nonumber\\
&-\left(x_t^*-\frac{kx_t^{*2}}{2}\right)(1-l_C-l_M)\rho^{\frac{1}{\delta_C}}e^{\beta_C\Lambda_t}
-\frac{kx_t^{*2}}{2}(1-l_C-l_M)\rho^{\frac{1}{\delta_M}}e^{\beta_M\Lambda_t},\label{r0}\\
\mu_t=&r_t-\sigma_t\Lambda_t+\sqrt{2H}\varepsilon^h\kappa_t\sigma_t
-\frac{(1-x_t^*)(1-l_C-l_M)}{(1-y_t)\rho^{-\frac{1}{\delta_C}}e^{-\beta_C\Lambda_t}
+y_t\rho^{-\frac{1}{\delta_M}}e^{-\beta_M\Lambda_t}},\label{miu0}\\
\sigma_t=&\frac{\sigma_D-\frac{h}{\varepsilon}
\left[\beta_C(1-y_t)+\beta_My_t\right]}
{\left[n_{Ct}^*\rho^{\frac{1}{\delta_C}}e^{\beta_C\Lambda_t}+(1-n_{Ct}^*)\rho^{\frac{1}{\delta_M}}e^{\beta_M\Lambda_t}\right]
\left[(1-y_t)\rho^{-\frac{1}{\delta_C}}e^{-\beta_C\Lambda_t}+y_t\rho^{-\frac{1}{\delta_M}}e^{-\beta_M\Lambda_t}\right]},\label{sigma0}\\
\kappa_t=&\frac{\sqrt{2H}\varepsilon^h\delta_M\rho^{\frac{1}{\delta_M}}e^{\beta_M\Lambda_t}(1-n_{Ct}^*)
\left[\sigma_D-\frac{h}{\varepsilon}
\left(\beta_C(1-y_t)+\beta_My_t\right)\right]}
{y_t\left[n_{Ct}^*\rho^{\frac{1}{\delta_C}}e^{\beta_C\Lambda_t}+(1-n_{Ct}^*)\rho^{\frac{1}{\delta_M}}e^{\beta_M\Lambda_t}\right]},\label{kappa0}\\
\sigma_{yt}=&y_t\left(\frac{\kappa_t}{\sqrt{2H}\varepsilon^h\delta_M}+\beta_M\frac{h}{\varepsilon}
-\sigma_D\right),\label{mus0}\\
\mu_{yt}=&-\sigma_{yt}(\Lambda_t+2H\varepsilon^{2h}\sigma_D)+l_M\rho^{\frac{1}{\delta_M}}e^{\beta_M\Lambda_t}
+\frac{kx_t^{*2}}{2}(1-l_C-l_M)\rho^{\frac{1}{\delta_M}}e^{\beta_M\Lambda_t}\nonumber\\
&+y_t\bigg[r_t-\mu_D-\sigma_D\Lambda_t+\frac{\kappa_t^2}{\delta_M}+\sqrt{2H}h\varepsilon^{h-1}\beta_M
\frac{\kappa_t}{\delta_M}-\rho^{\frac{1}{\delta_M}}e^{\beta_M\Lambda_t}
+\beta_M\lambda_t
+Hh^2\varepsilon^{2h-2}\beta_M^2\bigg].\label{muy0}
\end{align}
The shareholders' optimal bond holdings $b_{it}^*$ and stock holdings $n_{it}^*$ ($i\in\{C,M\}$) could be expressed as follows:
\begin{align*}
b_{Ct}^*=&\exp\left(-\int_0^tr_sds\right)(W_{Ct}-n_{Ct}^*S_t),\\
b_{Mt}^*=&-b_{Ct}^*,\\
n_{Ct}^*=&n_{Ct,j}^*,\quad \text{if}\, J_C^*=J_C(n_{Ct,j}^*;x_t^*)(\text{Region}\, j),\, j\in\{1,2,3,4\},\\
n_{Mt}^*=&1-n_{Ct}^*,
\end{align*}
where $n_{Ct,1}^*$ satisfies the following fixed point problem
\begin{align}
n_{Ct,1}^*=&n_{Ct,2}^*+n_{Ct,2}^*\left(1-\frac{n_{Ct,1}^*}{n_{Ct,3}^*}\right)\frac{(1-p)(1-l_C-l_M)}{2H\varepsilon^{2h}}
\frac{y_t}{{\delta_M}\rho^{\frac{1}{\delta_M}}e^{\beta_M\Lambda_t}}\nonumber\\
&\times\frac{\left[n_{Ct,1}^*\rho^{\frac{1}{\delta_C}}e^{\beta_C\Lambda_t}+(1-n_{Ct,1}^*)\rho^{\frac{1}{\delta_M}}e^{\beta_M\Lambda_t}\right]^2
}{\left[\sigma_D-\frac{h}{\varepsilon}\left(\beta_C(1-y_t)+\beta_My_t\right)\right]^2};\label{co2nc1}
\end{align}
$n_{Ct,i}^*$ ($i\in\{2,3,4\}$) satisfy that
\begin{align}
n_{Ct,2}^*=&\frac{\frac{1-y_t}{{\delta_C}\rho^{\frac{1}{\delta_C}}e^{\beta_C\Lambda_t}}}
{\frac{1-y_t}{{\delta_C}\rho^{\frac{1}{\delta_C}}e^{\beta_C\Lambda_t}}+
\frac{y_t}{{\delta_M}\rho^{\frac{1}{\delta_M}}e^{\beta_M\Lambda_t}}},\label{co2nc2}\\
n_{Ct,3}^*=&\frac{1}{1+(1-p)k},\label{co2nc3}\\
n_{Ct,4}^*=&1;\label{co2nc4}
\end{align}
and $J_C^*$ is given by (\ref{JCC}) and (\ref{J_C}) with $x_t^*,r_t,\mu_t,\sigma_t$ given by (\ref{xx}),(\ref{r0})-(\ref{sigma0}) and
\begin{align*}
\frac{{D}_t}{S_t}=&\frac{1-l_C-l_M}{(1-y_t)\rho^{-\frac{1}{\delta_C}}e^{-\beta_C\Lambda_t}+y_t\rho^{-\frac{1}{\delta_M}}e^{-\beta_M\Lambda_t}},\\
\frac{{D}_t}{W_{Ct}}=&\frac{1-l_C-l_M}{(1-y_t)\rho^{-\frac{1}{\delta_C}}e^{-\beta_C\Lambda_t}},\\
\frac{S_t}{W_{Ct}}=&\frac{(1-y_t)\rho^{-\frac{1}{\delta_C}}e^{-\beta_C\Lambda_t}+y_t\rho^{-\frac{1}{\delta_M}}e^{-\beta_M\Lambda_t}}{(1-y_t)\rho^{-\frac{1}{\delta_C}}e^{-\beta_C\Lambda_t}}.
\end{align*}
}
\end{corollary}

\begin{corollary}\label{co3}{\it
Assume the conditions in Corollary \ref{co1} hold, and functions $\phi,\psi$ further satisfy (\ref{psi}) and (\ref{phi}). Then the shareholders' optimal consumptions $c_{it}^\mathfrak{B}$ ($i\in\{C,M\}$) are given by (\ref{ps_c}), and the interest rate $r_t^\mathfrak{B}$, the stock mean-return $\mu_t^\mathfrak{B}$, the stock volatility $\sigma_t^\mathfrak{B}$, the Sharpe ratio $\kappa_t^\mathfrak{B}$, the parameters $\mu_{yt}^\mathfrak{B}$ and $\sigma_{yt}^\mathfrak{B}$ are given by:
\begin{align*}
r_t^\mathfrak{B}=&\mu_D+\sigma_D\Lambda_t
-n_{Ct}^\mathfrak{B}\sigma_t^\mathfrak{B}\left[\sqrt{2H}\varepsilon^h\kappa_t^\mathfrak{B}+2Hh\varepsilon^{2h-1}\beta_C\right]
\left[1-y_t+y_t\rho^{\frac{1}{\delta_C}-\frac{1}{\delta_M}}e^{(\beta_C-\beta_M)\Lambda_t}\right]
-l_C\rho^{\frac{1}{\delta_C}}e^{\beta_C\Lambda_t}\nonumber\\
&-(1-n_{Ct}^\mathfrak{B})\sigma_t^\mathfrak{B}\left[\sqrt{2H}\varepsilon^h\kappa_t+2Hh\varepsilon^{2h-1}\beta_M\right]
\left[y_t+(1-y_t)\rho^{\frac{1}{\delta_M}-\frac{1}{\delta_C}}e^{(\beta_M-\beta_C)\Lambda_t}\right]
-l_M\rho^{\frac{1}{\delta_M}}e^{\beta_M\Lambda_t}\nonumber\\
&+(1-y_t)\left[\rho^{\frac{1}{\delta_C}}e^{\beta_C\Lambda_t}-\beta_C\lambda_t
-Hh^2\varepsilon^{2h-2}\beta_C^2\right]
+y_t\left[\rho^{\frac{1}{\delta_M}}e^{\beta_M\Lambda_t}-\beta_M\lambda_t
-Hh^2\varepsilon^{2h-2}\beta_M^2)\right],\nonumber\\
\mu_t^\mathfrak{B}=&r_t^\mathfrak{B}-\sigma_t^\mathfrak{B}\Lambda_t+\sqrt{2H}\varepsilon^h\kappa_t^\mathfrak{B}\sigma_t^\mathfrak{B}
-\frac{1-l_C-l_M}{(1-y_t)\rho^{-\frac{1}{\delta_C}}e^{-\beta_C\Lambda_t}
+y_t\rho^{-\frac{1}{\delta_M}}e^{-\beta_M\Lambda_t}},\nonumber\\
\sigma_t^\mathfrak{B}=&\frac{\sigma_D-\frac{h}{\varepsilon}
\left[\beta_C(1-y_t)+\beta_My_t\right]}
{\left[n_{Ct}^\mathfrak{B}\rho^{\frac{1}{\delta_C}}e^{\beta_C\Lambda_t}+(1-n_{Ct}^\mathfrak{B})\rho^{\frac{1}{\delta_M}}e^{\beta_M\Lambda_t}\right]
\left[(1-y_t)\rho^{-\frac{1}{\delta_C}}e^{-\beta_C\Lambda_t}+y_t\rho^{-\frac{1}{\delta_M}}e^{-\beta_M\Lambda_t}\right]},\label{sigma00}\\
\kappa_t^\mathfrak{B}=&\frac{\sqrt{2H}\varepsilon^h\delta_M\rho^{\frac{1}{\delta_M}}e^{\beta_M\Lambda_t}(1-n_{Ct}^\mathfrak{B})
\left[\sigma_D-\frac{h}{\varepsilon}
\left(\beta_C(1-y_t)+\beta_My_t\right)\right]}
{y_t\left[n_{Ct}^\mathfrak{B}\rho^{\frac{1}{\delta_C}}e^{\beta_C\Lambda_t}+(1-n_{Ct}^\mathfrak{B})\rho^{\frac{1}{\delta_M}}e^{\beta_M\Lambda_t}\right]},\nonumber\\
\sigma_{yt}^\mathfrak{B}=&y_t\left(\frac{\kappa_t^\mathfrak{B}}{\sqrt{2H}\varepsilon^h\delta_M}+\beta_M\frac{h}{\varepsilon}
-\sigma_D\right),\nonumber\\
\mu_{yt}^\mathfrak{B}=&-\sigma_{yt}^\mathfrak{B}(\Lambda_t+2H\varepsilon^{2h}\sigma_D)+l_M\rho^{\frac{1}{\delta_M}}e^{\beta_M\Lambda_t}\nonumber\\
&+y_t\bigg[r_t^\mathfrak{B}-\mu_D-\sigma_D\Lambda_t+\frac{(\kappa_t^\mathfrak{B})^2}{\delta_M}+\sqrt{2H}h\varepsilon^{h-1}\beta_M
\frac{\kappa_t^\mathfrak{B}}{\delta_M}-\rho^{\frac{1}{\delta_M}}e^{\beta_M\Lambda_t}
+\beta_M\lambda_t
+Hh^2\varepsilon^{2h-2}\beta_M^2\bigg].\nonumber
\end{align*}
The fraction of diverted output $x_t^\mathfrak{B}$, the shareholders' optimal bond holdings $b_{it}^\mathfrak{B}$ and stock holdings $n_{it}^\mathfrak{B}$ ($i\in\{C,M\}$) could be expressed as follows:
\begin{align*}
x_t^\mathfrak{B}=&0,\\
b_{Ct}^\mathfrak{B}=&\exp\left(-\int_0^tr_sds\right)(W_{Ct}-n_{Ct}^\mathfrak{B}S_t),\nonumber\\
b_{Mt}^\mathfrak{B}=&-b_{Ct}^\mathfrak{B},\\
n_{Ct}^\mathfrak{B}=&\frac{\frac{1-y_t}{{\delta_C}\rho^{\frac{1}{\delta_C}}e^{\beta_C\Lambda_t}}}
{\frac{1-y_t}{{\delta_C}\rho^{\frac{1}{\delta_C}}e^{\beta_C\Lambda_t}}+
\frac{y_t}{{\delta_M}\rho^{\frac{1}{\delta_M}}e^{\beta_M\Lambda_t}}},\\
n_{Mt}^\mathfrak{B}=&1-n_{Ct}^\mathfrak{B}.
\end{align*}
}
\end{corollary}

\begin{remark}\label{diff}{\it
Making use of Corollaries \ref{co2} and \ref{co3}, it is not hard to obtain
\begin{align*}
\sigma_t-\sigma_t^\mathfrak{B}=&\frac{\sigma_t^\mathfrak{B}\left(\rho^{\frac{1}{\delta_M}}e^{\beta_M\Lambda_t}-\rho^{\frac{1}{\delta_C}}e^{\beta_C\Lambda_t}\right)
(n_{Ct}^*-n_{Ct}^\mathfrak{B})}{n_{Ct}^*\rho^{\frac{1}{\delta_C}}e^{\beta_C\Lambda_t}+(1-n_{Ct}^*)\rho^{\frac{1}{\delta_M}}e^{\beta_M\Lambda_t}},\\
\kappa_t-\kappa^\mathfrak{B}=&-\frac{\sqrt{2H}\varepsilon^h\rho^{\frac{1}{\delta_C}}e^{\beta_C\Lambda_t}[\sigma_D-\frac{h}{\varepsilon}
(\beta_C(1-y_t)+\beta_My_t)](n_{Ct}^*-n_{Ct}^\mathfrak{B})}
{\left[\frac{1-y_t}{\delta_C}+\frac{y_t}{\delta_M}\right]\left[n_{Ct}^*\rho^{\frac{1}{\delta_C}}e^{\beta_C\Lambda_t}+(1-n_{Ct}^*)\rho^{\frac{1}{\delta_M}}e^{\beta_M\Lambda_t}\right]
n_{Mt}^\mathfrak{B}},\\
\sigma_{yt}-\sigma_{yt}^\mathfrak{B}=&-\frac{y_t\rho^{\frac{1}{\delta_C}}e^{\beta_C\Lambda_t}[\sigma_D-\frac{h}{\varepsilon}
(\beta_C(1-y_t)+\beta_My_t)](n_{Ct}^*-n_{Ct}^\mathfrak{B})}
{\delta_M\left[\frac{1-y_t}{\delta_C}+\frac{y_t}{\delta_M}\right]\left[n_{Ct}^*\rho^{\frac{1}{\delta_C}}e^{\beta_C\Lambda_t}+(1-n_{Ct}^*)\rho^{\frac{1}{\delta_M}}e^{\beta_M\Lambda_t}\right]
n_{Mt}^\mathfrak{B}}.
\end{align*}
Consequently, the equilibrium processes (\ref{r0})-(\ref{muy0}) can be expressed as functions of minority shareholder’s consumption share $y_t$, their counterparts in the benchmark economy, the fraction of diverted output $x_t^*$, the past information $\Lambda_t$ (but independent of $\lambda_t$) and the excess ownership concentration $n_{Ct}^*-n_{Ct}^\mathfrak{B}$. Furthermore, it can be verified that higher ownership concentration relative to the full protection benchmark leads to higher stock volatility if $\rho^{\frac{1}{\delta_M}}\beta^{\theta_M}-\rho^{\frac{1}{\delta_C}}\beta^{\theta_C}>0$ (lower stock volatility if $\rho^{\frac{1}{\delta_M}}\beta^{\theta_M}-\rho^{\frac{1}{\delta_C}}\beta^{\theta_C}<0$), and lower Sharpe ratio and lower volatility of the minority shareholder’s consumption share if $\sigma_D-\frac{h}{\varepsilon}
(\beta_C(1-y_t)+\beta_My_t)>0$ (higher Sharpe ratio and higher volatility of the minority shareholder’s consumption share if $\sigma_D-\frac{h}{\varepsilon}
(\beta_C(1-y_t)+\beta_My_t)<0$) .
}\end{remark}

\section{Numerical results}\label{section5}\noindent
\setcounter{equation}{0}
In this section, we first construct an approximation scheme with pathwise convergence for past information, and then by numerical results, we analyse how investor protection and past information affects dynamics in equilibrium.

Following the work of \cite{Basak}, we set $\mu_D=0.015,\sigma_D=0.13,\gamma_C=3,\gamma_M=3.5,k=3,l_C=0.1,l_M=0.5$
as basic parameters for numerical analysis based on Corollaries \ref{co2} and \ref{co3}. Moreover, we set $\varepsilon=0.1$, $\rho=0.05$, $\beta=0.1$, $\alpha_C=0.5$, $\alpha_M=0.75$, $\lambda_t=0$ and vary $y_t$ from 0 to 1 with the step length 0.01.

\subsection{Past information}\label{subsection1}
The past information $\Lambda_t$ is a remarkable characteristic in our approximate fractional economy and determines how investors treat the economy at the present time $t$ by using the historical realized  data. As we noted in Remark \ref{memory}, ``no memory" ($\Lambda_t=0$), ``good memory" ($\Lambda_t>0$) and  ``bad memory" ($\Lambda_t<0$) of the economy could lead to different equilibria.

For a fixed $T>0$ we use the information $\{\widehat{D}_t\}_{0\leq s \leq T}$ to estimate $\Lambda$ and $\lambda$ by constructing an approximation scheme with pathwise convergence.

Setting $Z_t=\ln \widehat{D}_t$, it is seen by It\^{o}'s Lemma that
\begin{equation}\label{Z}
dZ_t=\left(\mu_D+\sigma_D\Lambda_t-H\varepsilon^{2h}\sigma_D^2\right)dt+\sqrt{2H}\varepsilon^h\sigma_Ddw_t.
\end{equation}
Define the approximation scheme $\widehat{\Lambda}_{t_{n+1}}$ and $\widehat{\lambda}_{t_{n+1}}$ for $\Lambda$ and $\lambda$ at $t=t_{n+1}$ as
\begin{align}
\Delta \widehat{w}_{t_{n+1}}=&\frac{Z_{t_{n+1}}-Z_{t_{n}}-\left(\mu_D+\sigma_D\widehat{\Lambda}_{t_{n}}-H\varepsilon^{2h}\sigma_D^2\right)\Delta t}{\sqrt{2H}\varepsilon^h\sigma_D};\label{s1}\\
\widehat{\Lambda}_{t_{n+1}}=&\sum_{k=0}^n\sqrt{2H}h(t_{n+1}-t_k+\varepsilon)^{h-1}\Delta \widehat{w}_{t_{k+1}},\label{s2}\\
\widehat{\lambda}_{t_{n+1}}=&\sum_{k=0}^n\sqrt{2H}h(h-1)(t_{n+1}-t_k+\varepsilon)^{h-2}\Delta \widehat{w}_{t_{k+1}},\quad \widehat{\Lambda}_{t_0}=0,\label{s0}
\end{align}
where $n=0,1,\cdots,N-1$ with $N\in\mathbb{N}$ and $0=t_0<t_1<\cdots<t_n<\cdots<t_{N-1}<t_N=T$ are the time nodes of the discretization. For simplicity we also assume that $\Delta t=t_{n+1}-t_{n}=T/N\leq 1$ for each $n$.

In order to prove pathwise convergence of the approximation scheme, we assume that $T,\Delta t,\varepsilon,H,\mu_D,\sigma_D$ are all fixed constants. Then we can obtain the following theorem of pathwise convergence.

\begin{theorem}\label{th4.1}{\it
For all $\varsigma>0$, there exists a nonnegative random variable $K_\varsigma$ (depending only on $\varsigma$ with $\mathbb{E}K_\varsigma^q<+\infty$ for all $q\geq 1$) such that
\begin{align}
\max\limits_{n=0,1,\cdots,N-1}\left\{\left|\Lambda_{t_{n+1}}(\omega)-\widehat{\Lambda}_{t_{n+1}}(\omega)\right|,
\left|\lambda_{t_{n+1}}(\omega)-\widehat{\lambda}_{t_{n+1}}(\omega)\right|\right\}
\leq K_\varsigma(\omega)(\Delta t)^{\frac{1}{6}-\varsigma}
\end{align}
for almost all $\omega\in\Omega$.
}
\end{theorem}
With Theorem \ref{th4.1} and the approximation scheme (\ref{s1})-(\ref{s0}), we can estimate $\Lambda_t$ and $\lambda_t$ pathwise at the time nodes $0=t_0<t_1<\cdots<t_n<\cdots<t_{N-1}<t_N=T$.
\begin{remark}{\it
We note here that the approximation schemes (\ref{s1})-(\ref{s0}) works in the case $H=\frac{1}{2}$ though we do not have to use it to approximate $\Lambda=\lambda=0$. Indeed, since $h=H-\frac{1}{2}=0$ in the case $H=\frac{1}{2}$, it is obtained from (\ref{s2}) and (\ref{s0}) that $\widehat{\Lambda}_{t_{n+1}}=\widehat{\lambda}_{t_{n+1}}=0$, which clearly implies $\widehat{\Lambda}_{t_{n+1}}=\Lambda_{t_{n+1}}$ and $\widehat{\lambda}_{t_{n+1}}=\lambda_{t_{n+1}}$ for $n=0,1,\cdots,N-1.$
}
\end{remark}

Observing the approximation scheme (\ref{s1})-(\ref{s2}), there are multiple complex factors to determine the memory state among ``no memory", ``good memory" and ``bad memory". So we just roughly analyze the past information $\Lambda_t$ by numerical methods. We set specially  $\varepsilon=0.1$ (the case of $\varepsilon=10^{-5}$ is only used for comparison), the time interval $[0,T]=[0,1]$ with
$\Delta t=10^{-3}$ and some given historical realized data $Z$ as
\begin{align*}
Z_1(t)=&0.015+0.02(t-0.5),\\
Z_2(t)=&0.015-0.02(t-0.5),\\
Z_3(t)=&0.015.
\end{align*}
\begin{figure}[!htbp]
\centering
\subfigure[$H=0.35$]{
\includegraphics[ width = 0.31\textwidth]{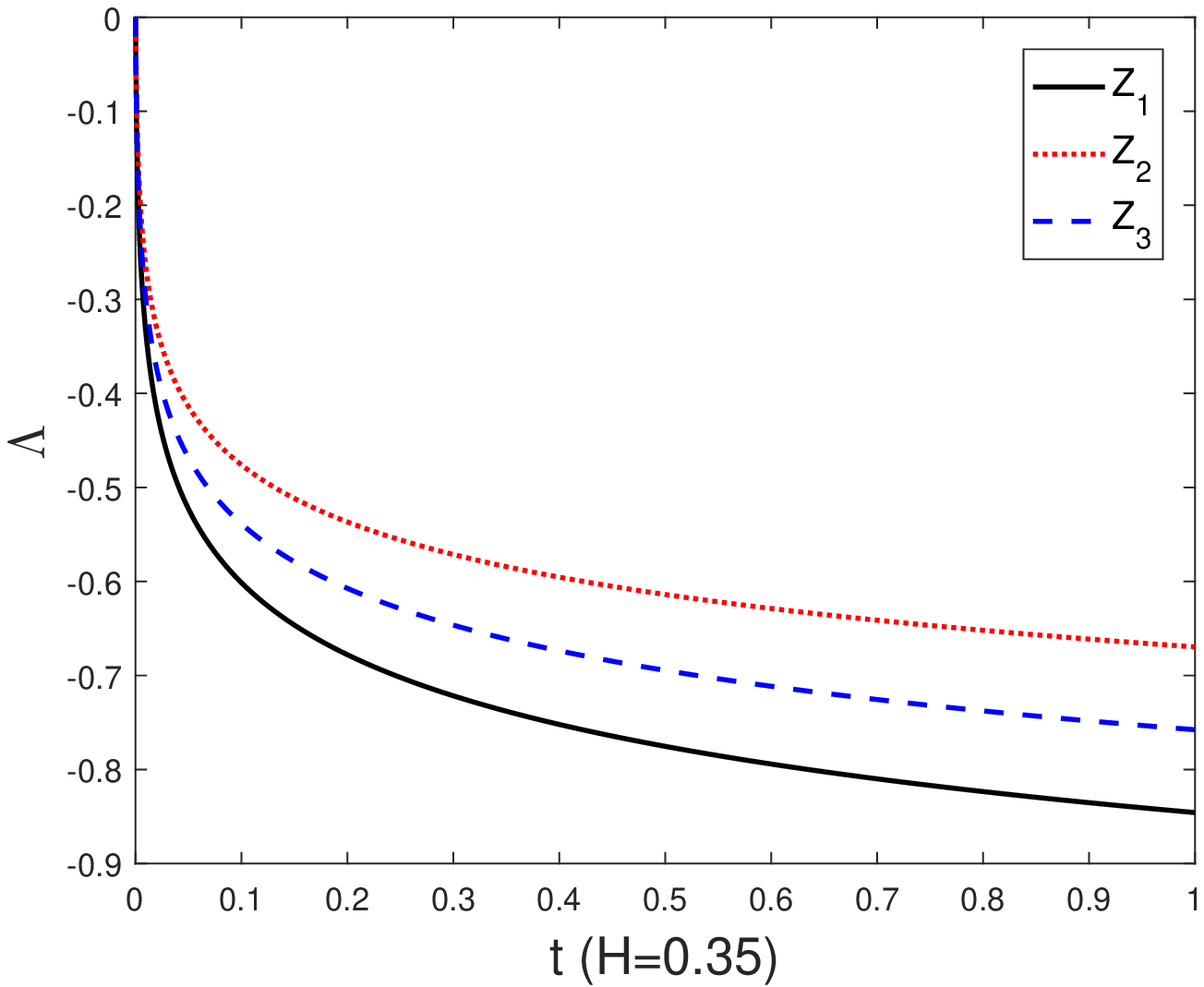}
}
\subfigure[$H=0.5$]{
\includegraphics[ width = 0.31\textwidth]{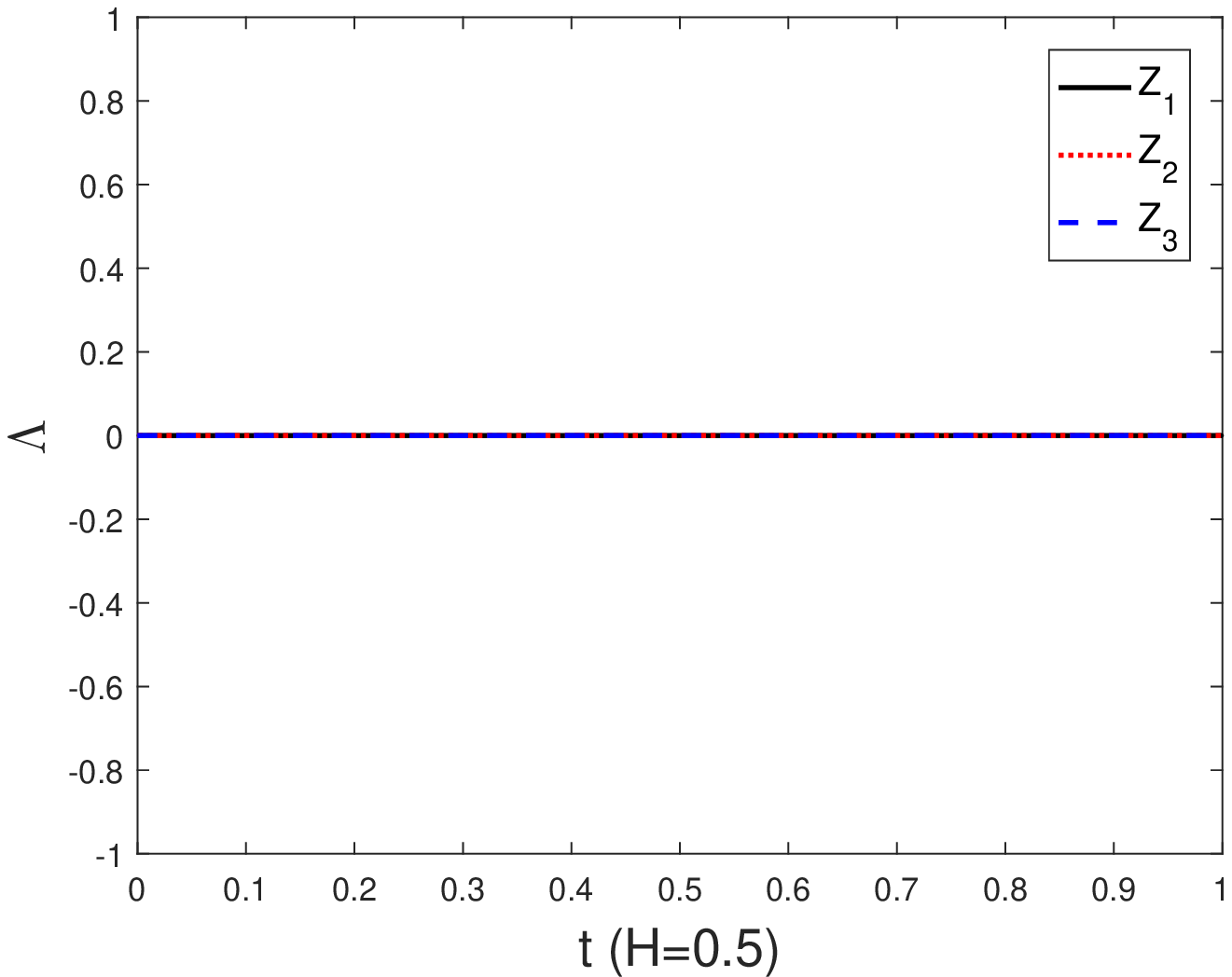}
}
\subfigure[$H=0.65$]{
\includegraphics[ width = 0.31\textwidth]{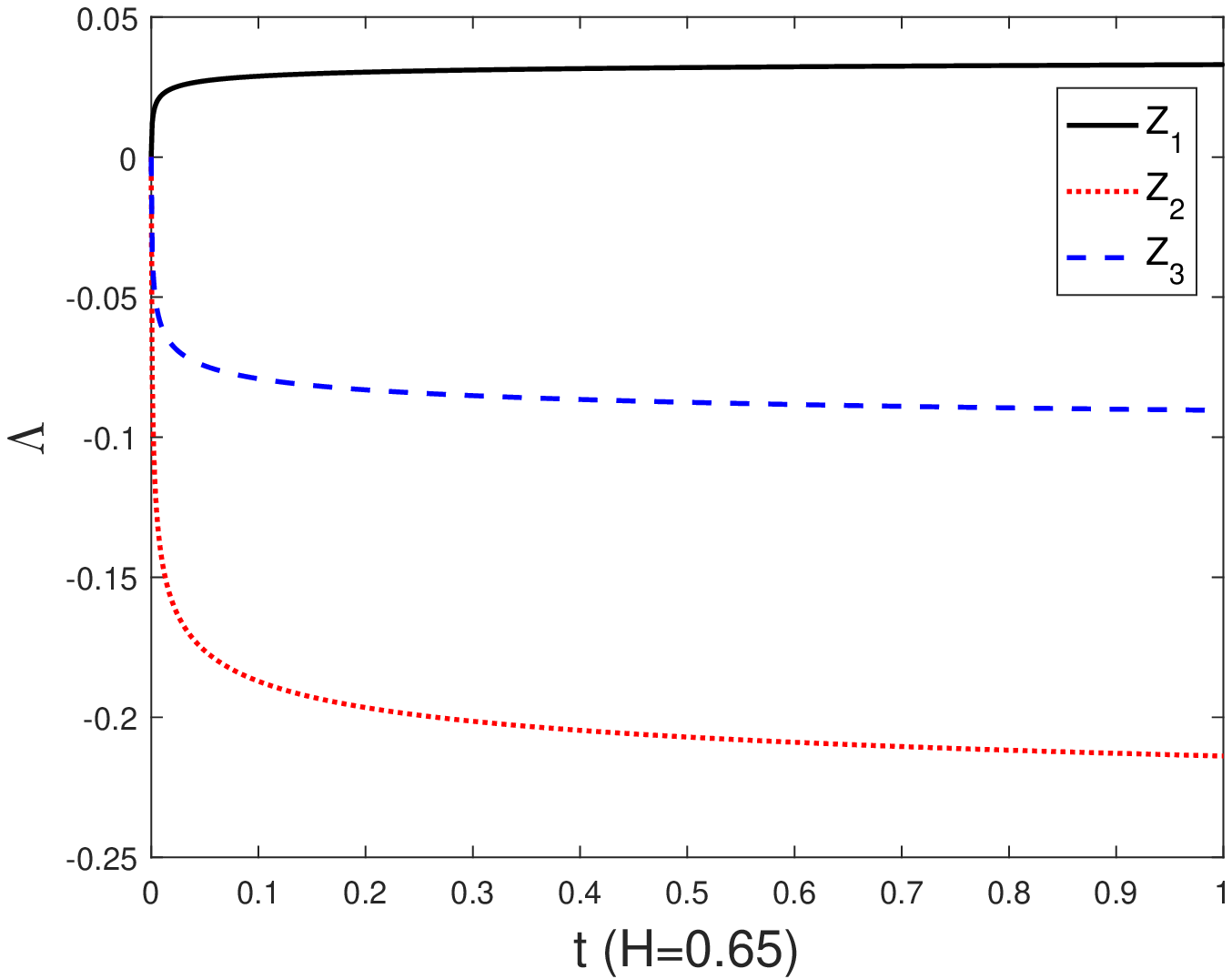}
}
\caption{Past information $\Lambda_t$ with $\varepsilon=10^{-5}$}
\label{fig_W}
\end{figure}

\begin{figure}[!htbp]
\centering
\subfigure[$H=0.35$]{
\includegraphics[ width = 0.31\textwidth]{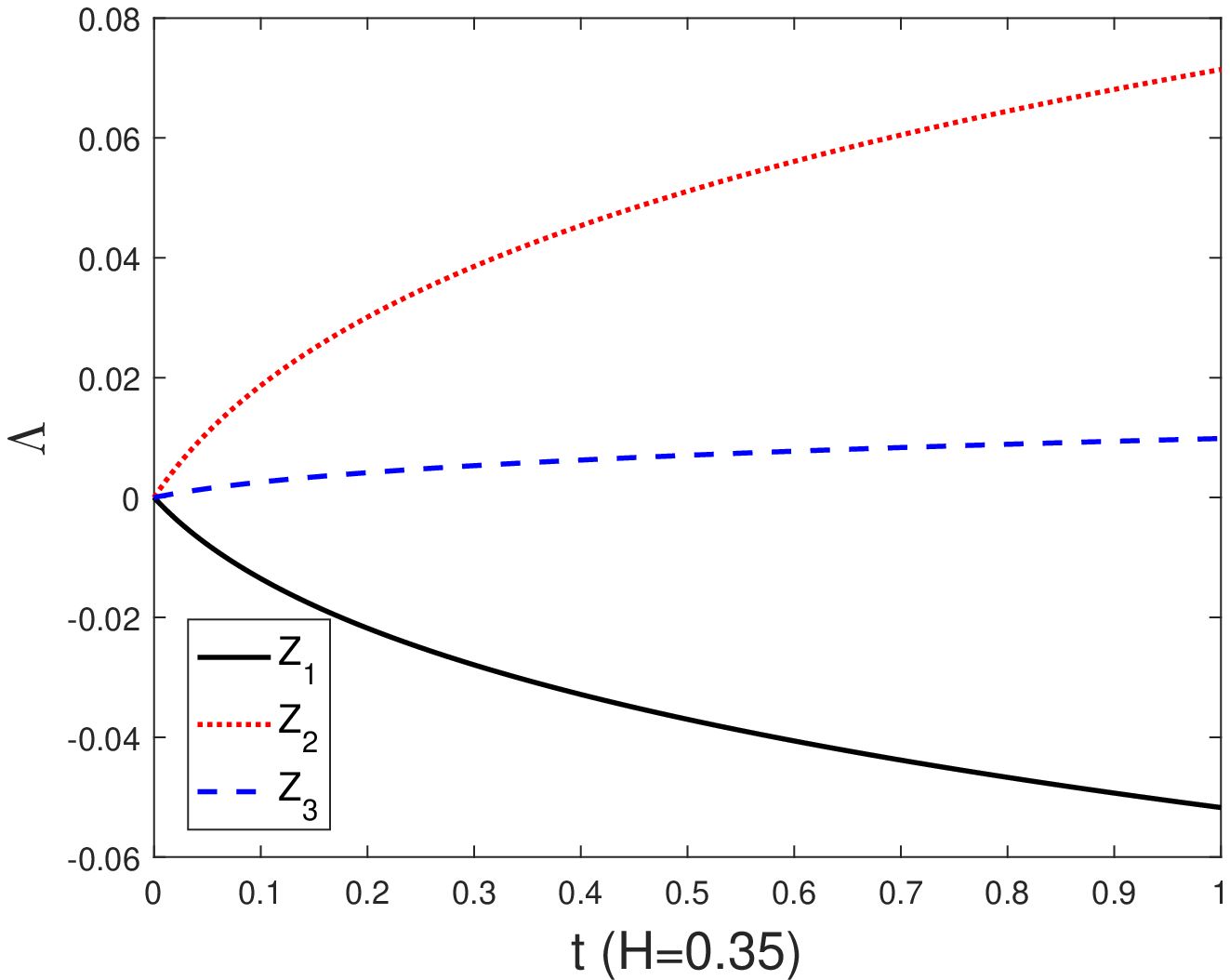}
}
\subfigure[$H=0.5$]{
\includegraphics[ width = 0.31\textwidth]{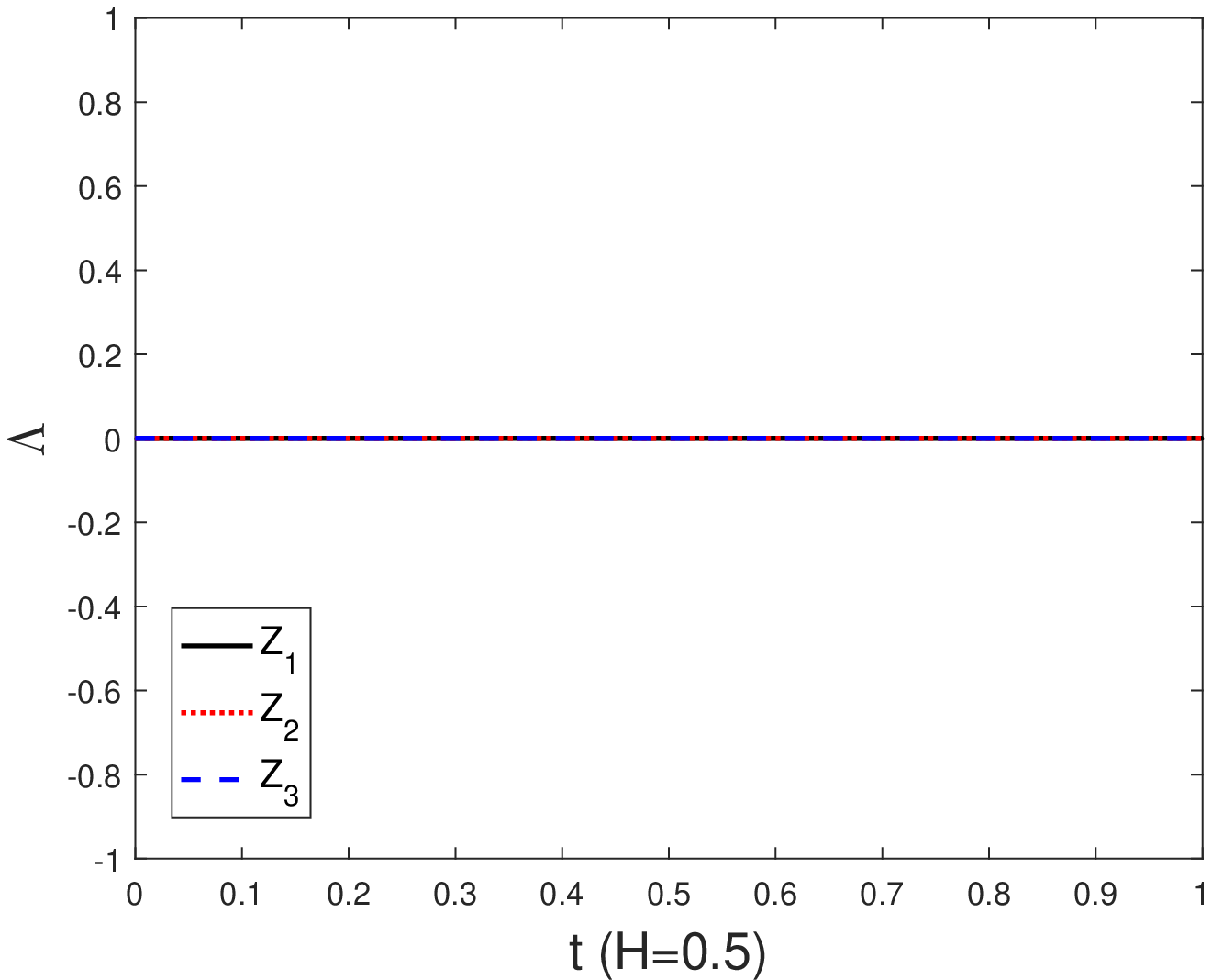}
}
\subfigure[$H=0.65$]{
\includegraphics[ width = 0.31\textwidth]{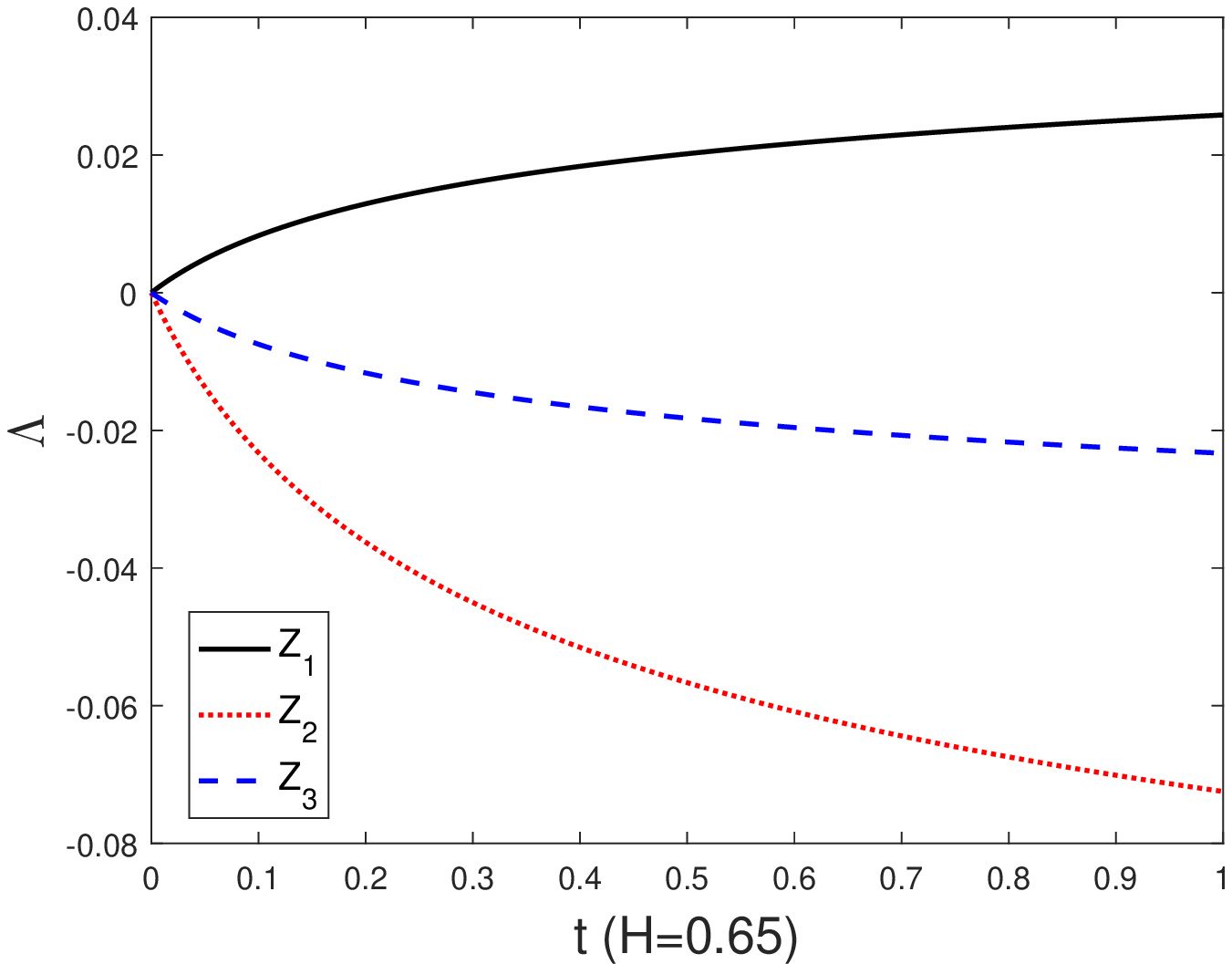}
}
\caption{Past information $\Lambda_t$ with $\varepsilon=0.1$}
\label{fig_WW}
\end{figure}

It is seen by (\ref{s1})-(\ref{s2}) and Figures \ref{fig_W} and \ref{fig_WW} that the parameter $\varepsilon$ plays an important role to determine the state of investors' memory.

When $\varepsilon$ is sufficiently small, the modified volatility $\sqrt{2H}\varepsilon^h\sigma_D$ of the output depends on Hurst index $H$. For $H<0.5$, the modified volatility of the output is extremely high. In such an economy, no matter how the historical returns of the output behave, investors can not trust the extremely volatile market. Panel (a) in Figure \ref{fig_W} (with parameters $H=0.35$ and $\varepsilon=10^{-5}$) shows that the past information is a kind of ``bad memory" (i.e. $\Lambda_t<0$) for shareholders in such an economy. For $H>0.5$, the modified volatility of the output is extremely low. In such an economy, investors would focus on wether the realized returns of the output exceed the mean returns of the output. Panel (c) in Figure \ref{fig_W} (with parameters $H=0.65$ and $\varepsilon=10^{-5}$) shows that the past information is a kind of ``good memory" (i.e. $\Lambda_t>0$) for shareholders when the returns of the output behave better than their expectations (i.e. $Z=Z_1$), while the past information is a kind of ``bad memory" for shareholders when the returns of the output behave worse than their expectations (i.e. $Z=Z_2$ and $Z=Z_3$).

When $\varepsilon$ is not sufficiently small, the modified volatility $\sqrt{2H}\varepsilon^h\sigma_D$ of the output is neither extremely high nor extremely low. In such an economy, investors should take both returns and volatilities into account by making use of the historical data. Both panel (a) (with parameters $H=0.35$, $\varepsilon=0.1$, $Z=Z_2$ and $Z=Z_3$) and panel (c) (with parameters $H=0.65$, $\varepsilon=0.1$, $Z=Z_1$ ) in Figure \ref{fig_WW} show that it is worth for investors risking volatilities to benefit from returns, which indicates that the past information is a kind of ``good memory" for shareholders. Both panel (a) (with parameters $H=0.35$, $\varepsilon=0.1$, $Z=Z_1$) and panel (c) (with parameters $H=0.65$, $\varepsilon=0.1$, $Z=Z_2$ and $Z=Z_3$) in Figure \ref{fig_WW} also show that it is unworthy for investors enduring risks to gain returns, which indicates that the past information is a kind of ``bad memory" for shareholders.

Specially, when $\varepsilon$ disappears in the economy (i.e. the case of $H=0.5$), our economy is driven by Brownian motion with ``no memory'' properties. In such an economy, whatever the historical realized  data $Z$ behave, investors can not get any useful information from the historical realized data.  Panel (b) in Figure \ref{fig_W} (with parameters $H=0.5$ and $\varepsilon=10^{-5}$) and panel (b) in Figure \ref{fig_WW} (with parameters $H=0.5$ and $\varepsilon=0.1$) present that the past information is a kind of ``no memory" for shareholders.

From the numerical analysis of past information, we conclude that the process $\Lambda_t$ is the key character in our approximate fractional economies with different Hurst indices $H$. Furthermore, the market with $\Lambda_t=0$ at time $t>0$ is almost equivalent to a classic economy with instantaneous return $\mu_D$ and volatility $\sqrt{2H}\varepsilon^h\sigma_D$ of the output. Then by considering different scenarios of $\Lambda_t$, not only
can we depict different past information, but a similar classic economy is also included as a comparative group. Thus, we set $H=0.65$ and consider different scenarios of $\Lambda_t$ with $\Lambda_t=-5,\Lambda_t=0,\Lambda_t=5$ for following numerical analysis, where the extreme cases of good memory and bad memory are set on purpose to ensure that different lines in the following figures can be distinguished by the naked eye.

\subsection{Stock holdings and diverted output}
Controlling shareholder's stock holding $n_{Ct}^*$ is the basic and crucial result in equilibrium. So we start our numerical analysis in equilibrium with Figures \ref{fig_stockL} and \ref{fig_stockp} demonstrating the effects of investor protection and past information on controlling shareholder's equilibrium stock holding.

\begin{figure}[!htbp]
\centering
\subfigure[$\Lambda_t=0$]{
\includegraphics[ width = 0.31\textwidth]{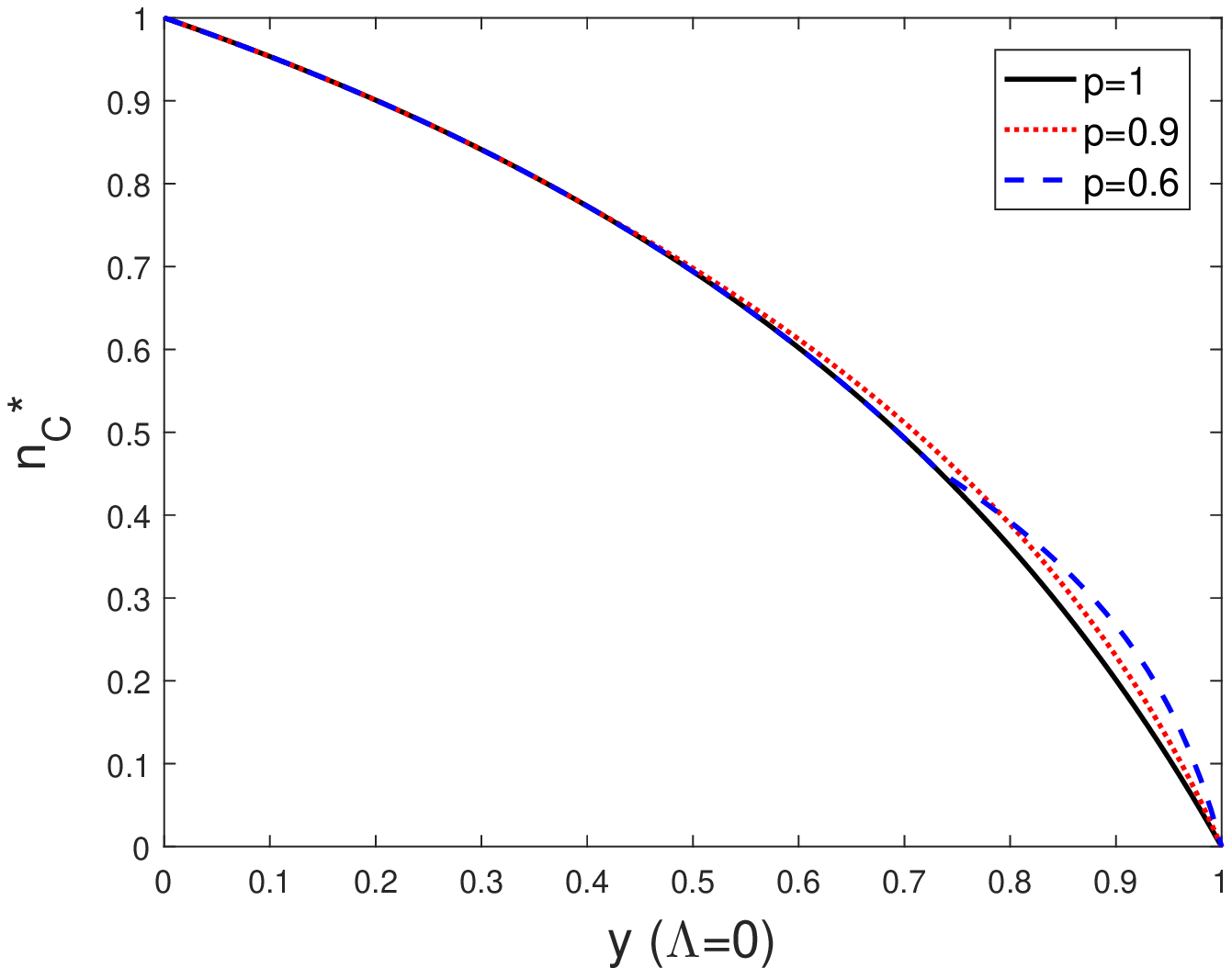}
}
\subfigure[$\Lambda_t=-5$]{
\includegraphics[ width = 0.31\textwidth]{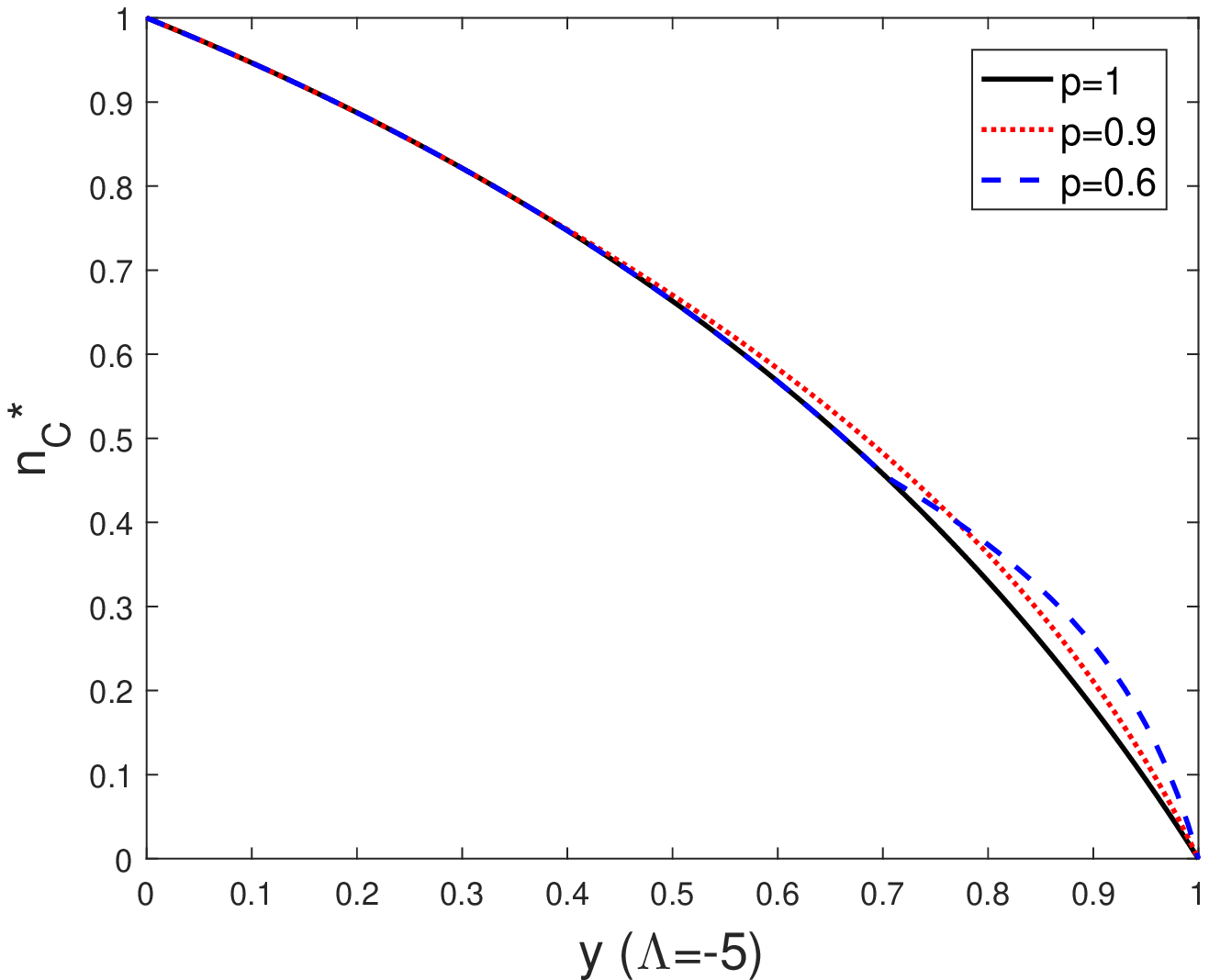}
}
\subfigure[$\Lambda_t=5$]{
\includegraphics[ width = 0.31\textwidth]{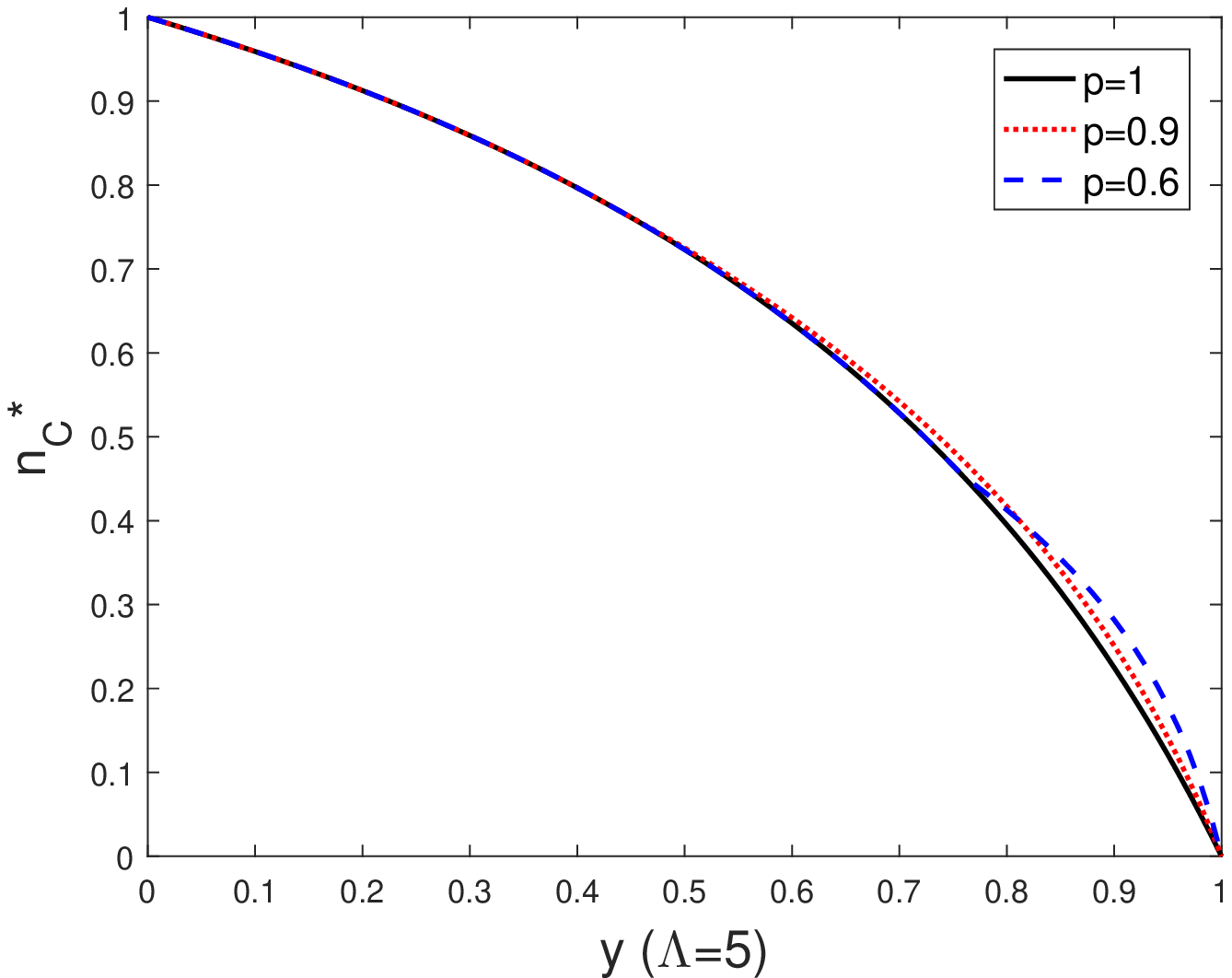}
}
\caption{Stock holdings $n_{Ct}^*$ with different  investor protect}
\label{fig_stockL}
\end{figure}
\begin{figure}[!htbp]
\centering
\subfigure[$p=1$]{
\includegraphics[ width = 0.31\textwidth]{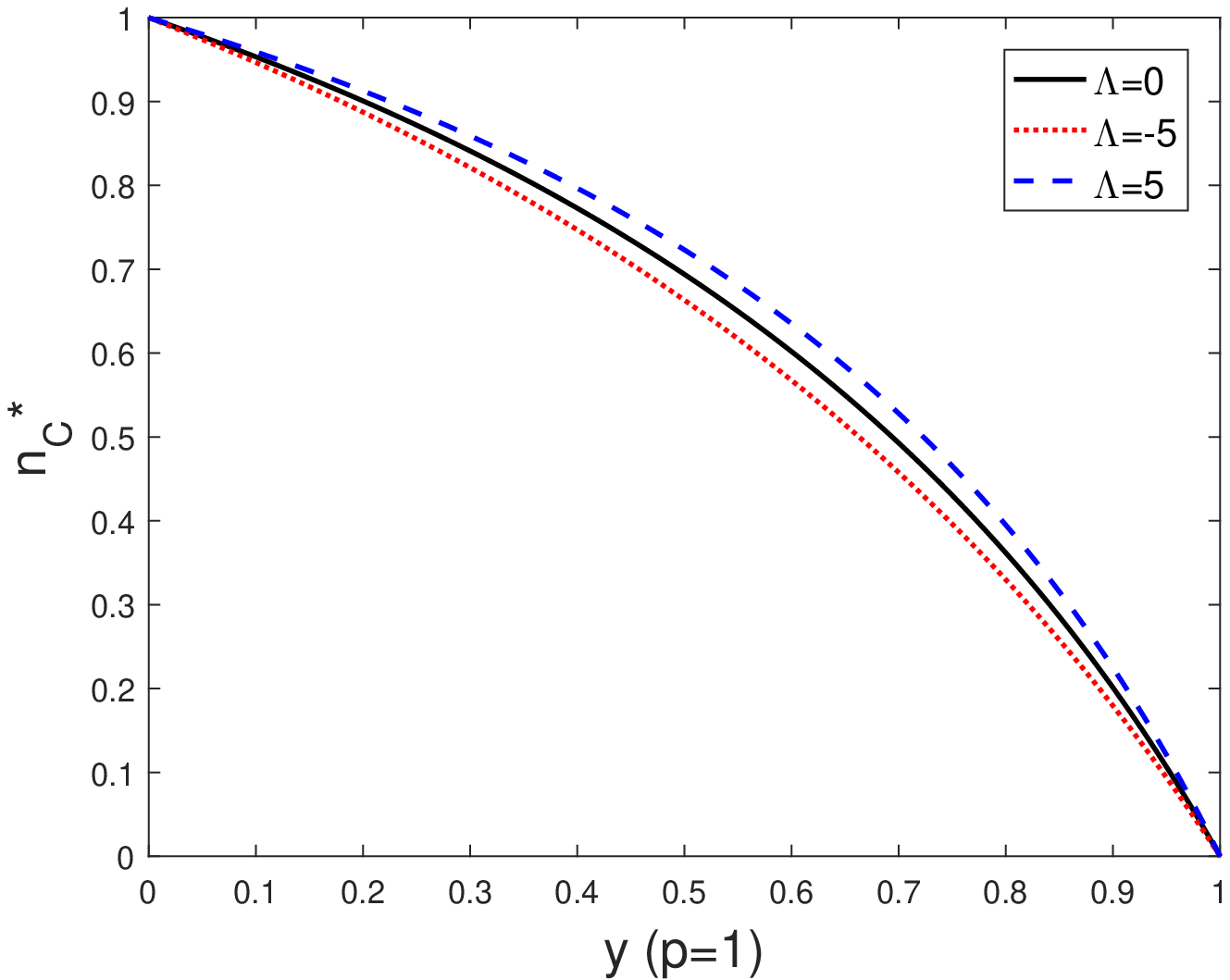}
}
\subfigure[$p=0.9$]{
\includegraphics[ width = 0.31\textwidth]{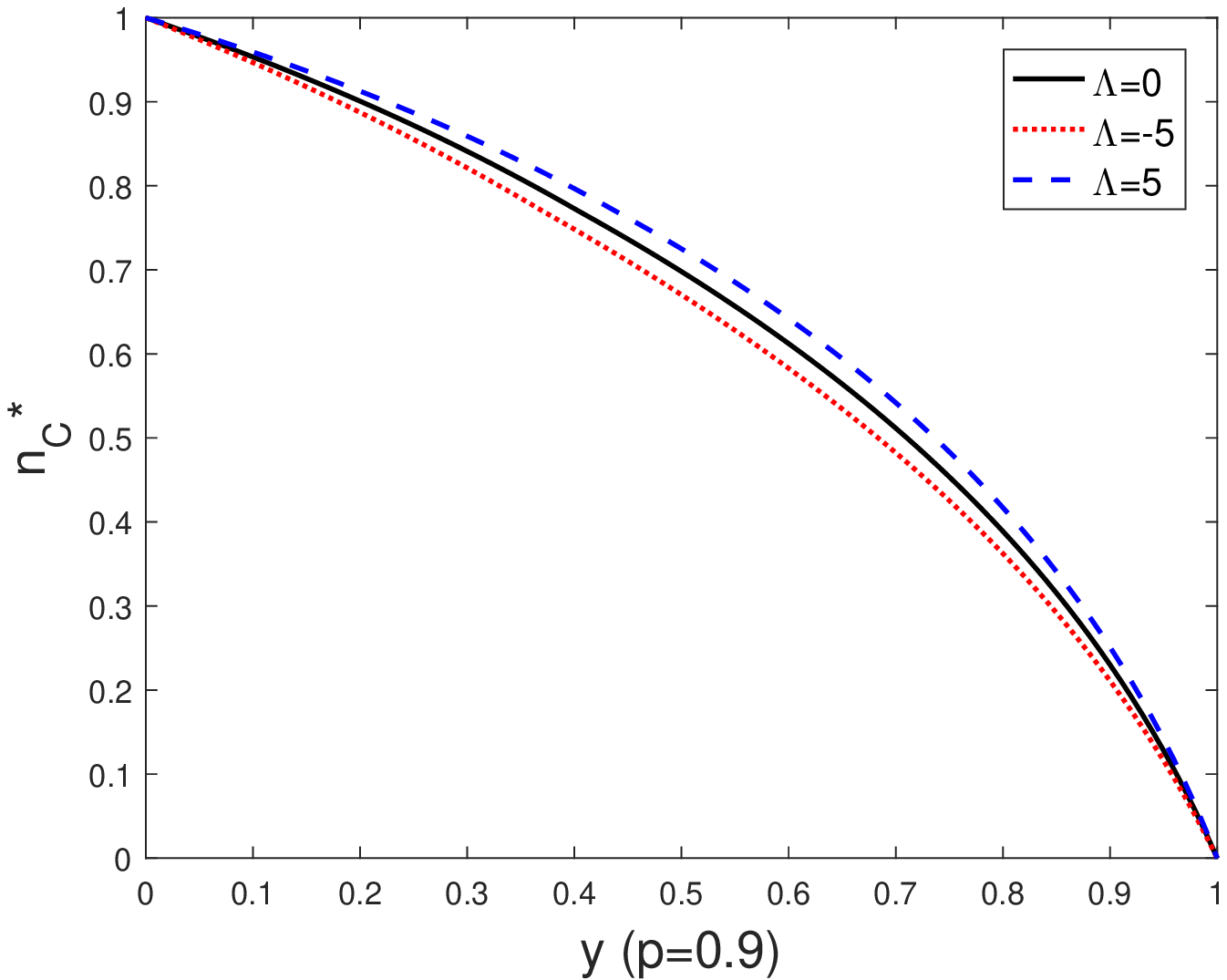}
}
\subfigure[$p=0.6$]{
\includegraphics[ width = 0.31\textwidth]{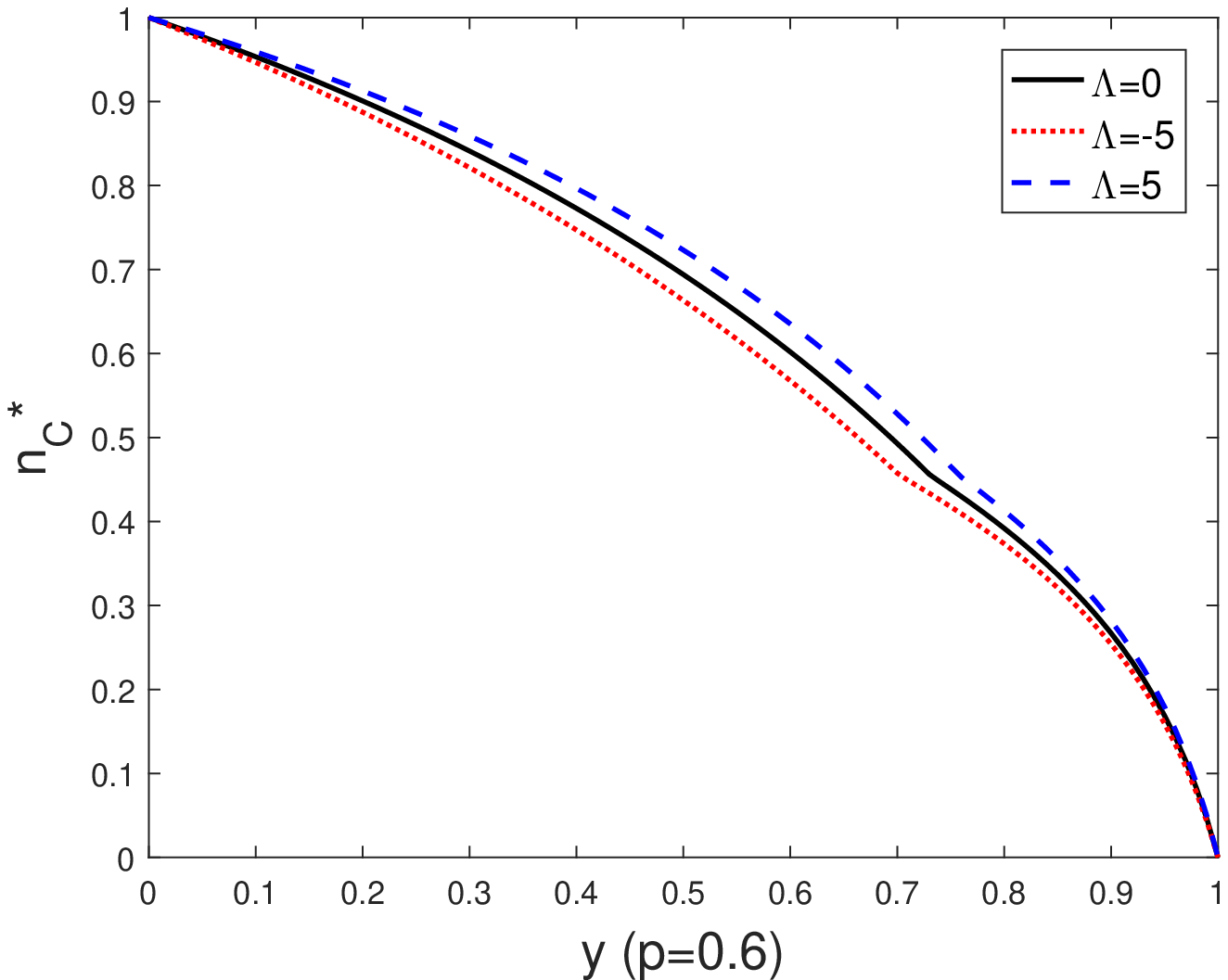}
}
\caption{Stock holdings $n_{Ct}^*$ with different  past information}
\label{fig_stockp}
\end{figure}

Figure \ref{fig_stockL} shows that no matter which kind of memory the past information brings to investors, poorer investor protection leads to higher stock holding of controlling shareholder relative to the full protection benchmark, which is consistent with panel (a) of Figure 2 in \cite{Basak}. Furthermore, Figure \ref{fig_stockL} also indicates more specifically that, with $y_t$ increasing, $n_C^*$ stays in Region 2 when the controlling shareholder's consumption share is high ($y_t$ is low), while $n_C^*$ switches to Region 1 when the controlling shareholder's consumption share is low ($y_t$ is high). In the case of Region 2, holding lots of shares reduces the benefits of diverting the output. It is from (\ref{co2nc2}) that $p$ does not involve, and then the investor protection constraint does not bind and the diversion of output is tempered only by the cost of stealing. Independent of the stock holding and past information, the cost parameter $k$ makes it helpless for the controlling shareholder to divert more output though buying shares beyond the benchmark level. Consequently, $n_C^*=n_C^\mathfrak{B}$ in Region 2. In the case of Region 1, the investor protection constraint binds, and then the controlling shareholder can divert a larger fraction of output when he owns more shares in the case of imperfect protection, which gives the controlling shareholder an incentive to acquire more shares. Hence, benefits of diverting output, rather than the inconsistency of investors' sensitivity to past information, dominates the change of controlling shareholder's stock holding, leading to $n_C^*\geq n_C^\mathfrak{B}$ in Region 1.

Figure \ref{fig_stockp} depicts that compared with classic economies $\Lambda_t=0$, the controlling shareholder's ``good memory" $\Lambda_t=5$ (``bad memory"  $\Lambda_t=-5$) urges the controlling shareholder to acquire more (fewer) shares whether or not the investor protection is perfect. The result coincides with empirical evidence in \cite{Guidolin}
and the reason is two-fold. Advantageous (disadvantageous) past information would prompt (reduce) the expectation of shareholders' future investment as we have explained in Remark \ref{memory}. Besides, since the condition $\alpha_M\geq\alpha_C$ implies the controlling shareholder is more sensitive to the past information than the minority shareholder, the increment (decrement) of shares caused by the controlling shareholder's ``good memory" (``bad memory") is greater than the one caused by minority shareholder's ``good memory" (``bad memory").

It is from the proof of Theorem \ref{th1} that the fraction of diverted output $x_t^*$ is directly related to controlling shareholder's stock holding, which contributes to understanding controlling shareholder's stock holding more precisely.  Figures \ref{fig_divertL} and \ref{fig_divertp} depict effects of investor protection and past information on the equilibrium fraction of diverted output.

\begin{figure}[!htbp]
\centering
\subfigure[$\Lambda_t=0$]{
\includegraphics[ width = 0.31\textwidth]{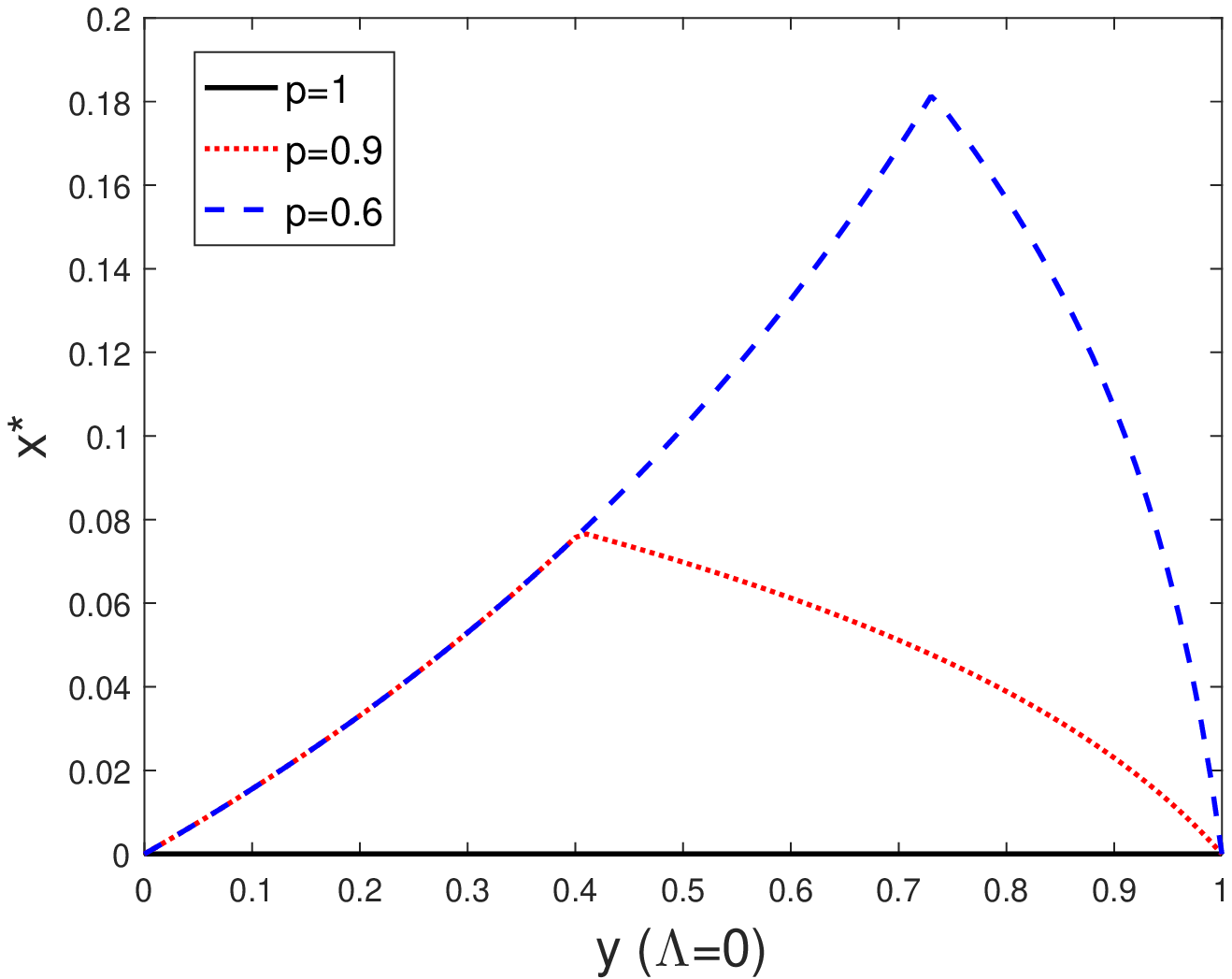}
}
\subfigure[$\Lambda_t=-5$]{
\includegraphics[ width = 0.31\textwidth]{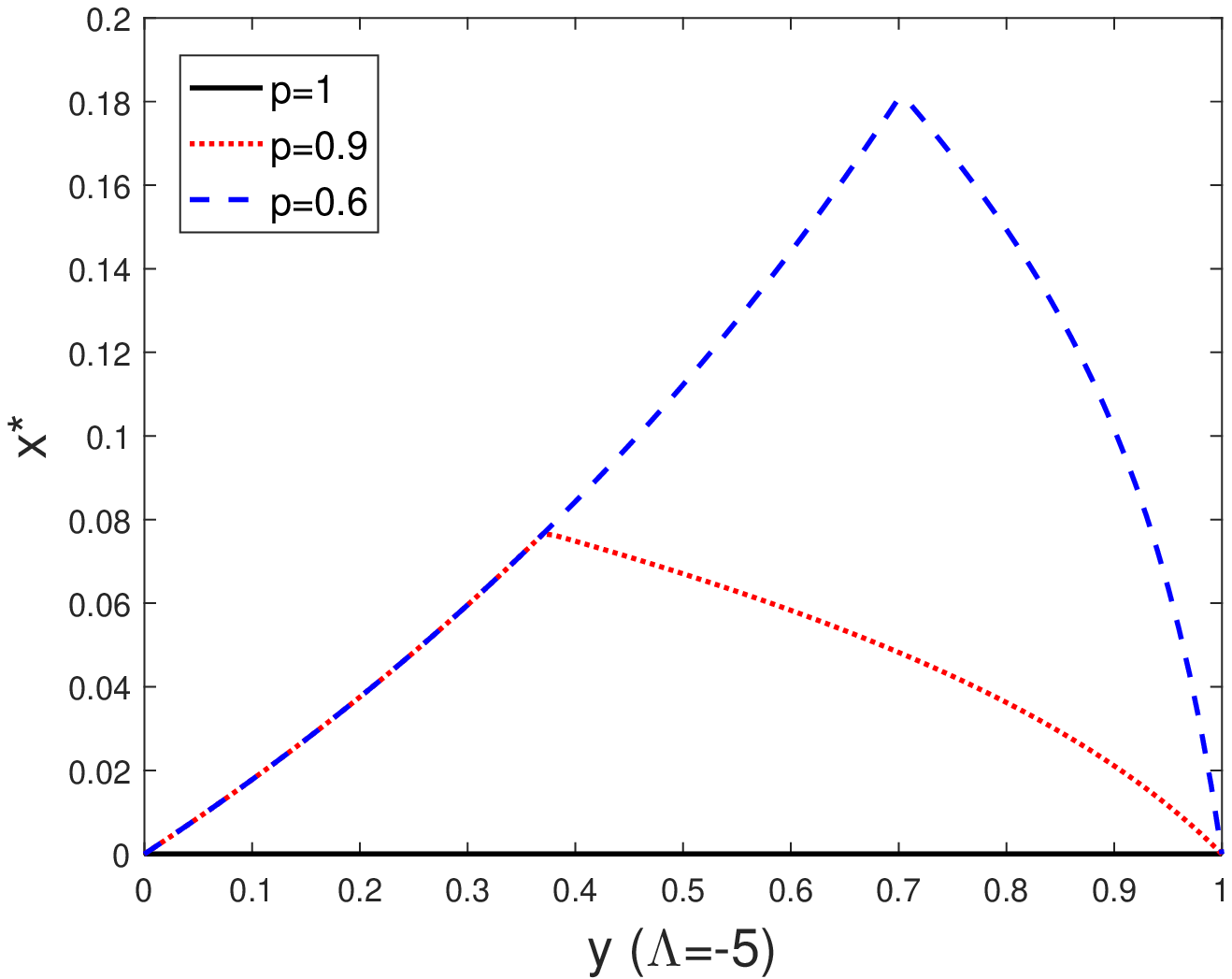}
}
\subfigure[$\Lambda_t=5$]{
\includegraphics[ width = 0.31\textwidth]{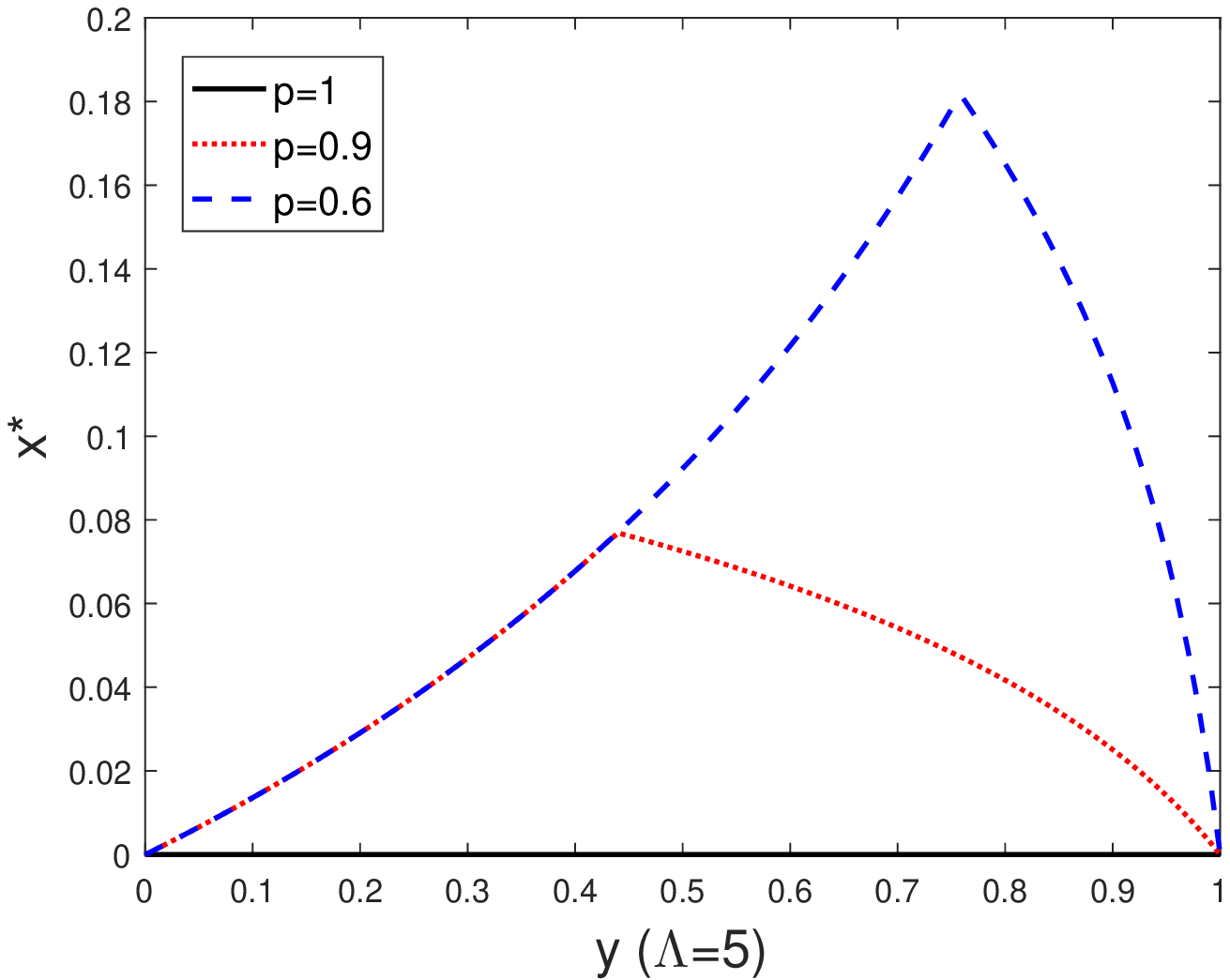}
}
\caption{The fraction of diverted output $x_t^*$ with different investor protect}
\label{fig_divertL}
\end{figure}
\begin{figure}[!htbp]
\centering
\subfigure[$p=1$]{
\includegraphics[ width = 0.31\textwidth]{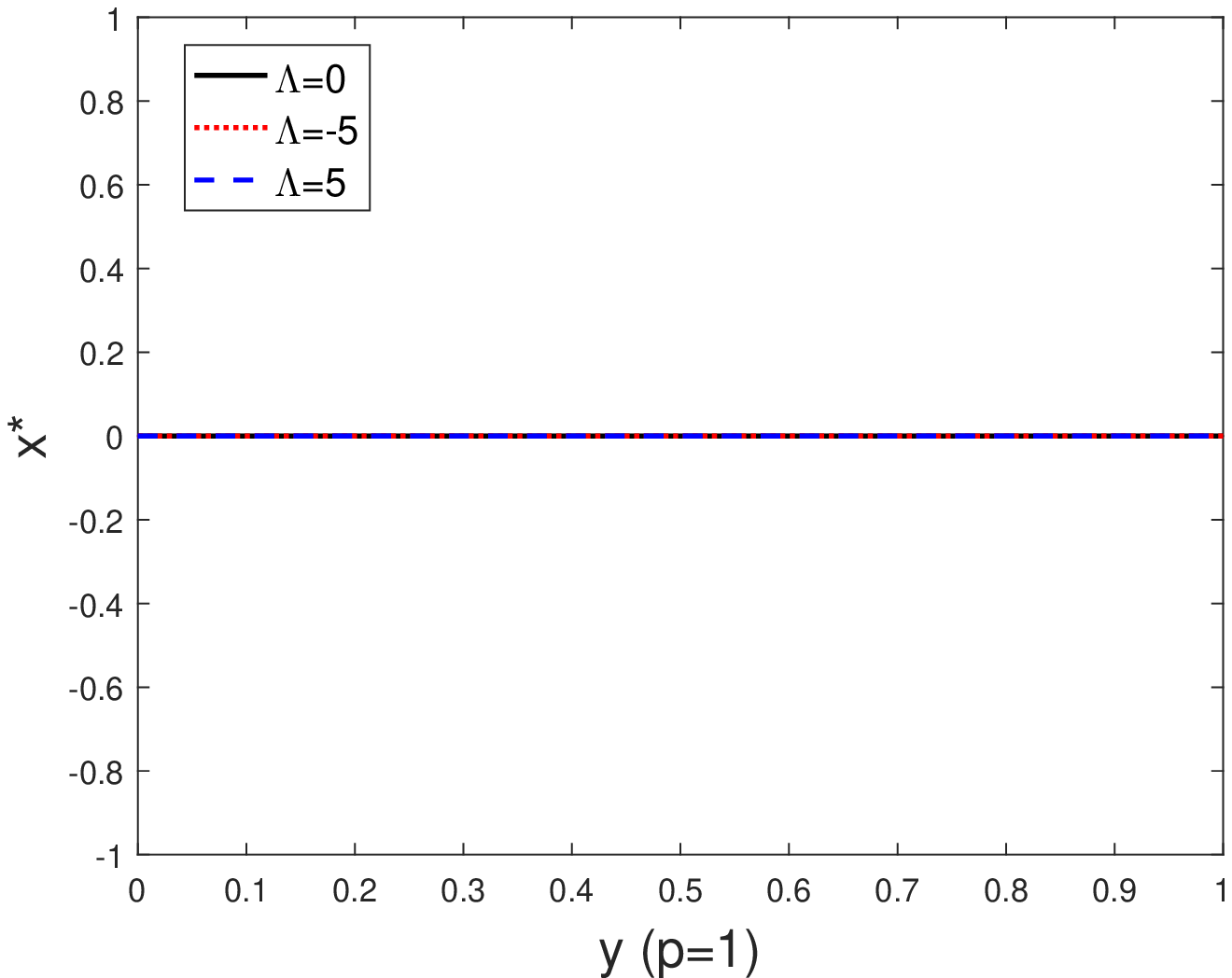}
}
\subfigure[$p=0.9$]{
\includegraphics[ width = 0.31\textwidth]{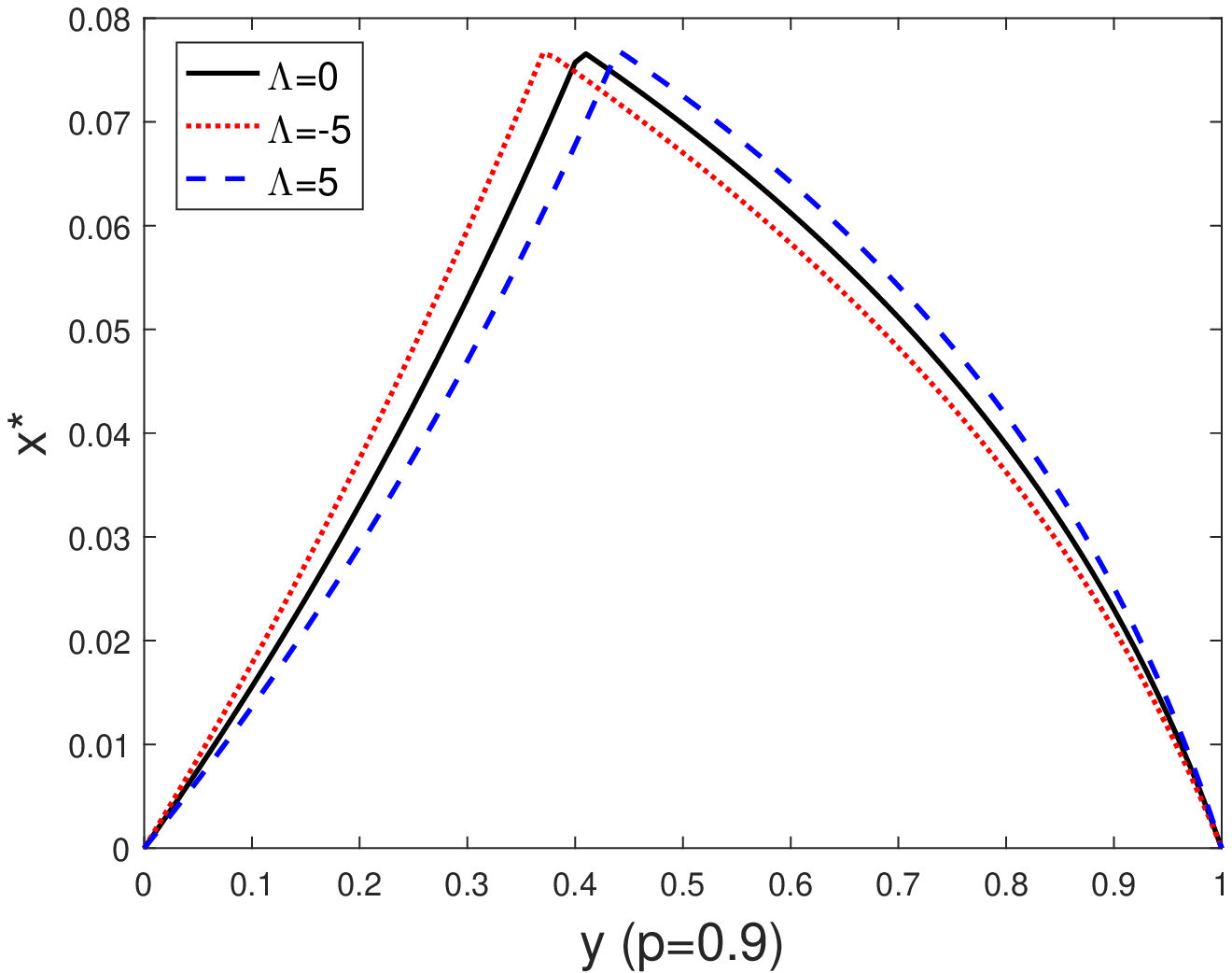}
}
\subfigure[$p=0.6$]{
\includegraphics[ width = 0.31\textwidth]{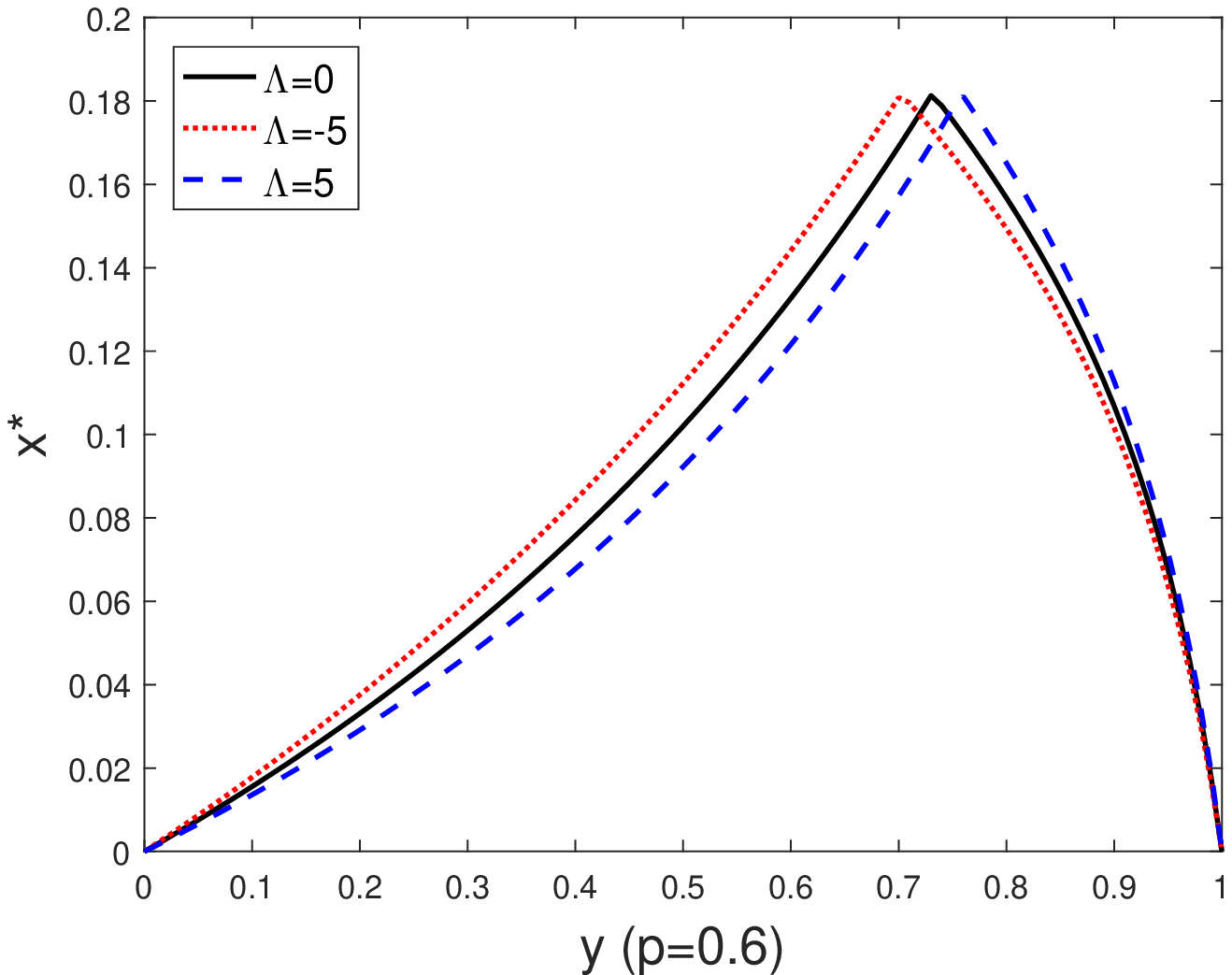}
}
\caption{The fraction of diverted output $x_t^*$ with different  past information}
\label{fig_divertp}
\end{figure}

Figure \ref{fig_divertL} presents that no matter which kind of memory investors have, poorer investor protection leads to a higher fraction of diverted output, which is consistent with our common sense. Compared with Figures \ref{fig_stockL} and \ref{fig_stockp}, the kinks in Figures \ref{fig_divertL} and \ref{fig_divertp} make it much easier to verify when $n_{Ct}^*$ lies in Region 1 or Region 2. As discussed in \cite{Basak}, the kink is the separation of point of Region 1 and Region 2: for sufficiently small $n_{Ct}$, the constraint $x_t\leq (1-p)n_{Ct}$ is binding and the controlling shareholder would divert more output by acquiring more shares, while for sufficiently large $n_{Ct}$, the constraint is not binding and the controlling shareholder can only divert the output through his stock holding at the benchmark level. Hence, if $y_t$ lies to the left-side (right-side) of the kink, then $n_{Ct}^*$ stays in Region 2 (Region 1) for imperfect protection.

Compared with classic economies, the effect of past information on the fraction of diverted output for imperfect protection could be considered by Figure \ref{fig_divertp} in two ways.  In Region 2, good memory (bad memory) leads to a lower (higher) fraction of diverted output. This is because controlling holder's stock holding increases (decreases) due to their good memory (bad memory) as we claimed in Figure \ref{fig_stockp}, and then the fraction of diverted output $x_t^*=\frac{1-n_{Ct}^*}{k}$ decreases (increase). In Region 1, good memory (bad memory) on the contrast leads to a higher (lower) fraction of diverted output, which is because good memory (bad memory) urges controlling shareholder to acquire more (fewer) shares so that the fraction of diverted output $x_t^*=(1-p)n_{Ct}^*$ increases (decreases).
Hence, Figure \ref{fig_divertp} also states that good memory (bad memory) would strengthen (weaken) investor protection for minority shareholder when the ownership concentration is sufficiently high, while good memory (bad memory) would inversely weaken (strengthen) investor protection for minority shareholder when the ownership concentration is sufficiently low.

Observing (\ref{thnC1})-(\ref{thnC4}), it is clear that controlling shareholder's stock holding $n_{Ct}^*$ critically depends on investor protect $p$ and past information $\Lambda_t$. However, $n_{Ct}^*$ in general is non-monotone in $p$ or $\Lambda_t$.  The investor protection parameter $p$ has two opposing influences on  stock holding $n_{Ct}^*$. Poor investor protection urges the controlling shareholder to divert the output, while poor investor protection extends Region 2 for imperfect protection (Figure \ref{fig_divertL}), loosing the role of the cost of stealing. Similarly, past information $\Lambda_t$ affects stock holding $n_{Ct}^*$ in different ways. Firstly, good memory (bad memory) urges the controlling shareholder to acquire more (fewer) shares. Secondly, good memory (bad memory) also urges minority shareholder to acquire more (fewer) shares, which on the contrary decreases (increases) controlling shareholder's shares. Last, it is from Figure \ref{fig_divertL} that good memory (bad memory) would extend (reduce) Region 2 for imperfect protection.

\subsection{Stock return and volatility}
Similar to discussion of  the output $\widehat{D}_t$ in Section \ref{section2} , it is easy to see that the modified stock return and volatility are $\mu^H_t=\mu_t+\sigma_t\Lambda_t$ and $\sigma_t^H=\sqrt{2H}\varepsilon^h\sigma_{t}$. For the purpose of comparison, we set the gross stock return as $\mu^G_t=\mu_t+\sigma_t\Lambda_t+(1-x_t^*)\frac{D_t}{S_t}$

\begin{figure}[!htbp]
\centering
\subfigure[$\Lambda_t=0$]{
\includegraphics[ width = 0.31\textwidth]{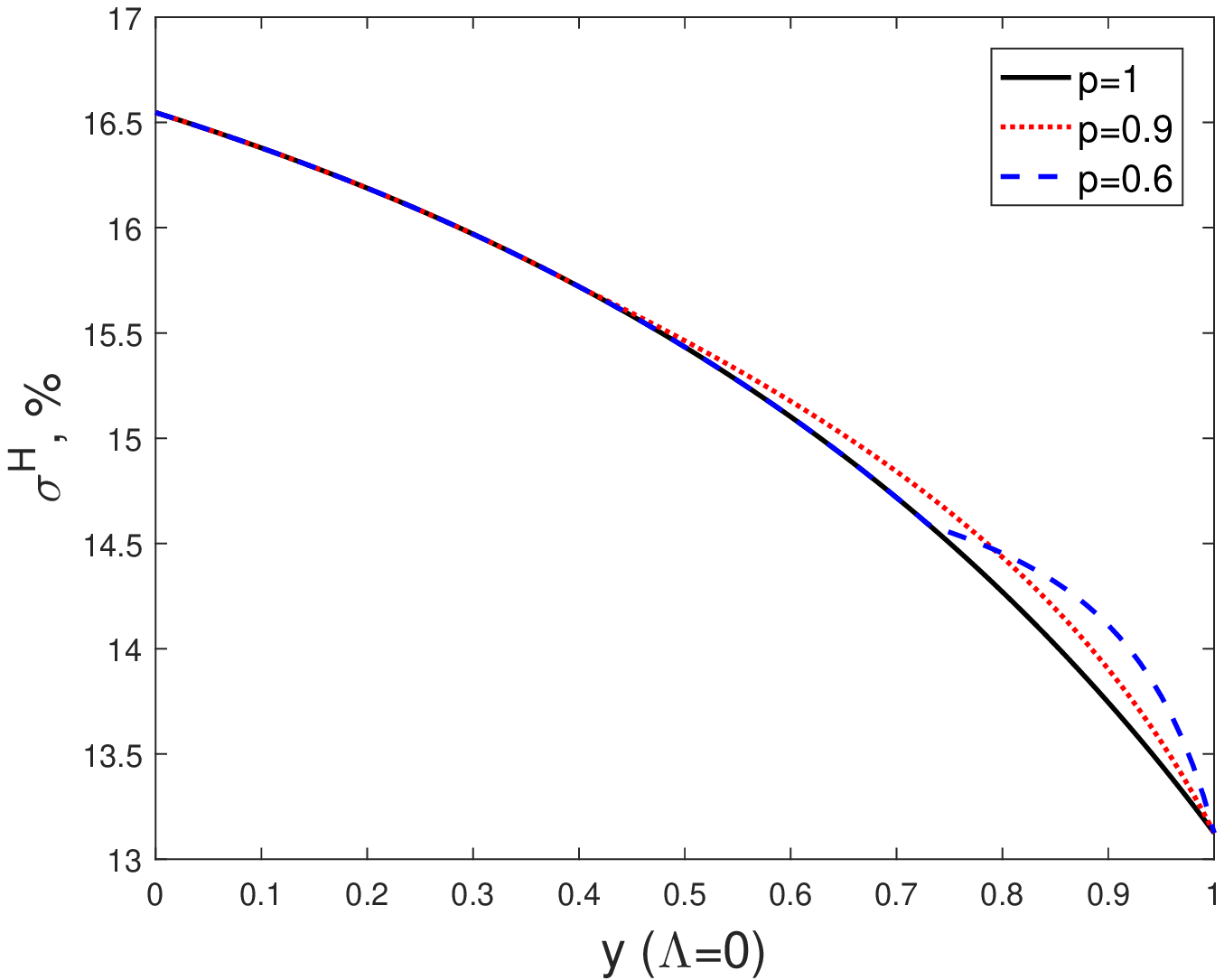}
}
\subfigure[$\Lambda_t=-5$]{
\includegraphics[ width = 0.31\textwidth]{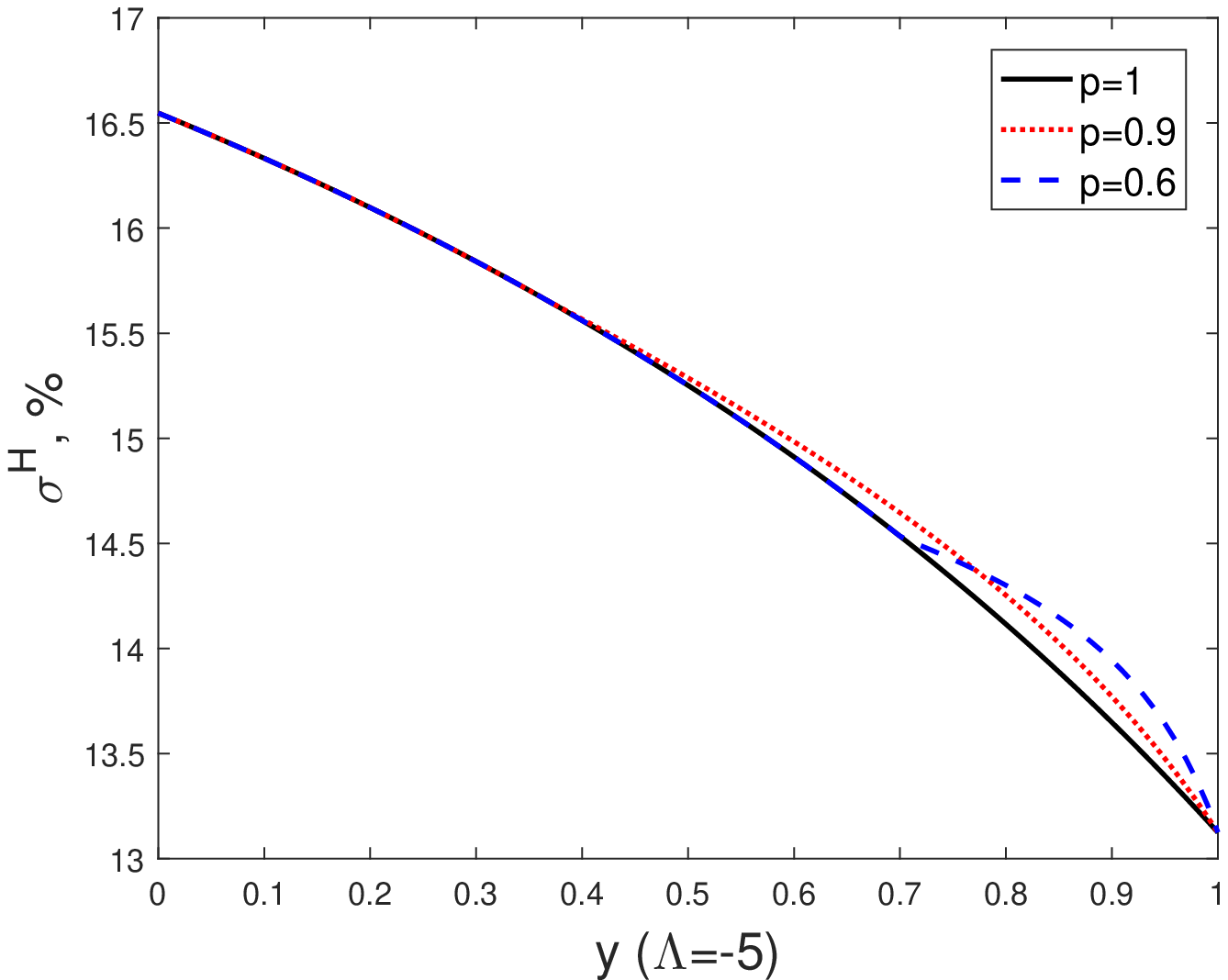}
}
\subfigure[$\Lambda_t=5$]{
\includegraphics[ width = 0.31\textwidth]{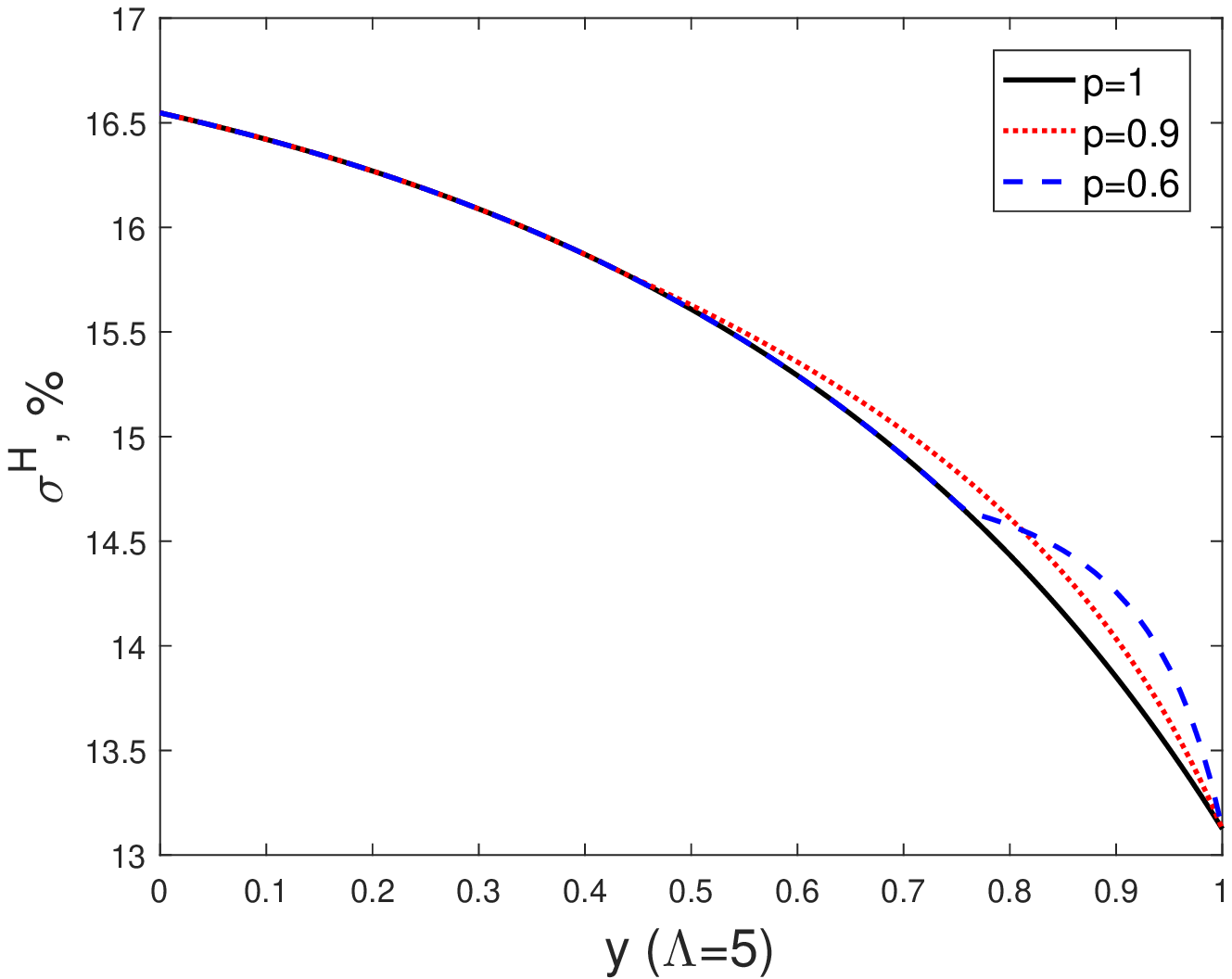}
}
\caption{Modified stock volatilities $\sigma^H$ with different investor protection}
\label{fig_volL}
\end{figure}
\begin{figure}[!htbp]
\centering
\subfigure[$p=1$]{
\includegraphics[ width = 0.31\textwidth]{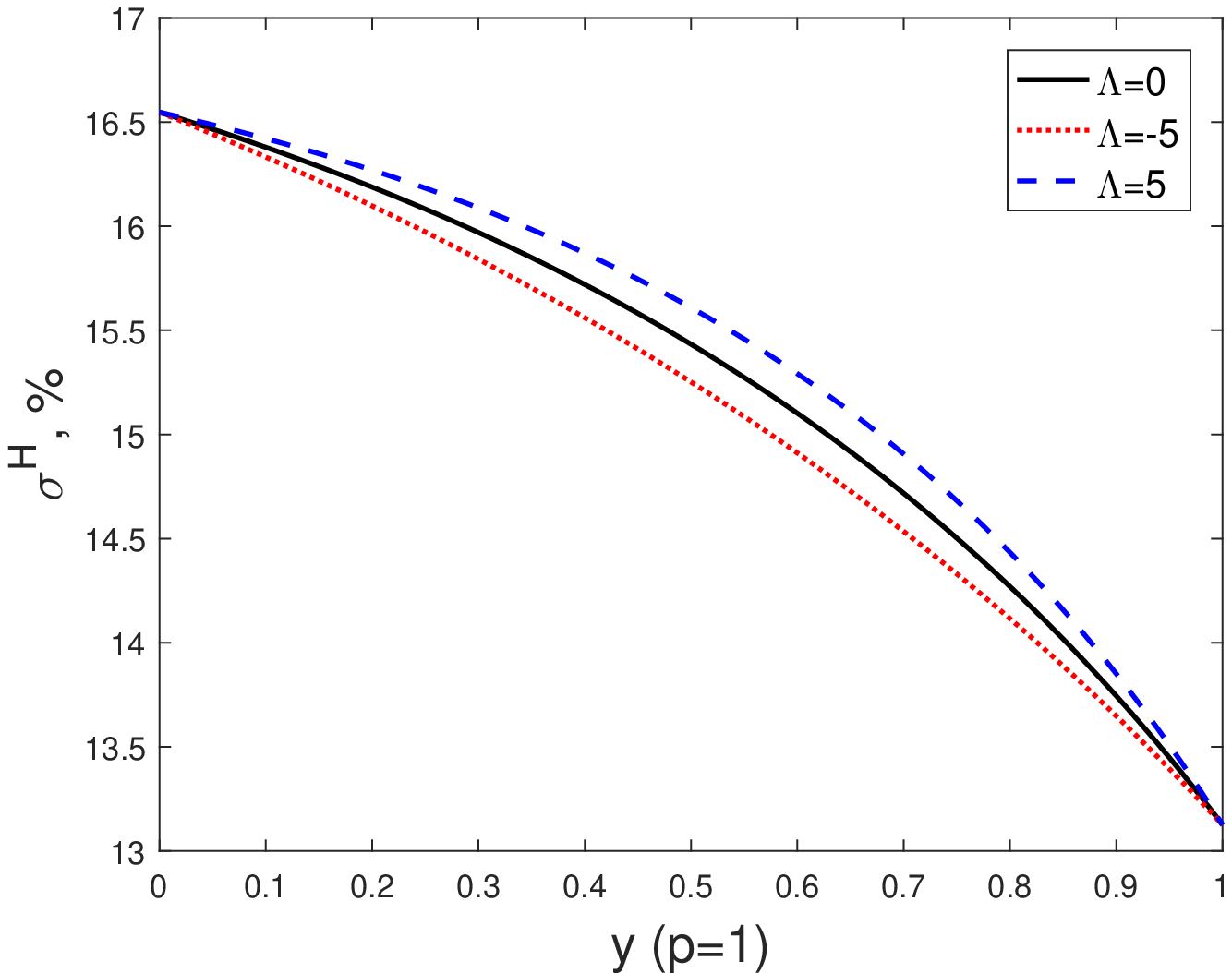}
}
\subfigure[$p=0.9$]{
\includegraphics[ width = 0.31\textwidth]{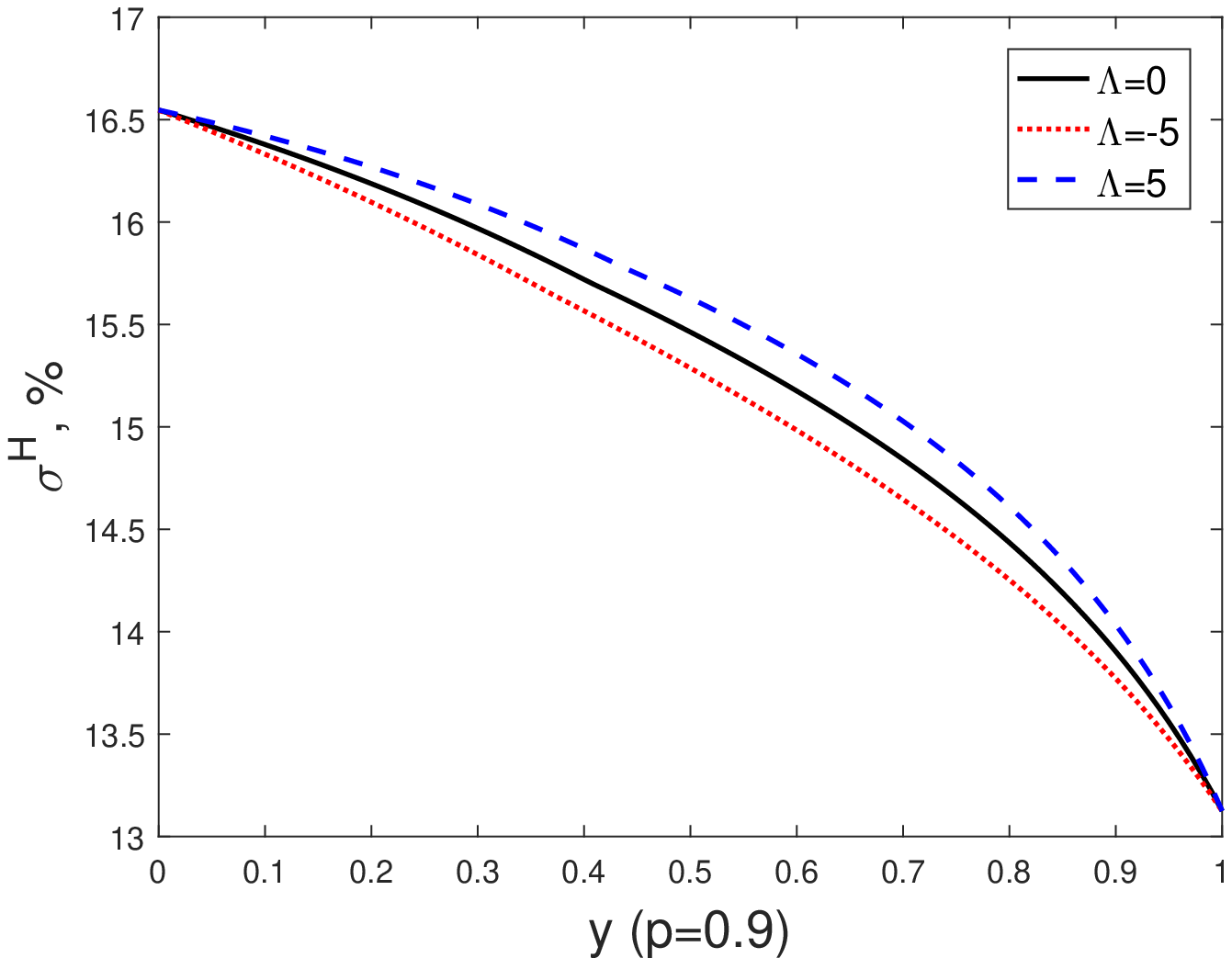}
}
\subfigure[$p=0.6$]{
\includegraphics[ width = 0.31\textwidth]{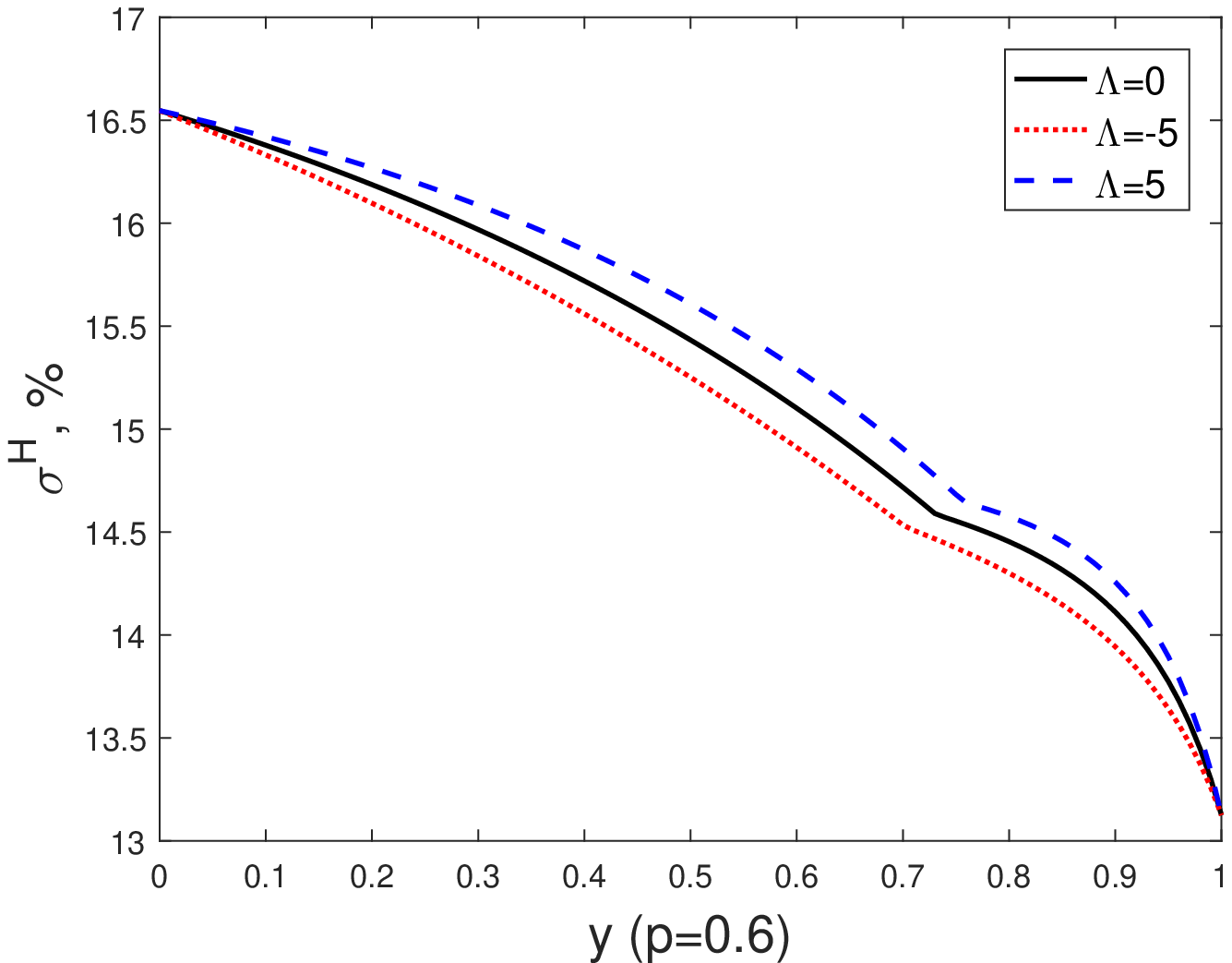}
}
\caption{Modified stock volatilities $\sigma^H$ with different past information}
\label{fig_volp}
\end{figure}

Figure \ref{fig_volL} shows that no matter which kind of memory investors have, the modified stock volatility $\sigma^H$ in the economy with imperfect protection is higher than $\sigma^{H,\mathfrak{B}}$ in the benchmark economy with perfect protection, i.e. $\sigma^H\geq\sigma^{H,\mathfrak{B}}$. In Region 2, $n_{Ct}^*=n_{Ct}^\mathfrak{B}$ from Figure \ref{fig_stockL} clearly leads to $\sigma^H=\sigma^{H,\mathfrak{B}}$ by Remark \ref{diff}. This is because the diversion of output has been tempered by the cost of stealing such that the modified stock volatility does not depend on diversion of the output and then is consistent with that in the benchmark economy. In Region 1, it is seen $\sigma^H>\sigma^{H,\mathfrak{B}}$ in general, which means that the stock is more volatile in the economy with imperfect protection. The reason may be explained in two aspects: ``with imperfect protection the controlling shareholder holds more shares than in the full protection benchmark, and hence is under-diversified'' (see \cite{Basak}), plus the controlling shareholder is less risk-averse than the minority shareholder. The under-diversified and less risk-averse properties make controlling shareholder's wealth and consumption more volatile, which then translates into the stock market.

Moreover, by panel (a) in Figure \ref{fig_volL}, there is a major difference between our classic economy discussed here and the classic economy studied in \cite{Basak} with $H=\frac{1}{2}$: $\sigma^H(0)<\sigma^H(1)$ in the former economy while $\sigma(0)=\sigma(1)=\sigma_D$ in the latter economy. In corporations, the number of controlling shareholders is much smaller than that of minority shareholders, and then the stock diversity in the case that controlling shareholders own all shares is much worse than the one in the case that minority shareholders own all shares. Hence, the difference of diversity leads to $\sigma^H(0)<\sigma^H(1)$ and the case in our classic economy seems reasonable.

Figure \ref{fig_volp} describes the effect of past information on the modified stock volatility showing that compared with classic economies, investors' good memory (bad memory) results in higher (lower) modified stock volatility in the economy with perfect protection, as well as in the economy with imperfect protection.  The reason can be explained similarly. As we have obtained that controlling shareholder's good memory (bad memory) makes him acquire more (fewer) shares, the under-diversified (diversified) property leads to higher (lower) modified stock volatility.

It is of importance to notice that some empirical evidence (see, \cite{Guidolin,Jones,Maheu}) indicates that volatilities are lower (higher) during the bull (bear) market period, which seems to contrast our results. We can explain such a contradiction in the following way. The volatility during the financial crisis does not remain high for long as showed by \cite{Schwert}, and then aforementioned empirical evidence may be not suitable for our stable equilibrium. Thus, past information has few efforts on stock volatility directly and under-diversified (or diversified) property affects stock volatility primarily.

\begin{figure}[!htbp]
\centering
\subfigure[$\Lambda_t=0$]{
\includegraphics[ width = 0.31\textwidth]{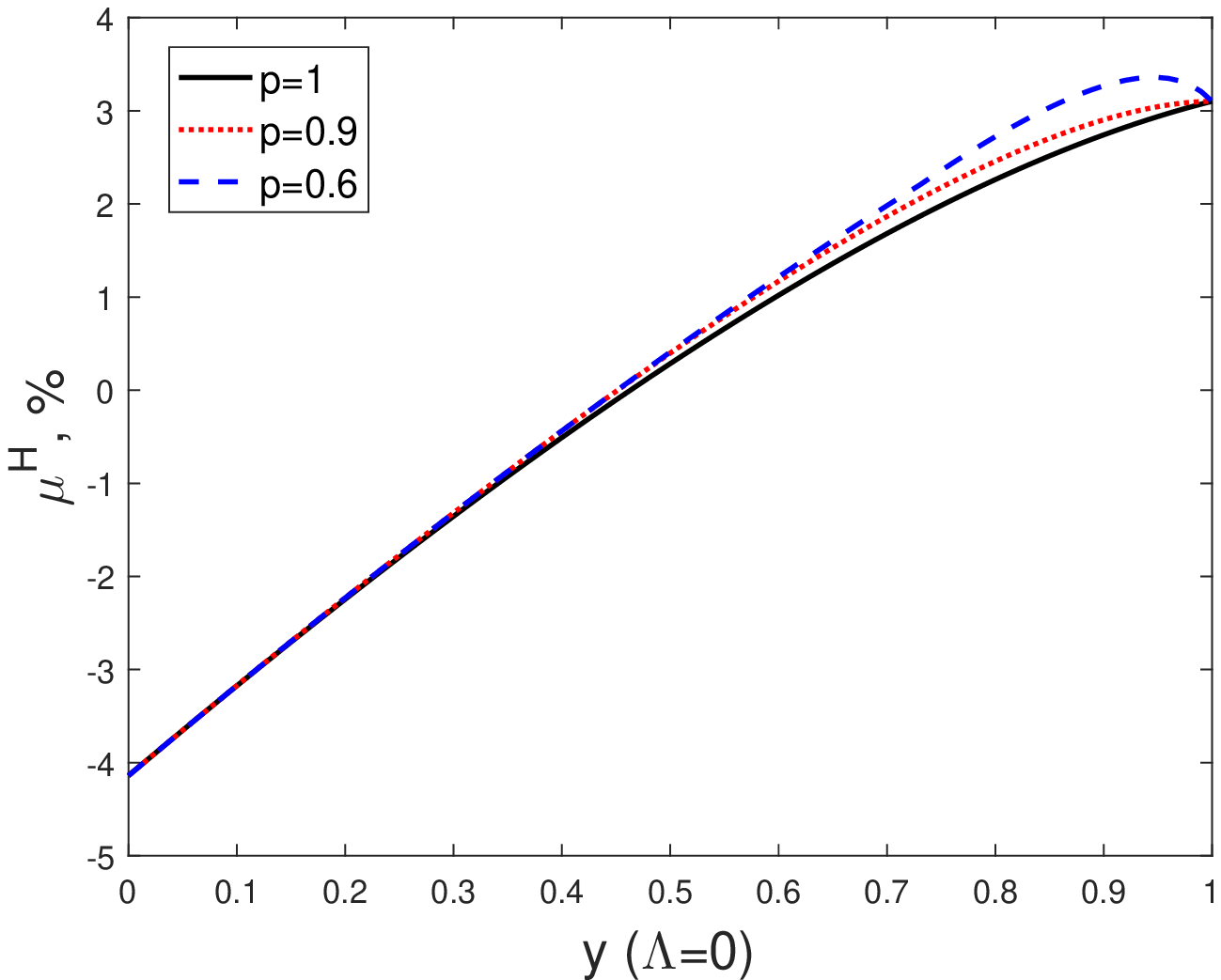}
}
\subfigure[$\Lambda_t=-5$]{
\includegraphics[ width = 0.31\textwidth]{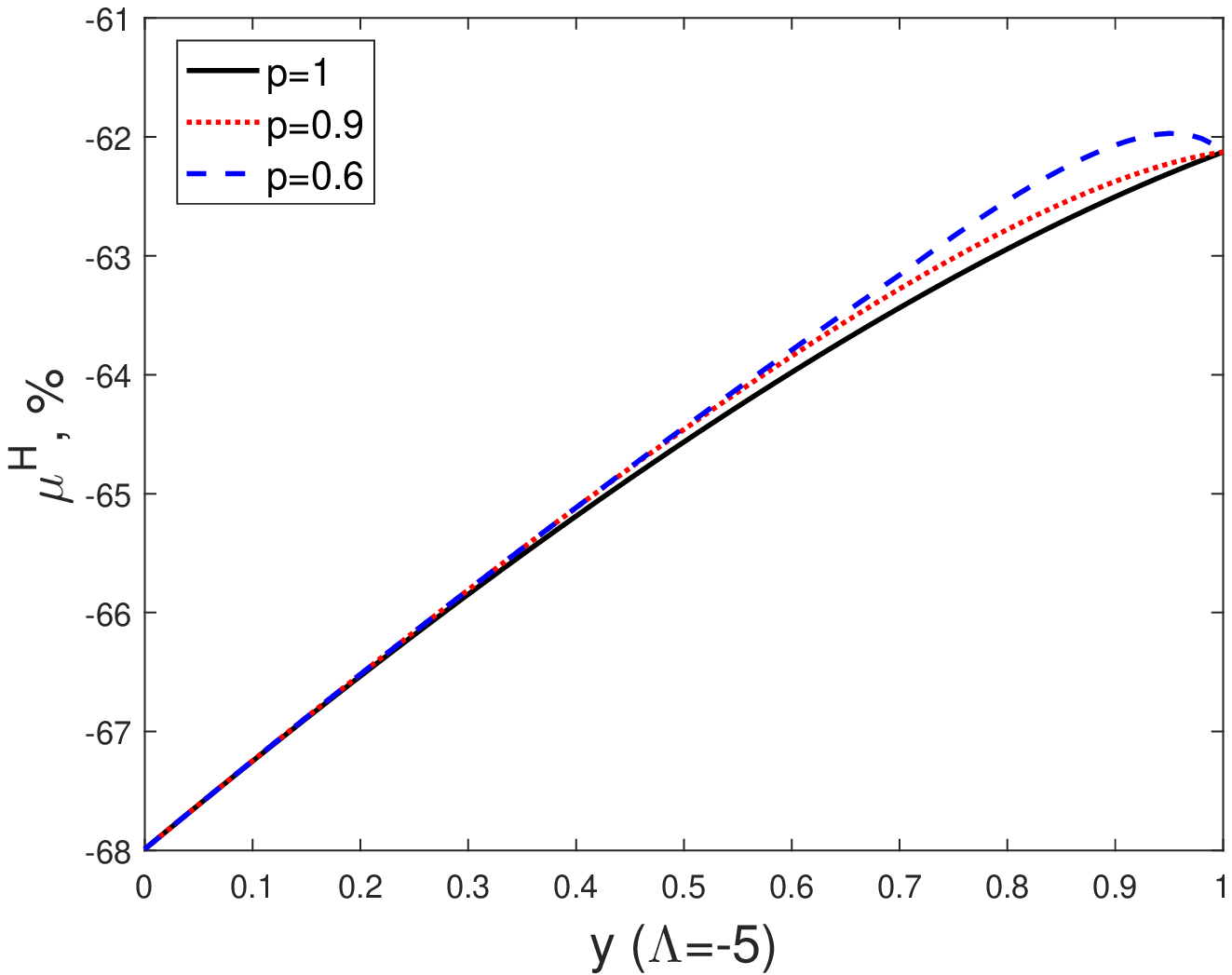}
}
\subfigure[$\Lambda_t=5$]{
\includegraphics[ width = 0.31\textwidth]{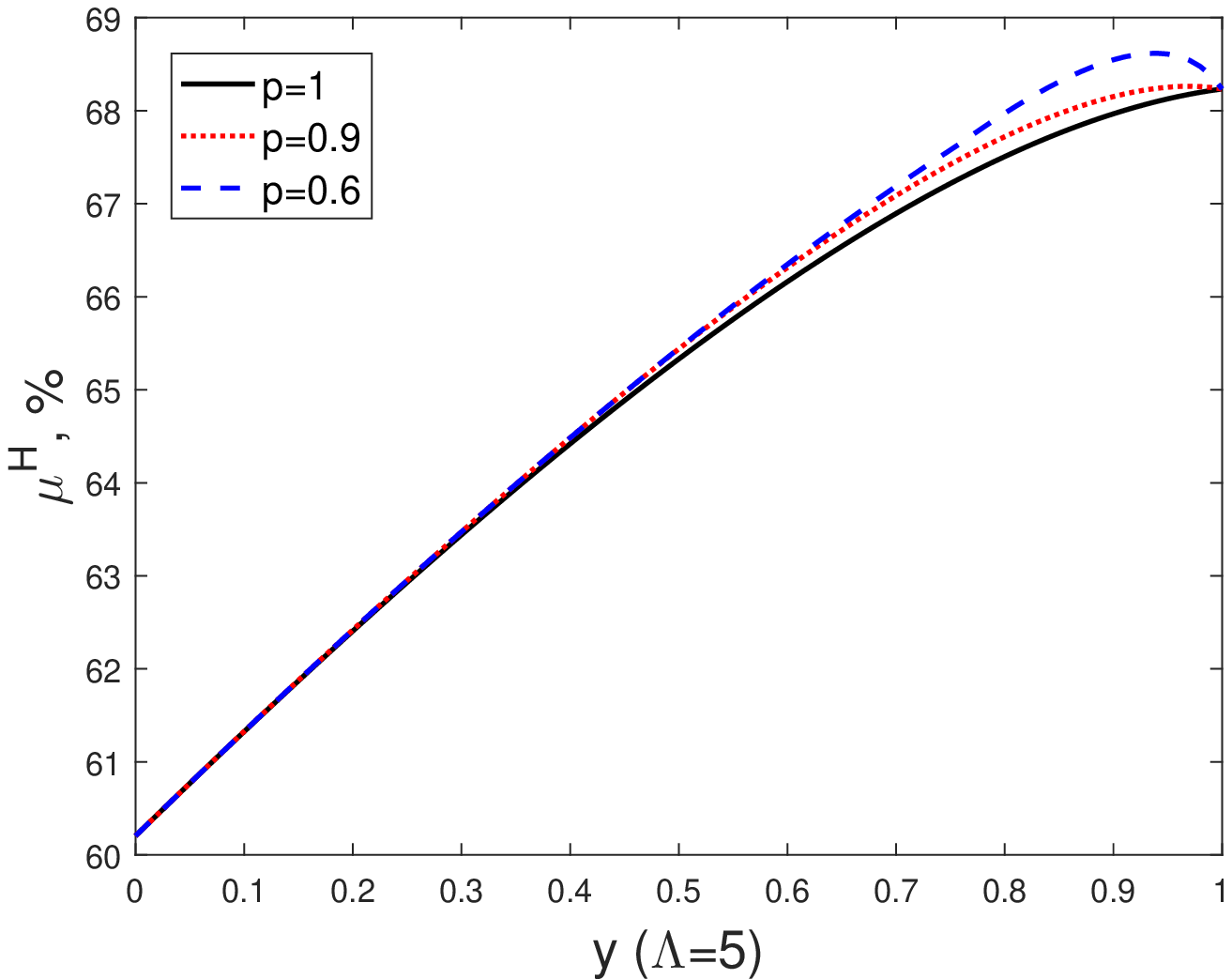}
}
\caption{Modified stock returns $\mu^H$ with different investor protection}
\label{fig_returnL}
\end{figure}
\begin{figure}[!htbp]
\centering
\subfigure[$p=1$]{
\includegraphics[ width = 0.31\textwidth]{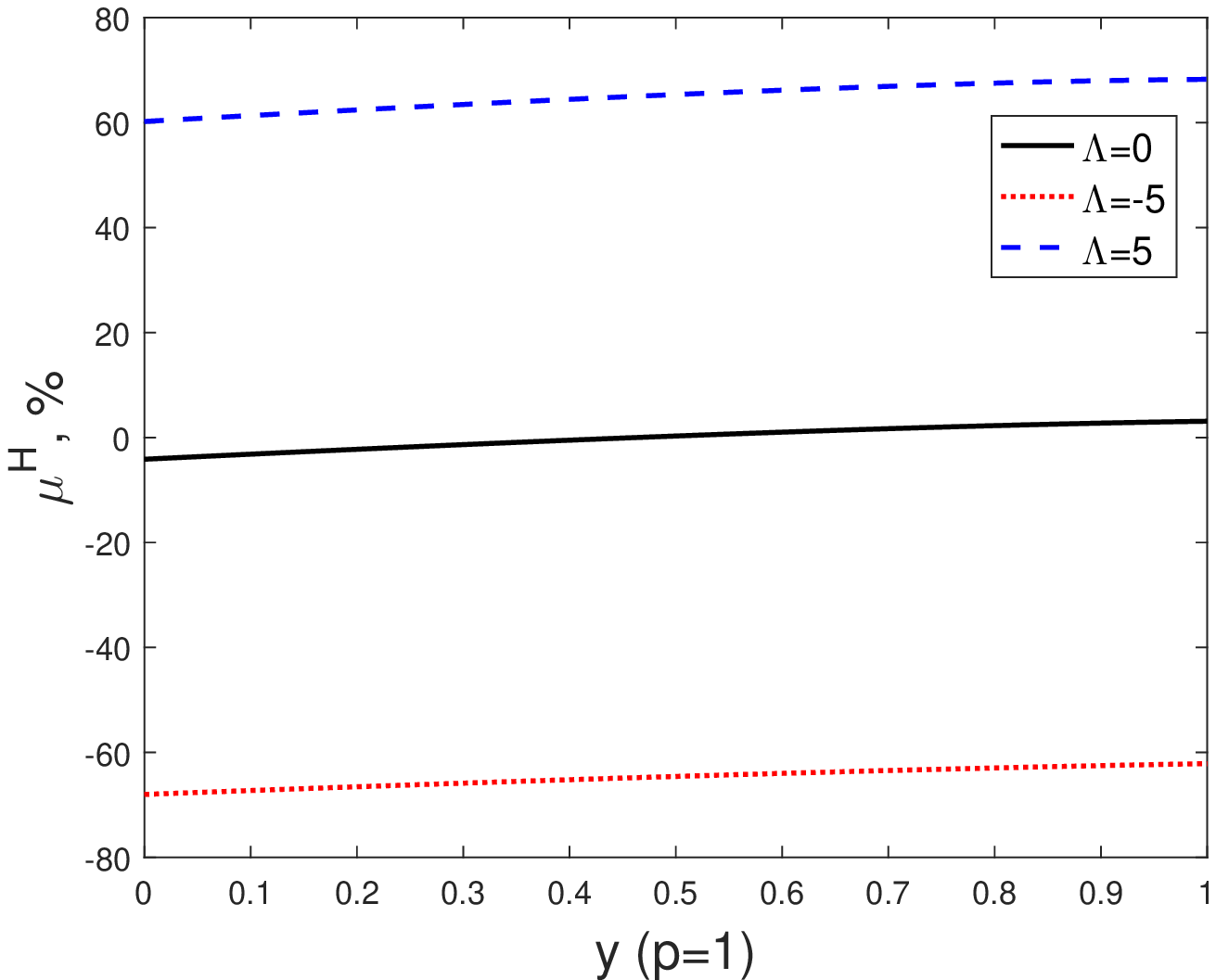}
}
\subfigure[$p=0.9$]{
\includegraphics[ width = 0.31\textwidth]{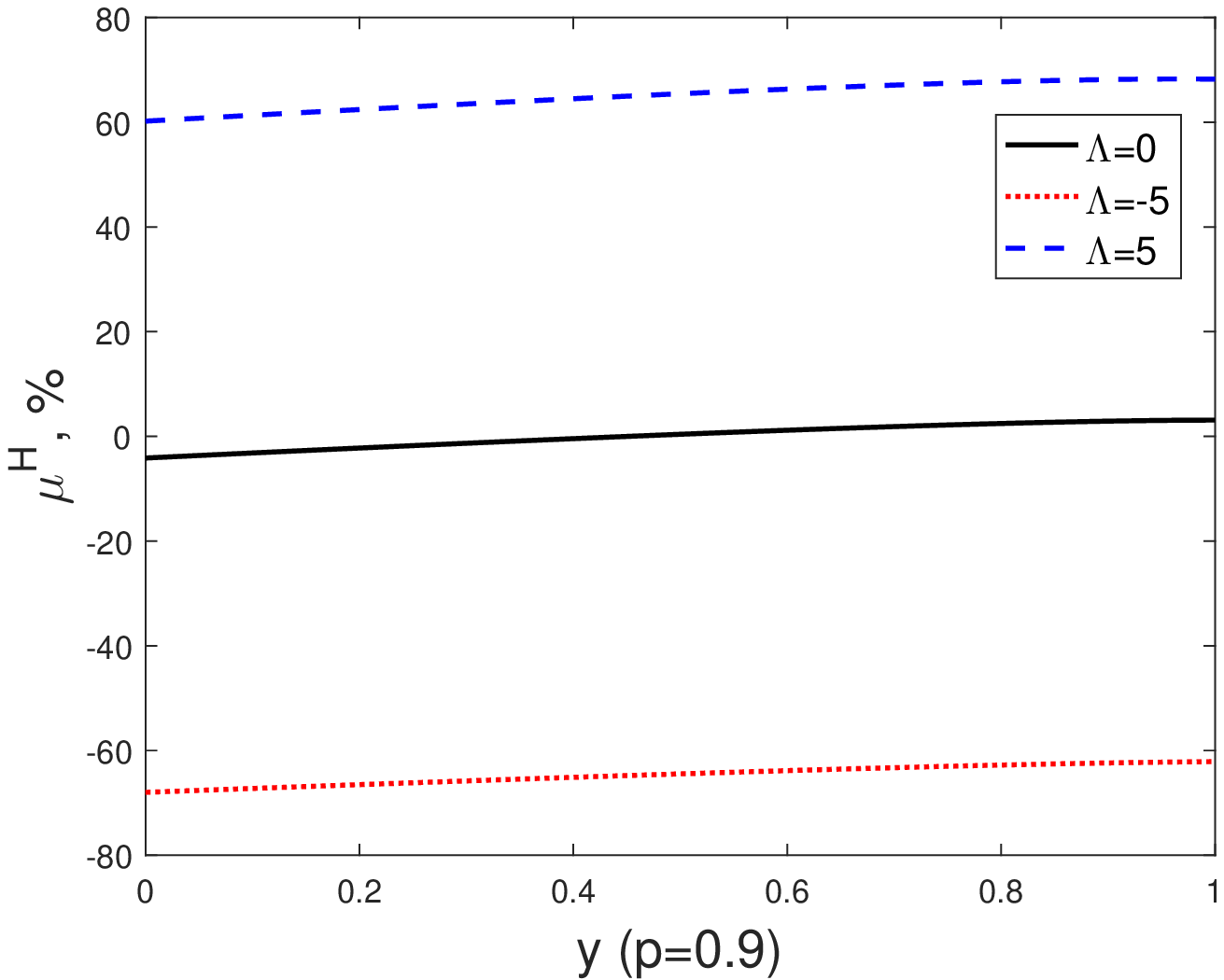}
}
\subfigure[$p=0.6$]{
\includegraphics[ width = 0.31\textwidth]{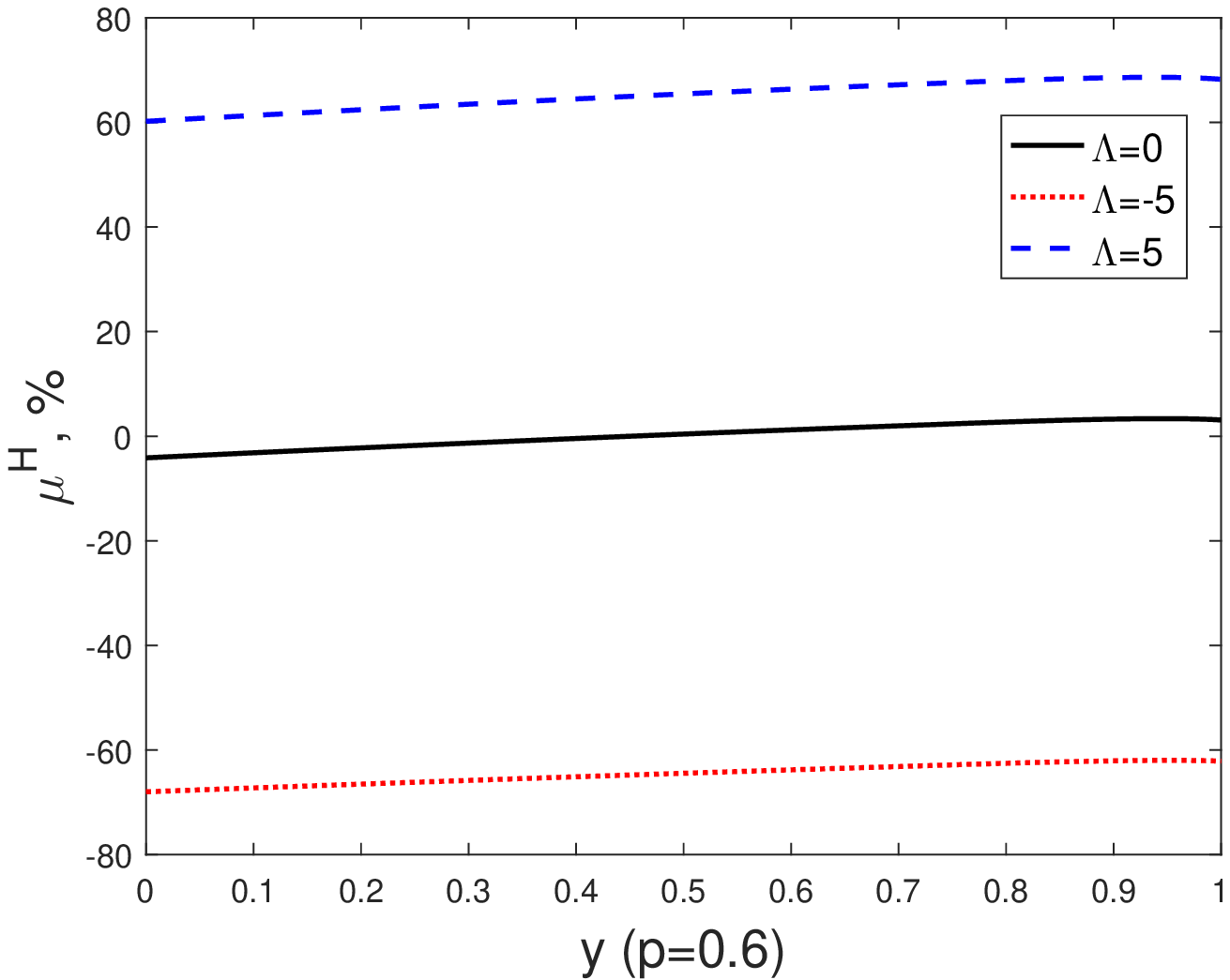}
}
\caption{Modified stock returns $\mu^H$ with different past information }
\label{fig_returnp}
\end{figure}

Figure \ref{fig_returnL} claims that whatever the past information presents, the modified stock return $\mu^H$ in the economy with imperfect protection is higher than $\mu^{H,\mathfrak{B}}$ in the benchmark economy with perfect protection. It is from the arbitrage-free theory (``there is no such thing as a free lunch'', see, e.g. \cite{Kwok} ) that higher risk leads to higher return, and by analysis in Figure \ref{fig_volL}, the fact $\sigma^H\geq\sigma^{H,\mathfrak{B}}$ naturally results in $\mu^H\geq\mu^{H,\mathfrak{B}}$.

Figure \ref{fig_returnp} states that compared with classic economies, investors' good memory (bad memory) gives rise to higher (lower) modified stock return wether or not investor protection in the economy is perfect, which is consistent with the general empirical evidence in \cite{Guidolin,Maheu}. This is because good memory (bad memory) certainly increases (decreases) the modified return of the output and by (\ref{ec-3}) investors' consumption increases (decreases), which then increases (decreases) stock value  through (\ref{st}) and (\ref{wit}) leading to higher (lower) modified stock return.

Due to the existence of diverting the output, the modified stock return may not reflect its true situation and then it is meaningful to additionally study the effects of investor protection and past information on the gross stock return.

\begin{figure}[!htbp]
\centering
\subfigure[$\Lambda_t=0$]{
\includegraphics[ width = 0.31\textwidth]{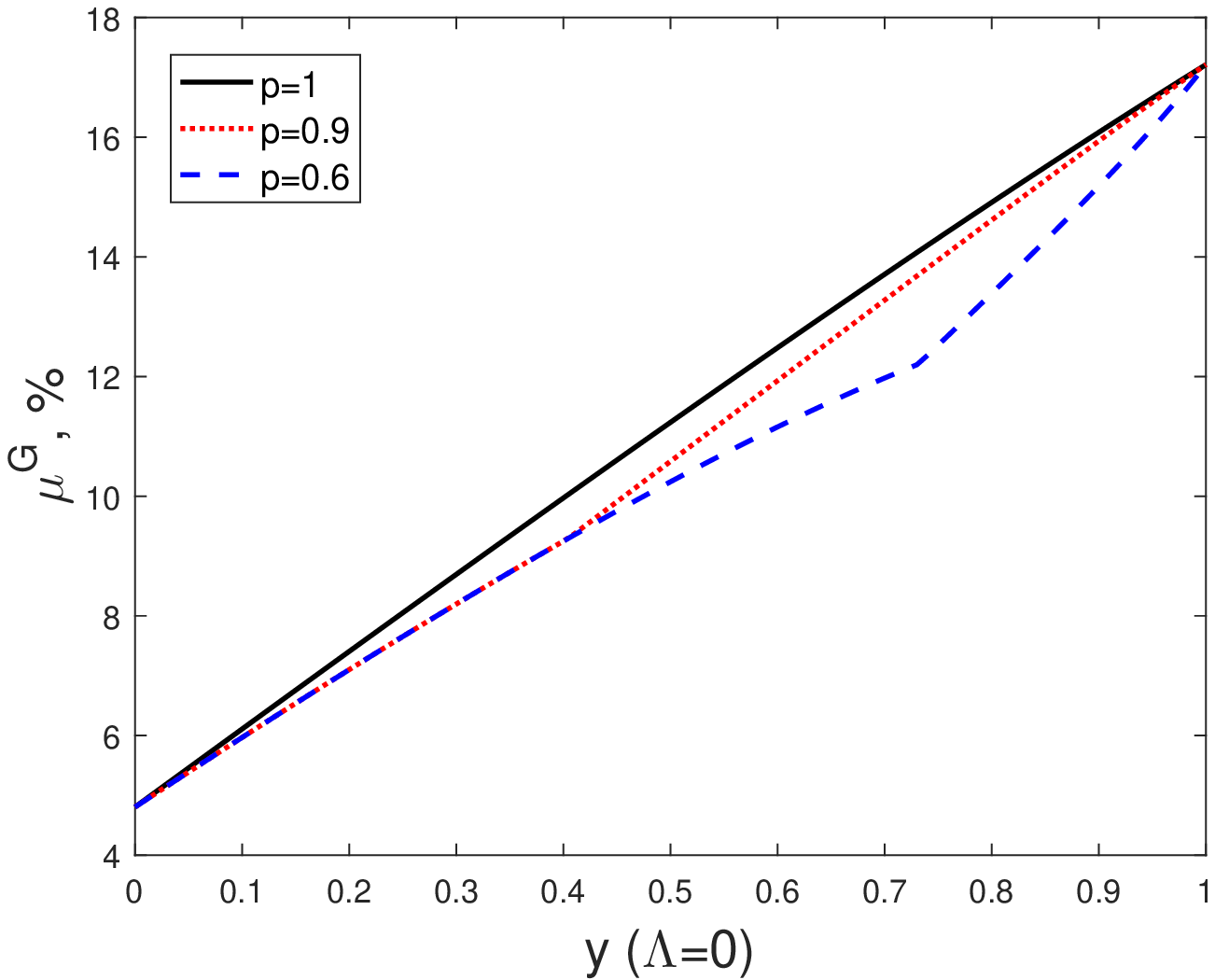}
}
\subfigure[$\Lambda_t=-5$]{
\includegraphics[ width = 0.31\textwidth]{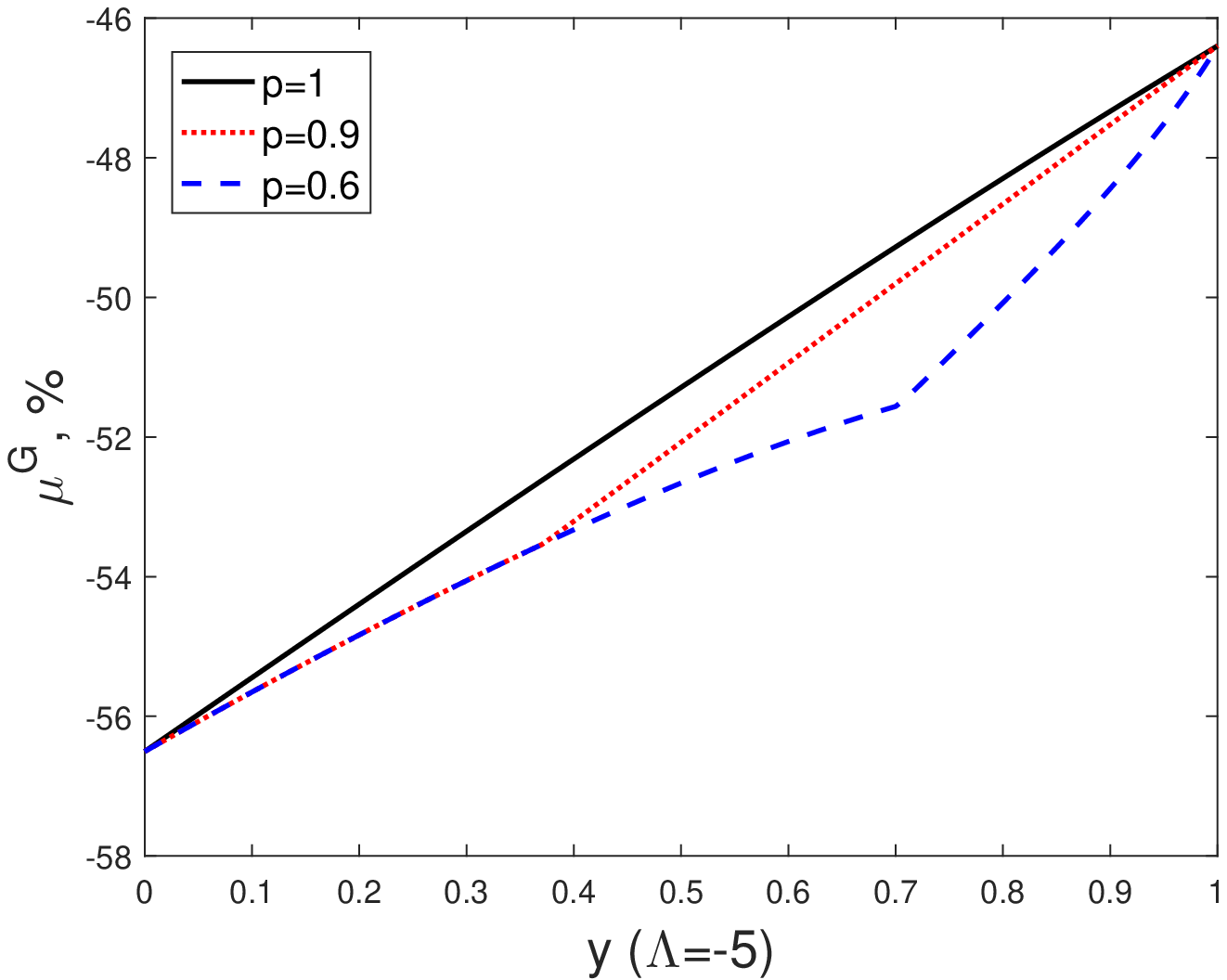}
}
\subfigure[$\Lambda_t=5$]{
\includegraphics[ width = 0.31\textwidth]{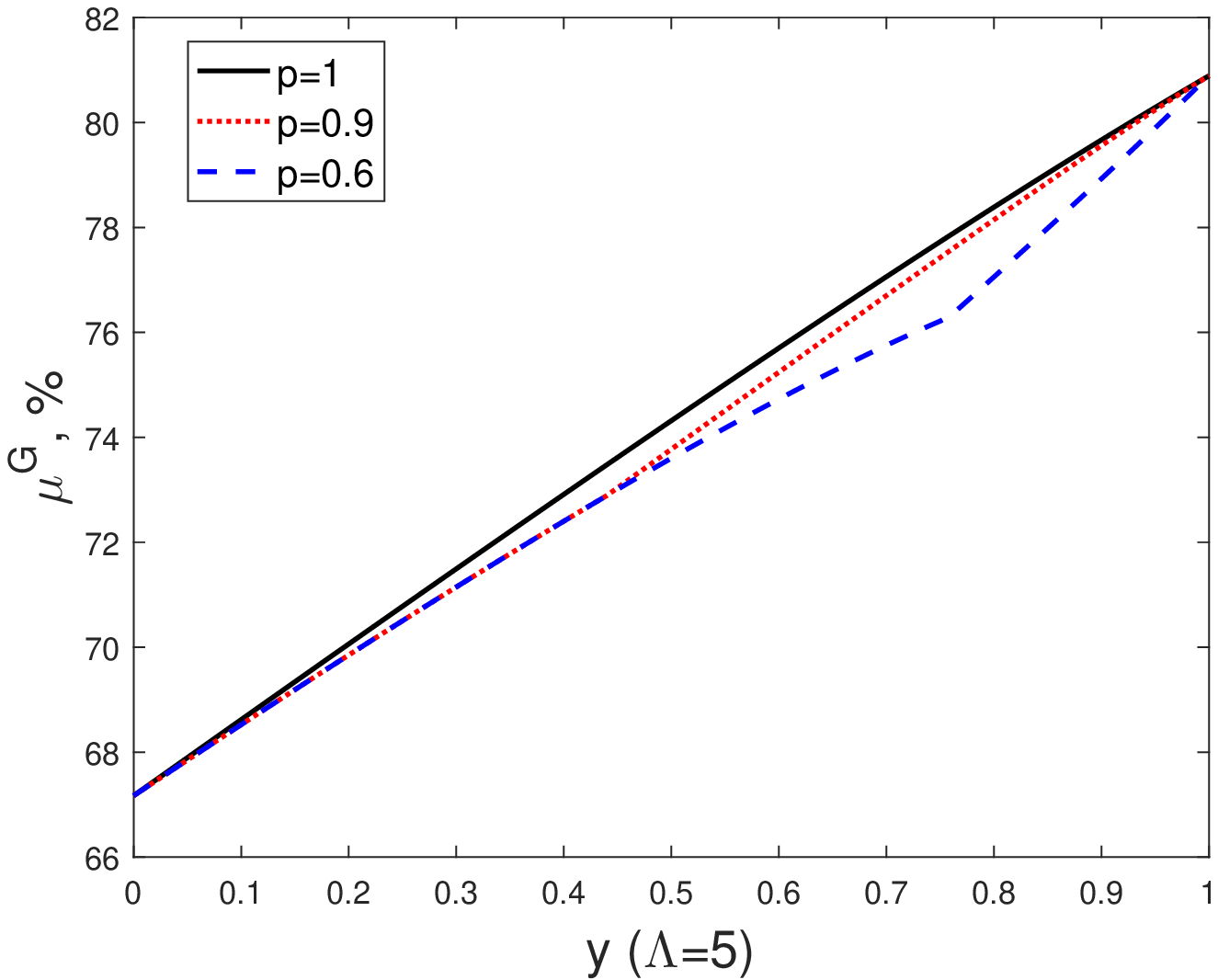}
}
\caption{Gross stock returns $\mu^G$ with different investor protection }
\label{fig_GreturnL}
\end{figure}
\begin{figure}[!htbp]
\centering
\subfigure[$p=1$]{
\includegraphics[ width = 0.31\textwidth]{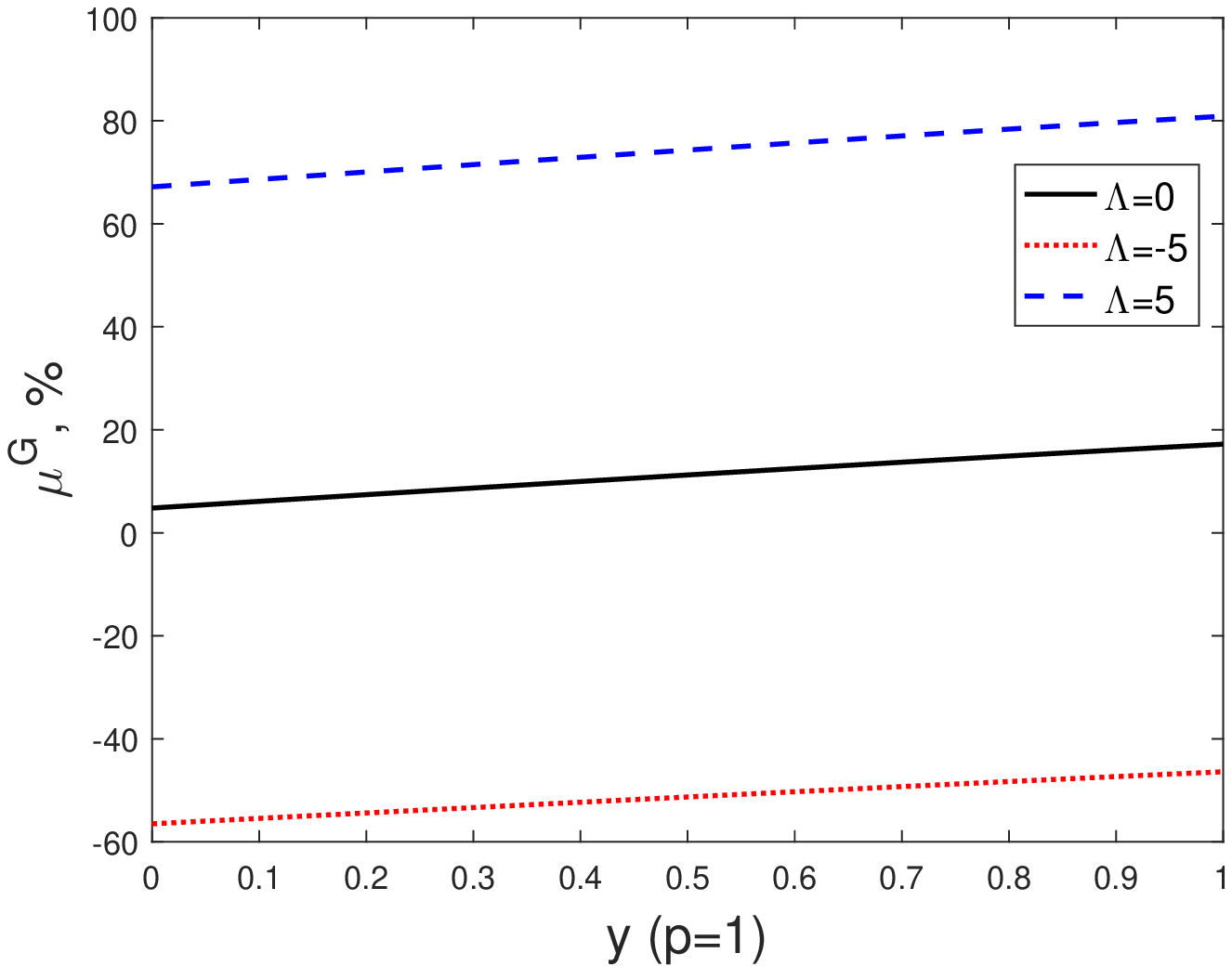}
}
\subfigure[$p=0.9$]{
\includegraphics[ width = 0.31\textwidth]{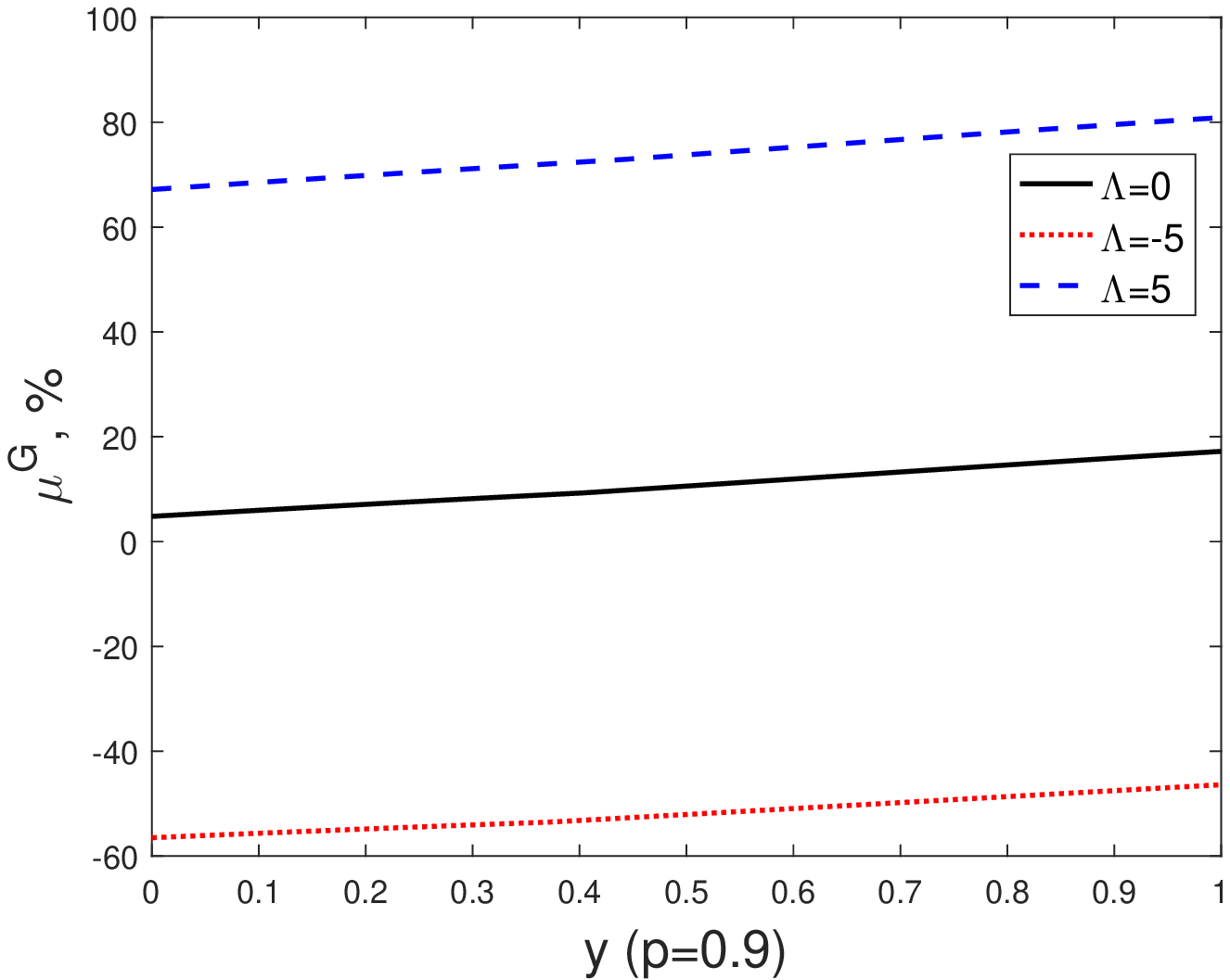}
}
\subfigure[$p=0.6$]{
\includegraphics[ width = 0.31\textwidth]{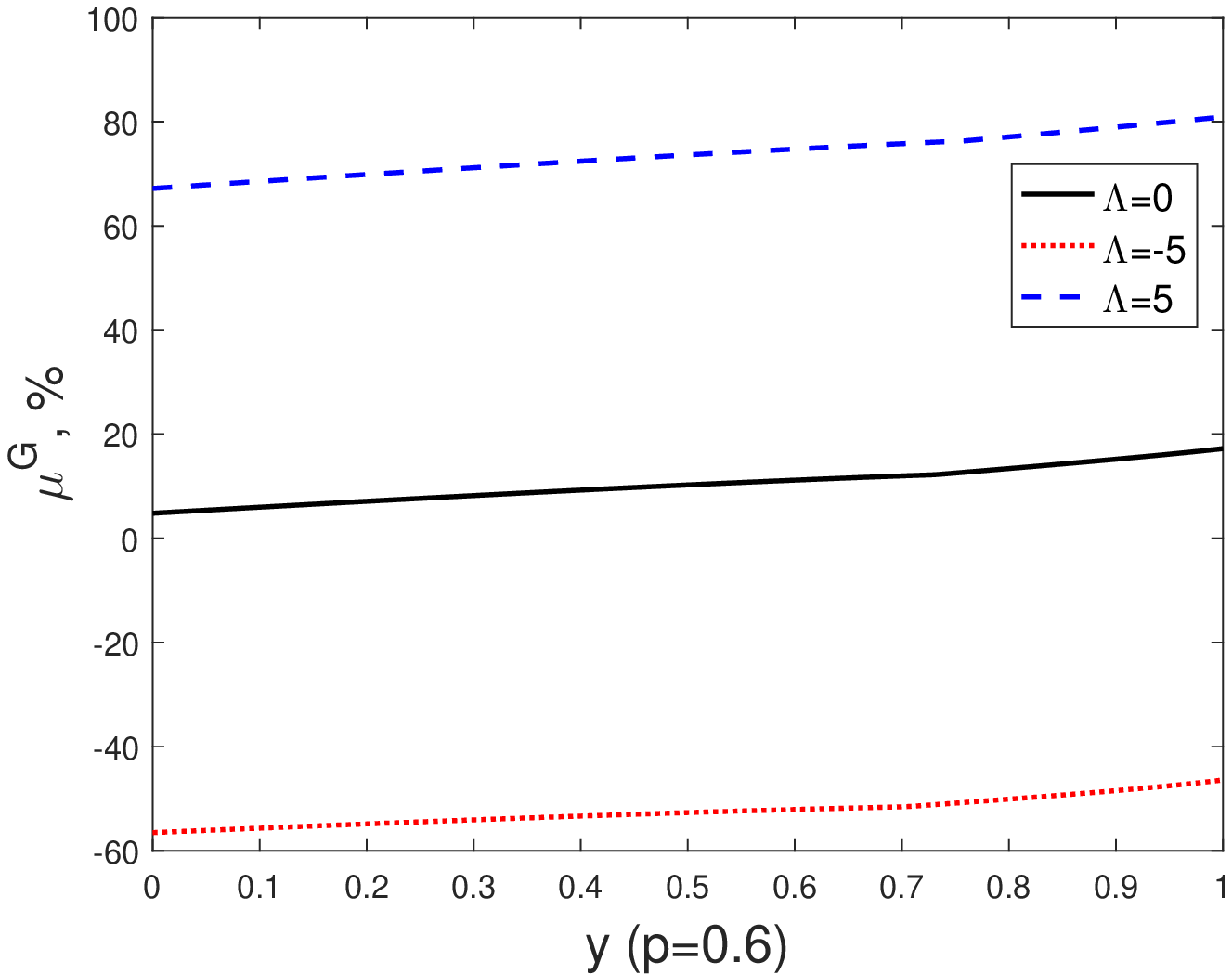}
}
\caption{Gross stock returns $\mu^G$ with different past information}
\label{fig_Greturnp}
\end{figure}

Figure \ref{fig_GreturnL} presents that no matter how investors are affected by the past information, poorer investor protection leads to lower gross stock return, which is consistent with our institution. The main reason is that poorer investor protection results in a higher fraction of diverted output. Compared with classic economies, it is showed by Figure \ref{fig_Greturnp} that investors' good memory (bad memory) brings about higher (lower) gross stock return wether or not investor protection in the economy is perfect. The reason is similar to the analysis in Figure \ref{fig_returnp}.

Summarizing, for each kind of investors' memory poorer investor protection leads to higher modified stock return and volatility and lower gross stock return and volatility, while compared with classic economies, good (bad) memory increases (decreases) modified stock return and volatility and gross stock return wether or not the investor protection is perfect.

\subsection{Interest rate}

The interest rate in the economy is another dynamic parameter indicating the situation of shareholders' investment in the bond. Figures \ref{fig_rateL} and \ref{fig_ratep} demonstrate the influence of investor protection and past information on the interest rate in equilibrium.

\begin{figure}[!htbp]
\centering
\subfigure[$\Lambda_t=0$]{
\includegraphics[ width = 0.31\textwidth]{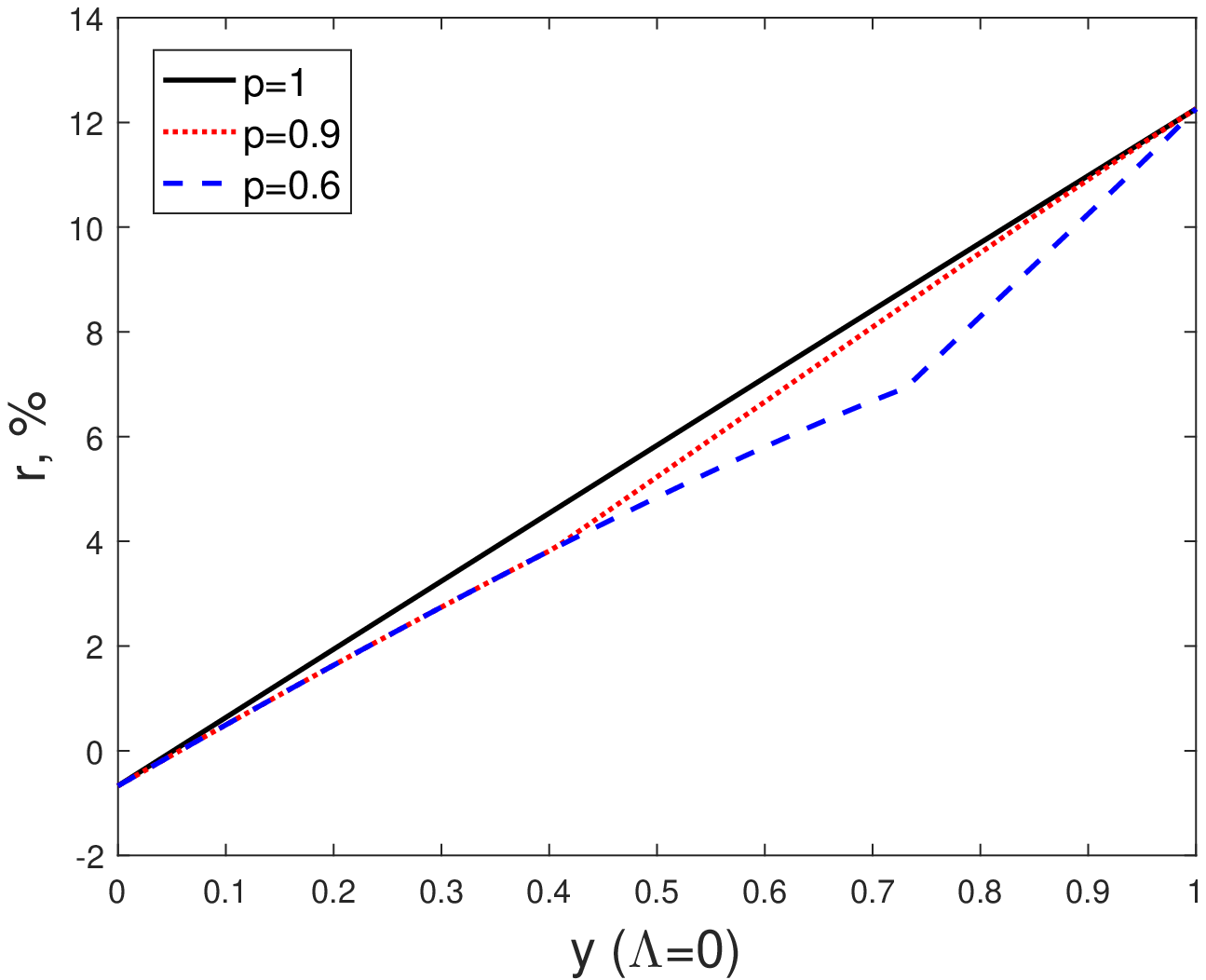}
}
\subfigure[$\Lambda_t=-5$]{
\includegraphics[ width = 0.31\textwidth]{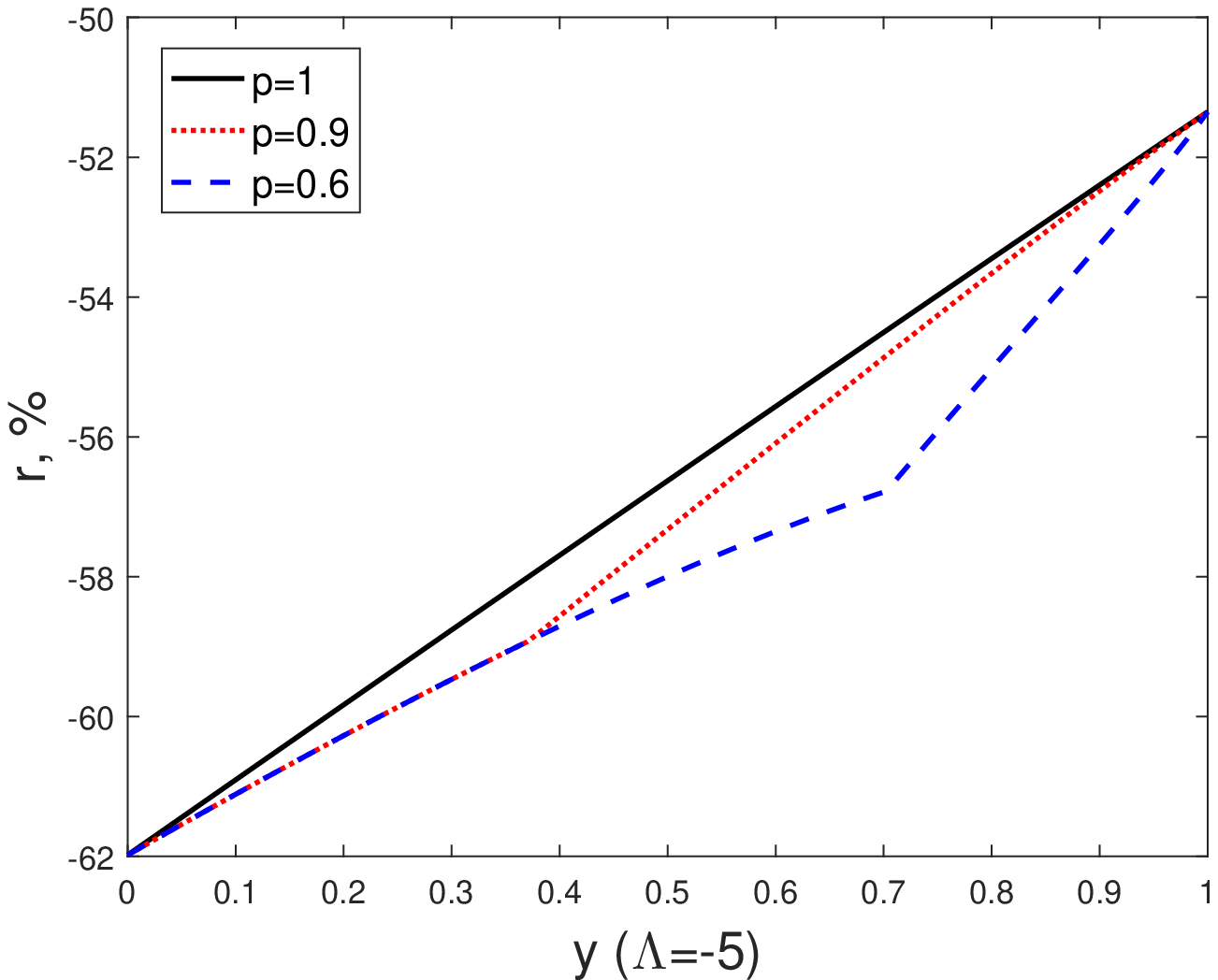}
}
\subfigure[$\Lambda_t=5$]{
\includegraphics[ width = 0.31\textwidth]{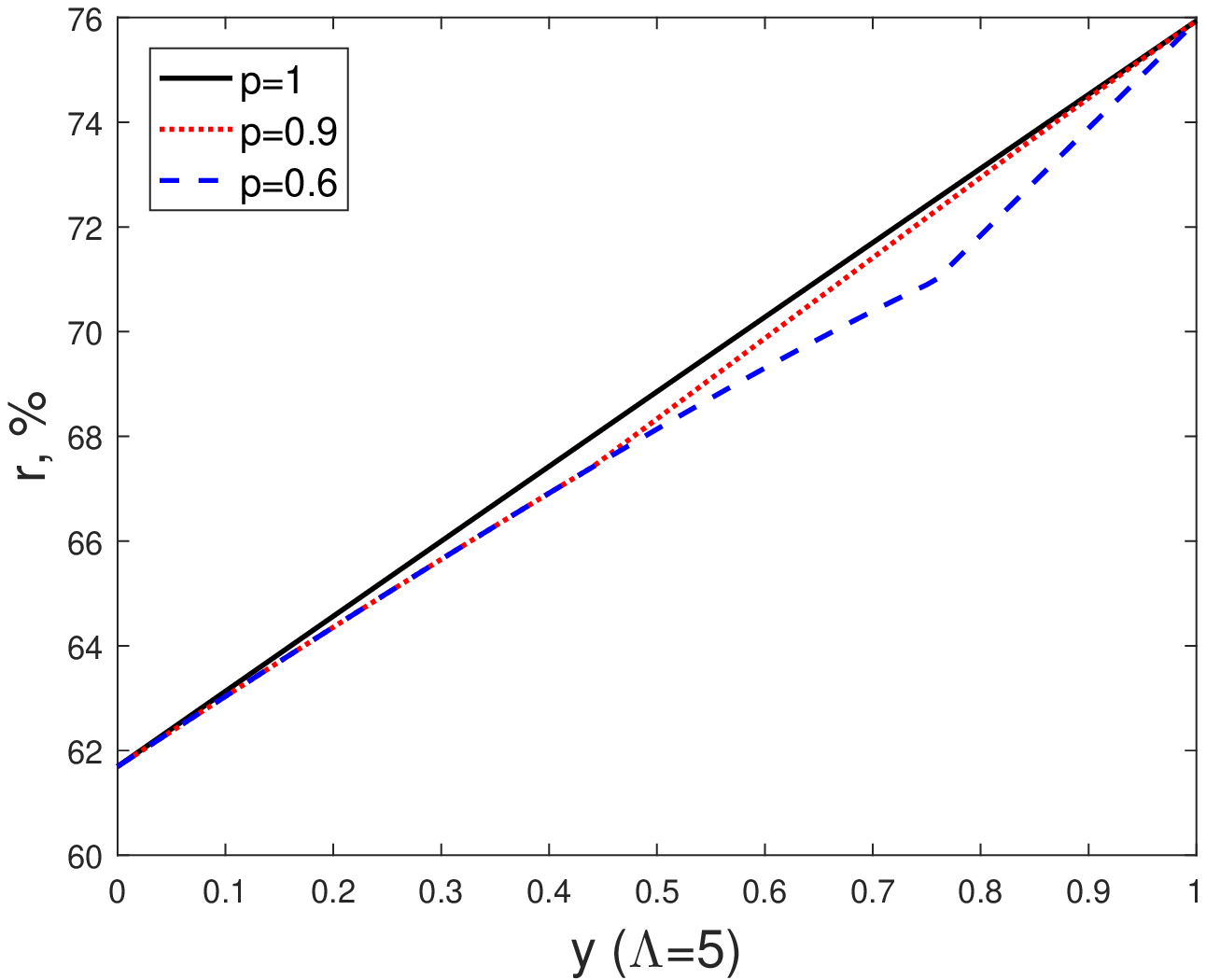}
}
\caption{Interest rates $r$ with different investor protection}
\label{fig_rateL}
\end{figure}
\begin{figure}[!htbp]
\centering
\subfigure[$p=1$]{
\includegraphics[ width = 0.31\textwidth]{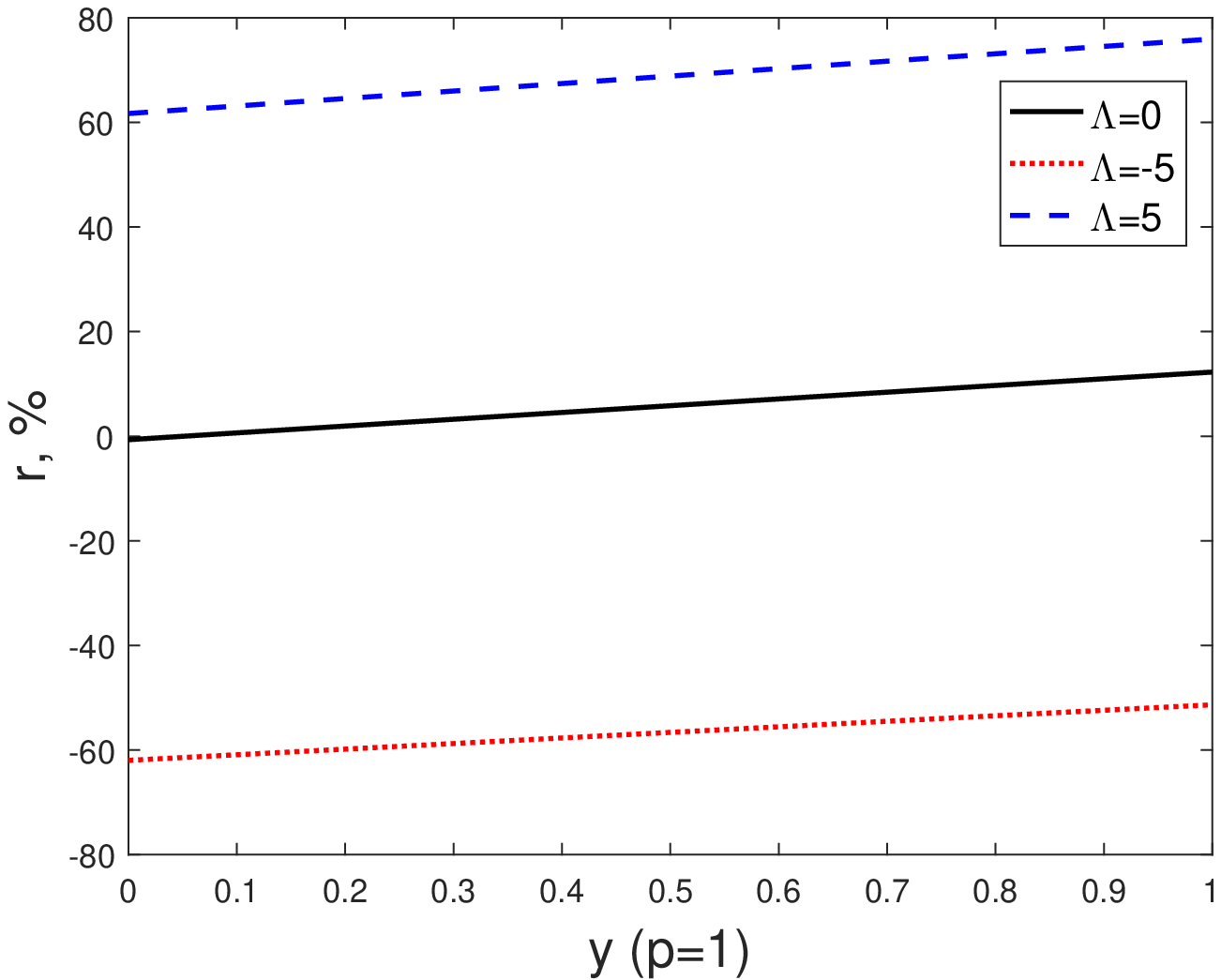}
}
\subfigure[$p=0.9$]{
\includegraphics[ width = 0.31\textwidth]{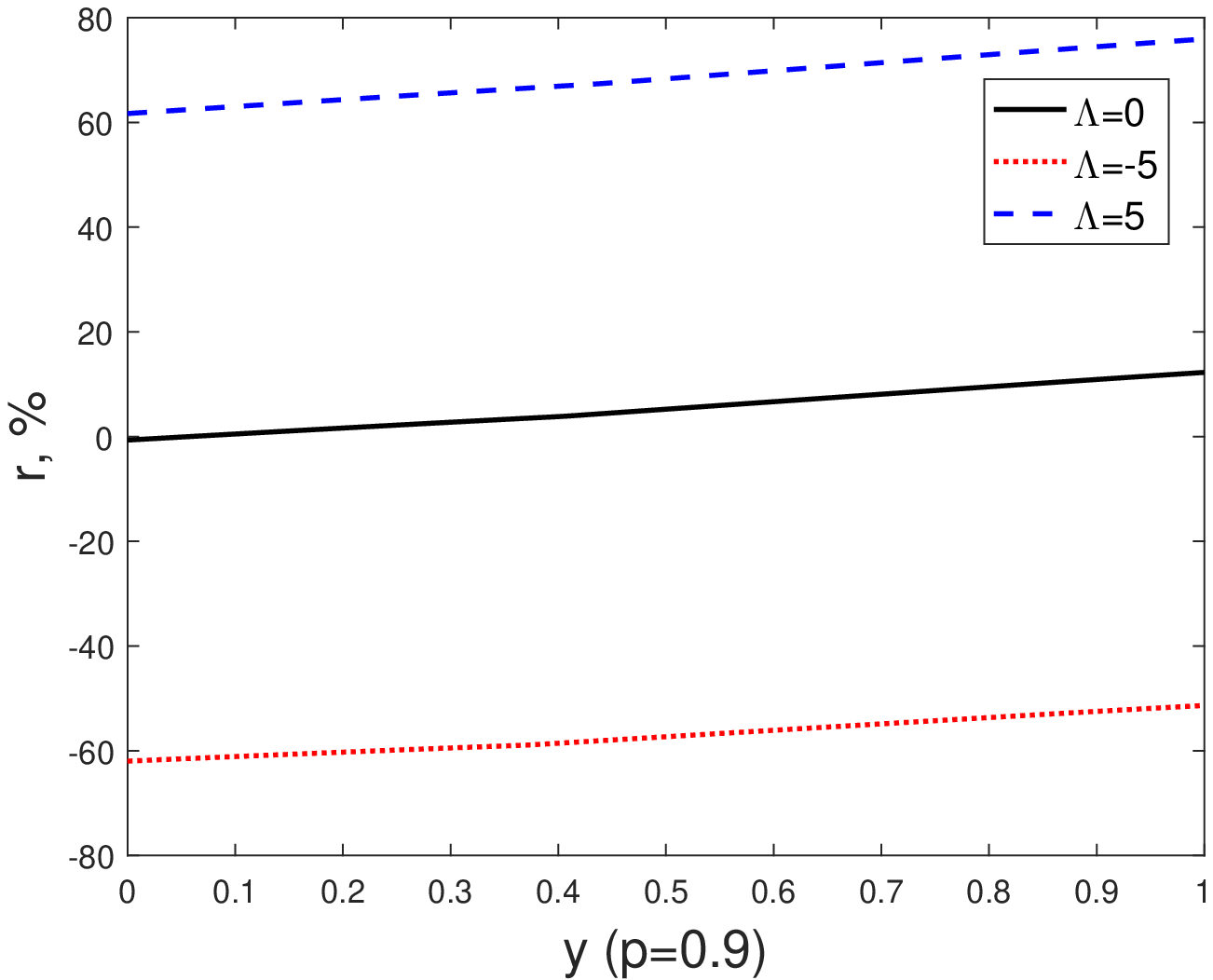}
}
\subfigure[$p=0.6$]{
\includegraphics[ width = 0.31\textwidth]{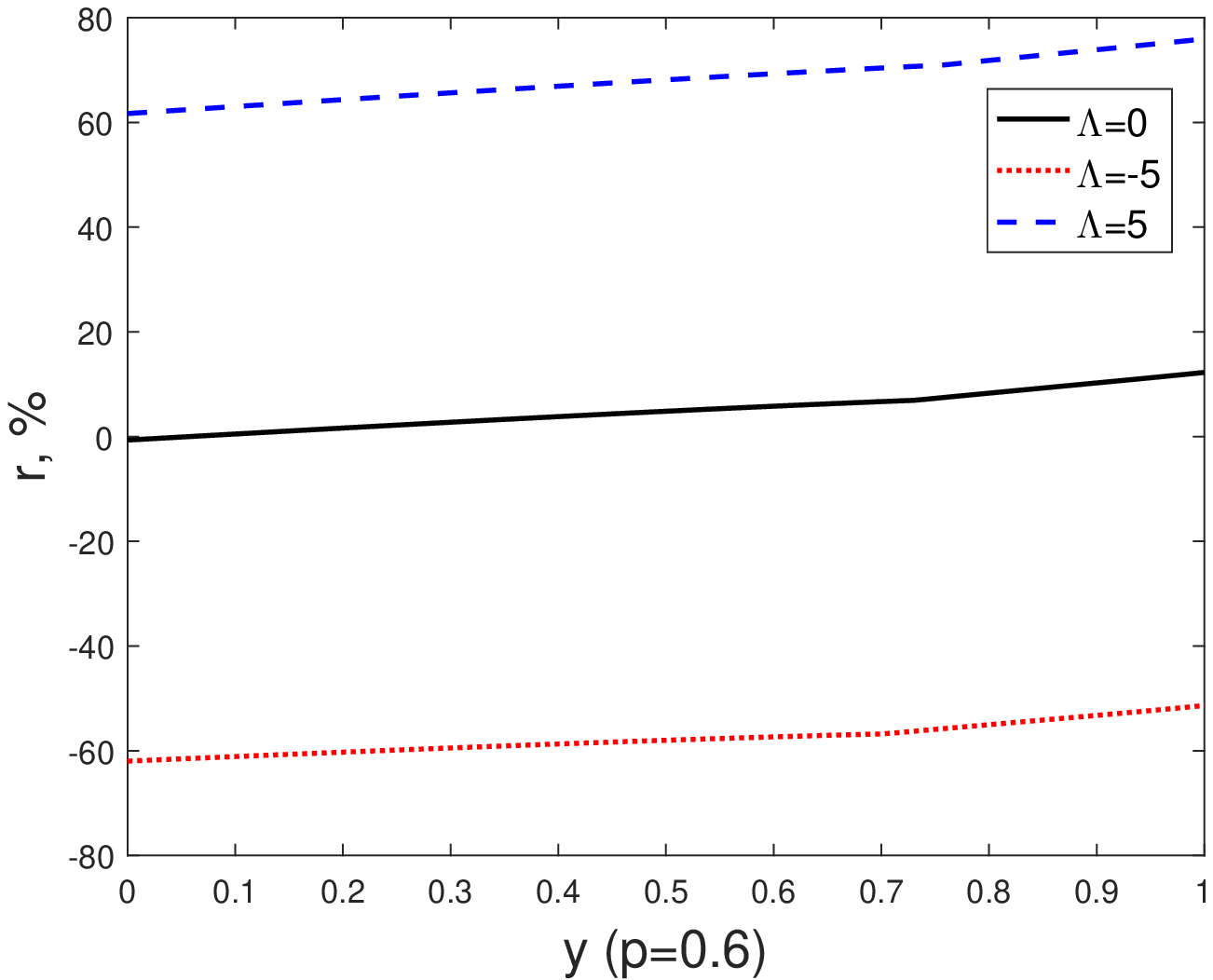}
}
\caption{Interest rates $r$ with different past information }
\label{fig_ratep}
\end{figure}

Figure \ref{fig_rateL} indicates that though the past information has different effects on investors' memory, poorer investor protection always leads to lower interest rate $r_t$. Since poorer investor protection decreases the gross stock return as we have discussed in Figure \ref{fig_GreturnL}, investors (especially the minority shareholder who can not benefit from diverting output) turn to investing the bond with cheaper credit, which then decreases the interest rate. Furthermore, poorer investor protection increases the fraction of diverted output $x_t^*$ as we have discussed in Figure \ref{fig_divertL}, and then stock investing could be partially covered by the diverted output leading to the decrease of the interest rate.

Figure \ref{fig_ratep} describes that compared with classic economies, investors' good memory (bad memory) results in higher (lower) interest rate wether or not investor protection in the economy is perfect. Bad memory would decreases the gross stock return as we have discussed in Figure \ref{fig_Greturnp}, which again decreases the interest rate. The reason for good memory can be considered similarly. This numerical result is consistent with Theorem 4.1 in \cite{Ohnishi}: ``Suppose two markets and assume that one market is more bullish than the other market. Then, the risk-free rate in a more bullish market is higher than that in a less bullish market."

It is interesting to note that not only is the entire line of interest rate below the zero line in the case of ``bad memory", but the interest rate in panel (a) in Figure \ref{fig_rateL} also turns negative for extremely low $y_t$. Negative interest rates of the former kind are experienced by a significant number of countries in the Global Financial Crisis (see, for instance, \cite{Chen}), and the main reason is easy to explain. With extreme ``bad memory", shareholders would not be willing to invest in the stock whose gross return is extremely low and does not match its comparatively high volatility (Figures \ref{fig_volL} and \ref{fig_volp}). Then substantial shareholders turn to invest the bond leading to the negative interest rate. For the latter, it suffices to explain the case around $y_t=0,\Lambda_t=0$. For one thing, the fact that controlling shareholder owns almost all stock shares urges minority shareholder to invest the bond, which produces cheaper credit. For another, though the past information appears to be ``no memory" at the present time, shareholders' sensitivity to past information make them prepare for future ``bad memory"  and then they must invest more money on the less risky bond to keep the consumption of necessary part. This behaviour of seeking safety are used to explain negative interest rates in unusual times in \cite{Anderson}. Hence, the negative interest rate is also caused by the parameter $\beta$ in (\ref{rto}).

\section{Conclusions}\label{section6}\noindent
\setcounter{equation}{0}
In this paper, based on the model in \cite{Basak} with investor protection and Brownian motion, we introduce an approximate fBm to a dynamic asset pricing model in an economy to address empirical regularities related to both investor protection for minority shareholders and memory properties of financial data. Besides a controlling shareholder who diverts a fraction of output, our model also features the good (or bad) memory obtained from the historical realized data in a pathwise way. Furthermore, for shareholders' objective functions (utility functions) in our approximate fractional economy, we consider shareholders' myopic preferences and sensitivity to the past information. In theory, we derive asset price dynamics in equilibria, and through the historical realized data, we establish an approximation scheme with pathwise convergence for the good (or bad) memory of investors for the past information, which ensures our theoretical results can be applied to numerical analysis. In numerical analysis, we find that poorer investor protection leads to higher stock holdings of controlling shareholders, lower gross stock returns, lower interest rates, and lower modified stock volatilities if the ownership concentration is sufficiently high. More importantly, we reveal that good (bad) memory would increase (decrease) stock holdings of controlling shareholders, modified stock returns and volatilities, gross stock returns, and interest rates, which is consistent with our intuition. Furthermore, we uncover how shareholders' memory affects investor protection for minority shareholders: good (bad) memory
would strengthen (weaken) investor protection for minority shareholder when the ownership concentration is sufficiently high, while good (bad) memory would inversely weaken (strengthen) investor protection for minority shareholder when the ownership concentration is sufficiently low.

There are several problems which have not been considered in this paper and we leave these as our future work. The driven approximate fBms of the output and the stock may be totally different, which means that investors have different memories for the output and the stock respectively. As we are able to deal with the memory for the output, how could we obtain the past information with pathwise convergence from the historical realized  price of the stock? Even in the economy driven by Brownian motion (i.e. $H=\frac{1}{2}$), there may not exist equilibria under market clearing conditions (\ref{ec-1})-(\ref{ec-3}) in the economy with extreme parameters. Then how could we deal with such extreme markets? Finally, the convergence rate ($(\Delta t)^{1/6}$) of the approximation scheme for $\Lambda_t$ and $\lambda_t$ is very slow and so developing some more efficient approximation schemes for $\Lambda_t$ and $\lambda_t$ remains to be solved.

\section*{Appendices}
\setcounter{equation}{0}
\renewcommand{\theequation}{A.\arabic{equation}}
\renewcommand{\thelemma}{A.\arabic{lemma}}

\textbf{\large Proof of Theorem \ref{th1}}\\
As discussed in Section \ref{section2}, $\Lambda_t$ is an $\mathcal{F}_t$-adapt process and known theoretically by using the information of $\{\widehat{D}_s\}_{0\leq s\leq t}$ and we may deal with $\Lambda_t$ as an exogenous process like $\widehat{D}_t$.

Solving (\ref{problemCc}) and (\ref{problemMc}) gives (\ref{ps_c}). Maximizing the quadratic objective function $J_M$ in (\ref{problemMn}), we obtain (\ref{ps_nM}) where $x_t^*$ is given latter.

It is easy to see that $J_c$ in (\ref{problemCn}) is a quadratic function of $x_t$ and
\begin{equation*}
\frac{\partial}{\partial x}J_C(n_{Ct};x_t)=-k\frac{D_t}{W_{Ct}}\left(x_t-\frac{1-n_{Ct}}{k}\right).
\end{equation*}
Hence, considering the constrain (\ref{protectC}), we get $x^*_t(n_{Ct})=\min\left\{\frac{1-n_{Ct}}{k},(1-p)n_{Ct}\right\}$. Substituting $x^*_t(n_{Ct})$ back into $J_C$ gives
\begin{equation}\label{J}
J_C(n_{Ct})=\left\{{\begin{array}{*{20}{l}}
J_{C1}(n_{Ct}),\quad n_{Ct}\leq \frac{1}{1+(1-p)k}\, \left(x^*_t(n_{Ct})=(1-p)n_{Ct}\right),\\
J_{C2}(n_{Ct}),\quad n_{Ct}\geq\frac{1}{1+(1-p)k}\, \left(x^*_t(n_{Ct})=\frac{1-n_{Ct}}{k}\right),
\end{array}}\right.
\end{equation}
where
\begin{align*}
J_{C1}(n_{Ct})=&-\frac{1}{2}\frac{S_t}{W_{Ct}}\left[2H\varepsilon^{2h}\delta_C\frac{S_t}{W_{Ct}}\sigma_t^2
+(2(1-p)+k(1-p)^2)\frac{D_t}{S_t}\right]n_{Ct}^2\nonumber\\
&+\frac{S_t}{W_{Ct}}\left[\mu_t-r_t+(2-p)\frac{D_t}{S_t}+\sigma_t\Lambda_t\right]n_{Ct},\\
J_{C2}(n_{Ct})=&-\frac{1}{2}\frac{S_t}{W_{Ct}}\left[2H\varepsilon^{2h}\delta_C\frac{S_t}{W_{Ct}}\sigma_t^2
-\frac{1}{k}\frac{D_t}{S_t}\right]n_{Ct}^2\nonumber\\
&+\frac{S_t}{W_{Ct}}\left[\mu_t-r_t+\left(1-\frac{1}{k}\right)\frac{D_t}{S_t}+\sigma_t\Lambda_t\right]n_{Ct}+\frac{D_t}{2kW_{Ct}}.
\end{align*}
In the case of $n_{Ct}\leq \frac{1}{1+(1-p)k}$, it is obvious that $J_{C1}$ is a quadratic concave function of $n_{Ct}$, and there are two possible maximal points for $J_{C1}$---$n_{Ct,1}^*$ and $n_{Ct,3}^*$. In the case of $n_{Ct}\geq \frac{1}{1+(1-p)k}$, $J_{C2}$ is a quadratic concave or convex function (or even a linear function) of $n_{Ct}$, but it is clear that there exist three possible maximal points for $J_{C2}$---$n_{Ct,2}^*,n_{Ct,3}^*$ and $n_{Ct,4}^*$.   Summarizing aforementioned analysis and searching over above points, we can determine the global maximum and the associated maximum point $n_{Ct}^*$, which gives (\ref{ps_nC}) and (\ref{nC1})-(\ref{nC4}).
\\
\\
\textbf{\large Proof of Theorem \ref{th2}}\\
Using (\ref{budgetB}), (\ref{ps_c}) and the market clearing conditions (\ref{ec-1})-(\ref{ec-3}), we have
\begin{align}
S_t=&W_{Ct}+W_{Mt},\label{st}\\
W_{it}=&\rho^{-\frac{1}{\delta_i}}\varphi(\Lambda_t)^{-\theta_i}c_{it}^*,\;i\in\{C,M\},\label{wit}\\
\widehat{D}_t=&\rho^{\frac{1}{\delta_C}}\varphi(\Lambda_t)^{\theta_C}W_{Ct}
+\rho^{\frac{1}{\delta_M}}\varphi(\Lambda_t)^{\theta_M}W_{Mt},\label{dw}
\end{align}
which clearly implies that
\begin{align}
\frac{\widehat{D}_t}{S_t}=&\frac{\widehat{D}_t}{\rho^{-\frac{1}{\delta_C}}\varphi(\Lambda_t)^{-\theta_C}c_{Ct}^*
+\rho^{-\frac{1}{\delta_M}}\varphi(\Lambda_t)^{-\theta_M}c_{Mt}^*}=
\frac{1}{(1-y_t)\rho^{-\frac{1}{\delta_C}}\varphi(\Lambda_t)^{-\theta_C}+y_t\rho^{-\frac{1}{\delta_M}}\varphi(\Lambda_t)^{-\theta_M}},\label{rato1}\\
\frac{W_{Ct}}{S_t}=&\frac{\rho^{-\frac{1}{\delta_C}}\varphi(\Lambda_t)^{-\theta_C}c_{Ct}^*}
{\rho^{-\frac{1}{\delta_C}}\varphi(\Lambda_t)^{-\theta_C}c_{Ct}^*
+\rho^{-\frac{1}{\delta_M}}\varphi(\Lambda_t)^{-\theta_M}c_{Mt}^*}
=\frac{(1-y_t)\rho^{-\frac{1}{\delta_C}}\varphi(\Lambda_t)^{-\theta_C}}{(1-y_t)\rho^{-\frac{1}{\delta_C}}\varphi(\Lambda_t)^{-\theta_C}+y_t\rho^{-\frac{1}{\delta_M}}\varphi(\Lambda_t)^{-\theta_M}},\label{rato2}\\
\frac{W_{Mt}}{S_t}=&\frac{\rho^{-\frac{1}{\delta_M}}\varphi(\Lambda_t)^{-\theta_M}c_{Mt}^*}
{\rho^{-\frac{1}{\delta_C}}\varphi(\Lambda_t)^{-\theta_C}c_{Ct}^*
+\rho^{-\frac{1}{\delta_M}}\varphi(\Lambda_t)^{-\theta_M}c_{Mt}^*}
=\frac{y_t\rho^{-\frac{1}{\delta_M}}\varphi(\Lambda_t)^{-\theta_M}}{(1-y_t)\rho^{-\frac{1}{\delta_C}}\varphi(\Lambda_t)^{-\theta_C}+y_t\rho^{-\frac{1}{\delta_M}}\varphi(\Lambda_t)^{-\theta_M}}.\label{rato3}
\end{align}

The above equations clearly gives (\ref{rato4})-(\ref{rato6}). Inserting (\ref{rato4}) into (\ref{kt}), the stock mean-return $\mu_t$ (the equation (\ref{miut})) is easily found.

Applying It\^{o}'s Lemma to both sides of (\ref{dw}) and then matching the terms of $dt$ and $dw_t$, it is seen that
\begin{align*}
\mu_D+\sigma_D\Lambda_t=&\frac{c^*_{Ct}}{\widehat{D}_t}\bigg[r_t+\frac{\sqrt{2H}\varepsilon^h n_{Ct}^*\sigma_t\kappa_tS_t}{W_{Ct}}+2Hh\varepsilon^{2h-1}\theta_C\frac{\varphi'(\Lambda_t)}{\varphi(\Lambda_t)}
\frac{n^*_{Ct}S_t\sigma_t}{W_{Ct}}+\left(x_t^*-\frac{kx_t^{*2}}{2}\right)\frac{D_t}{W_{Ct}}\\
&-\rho^{\frac{1}{\delta_C}}\varphi(\Lambda_t)^{\theta_C}+\theta_C\lambda_t\frac{\varphi'(\Lambda_t)}{\varphi(\Lambda_t)}
+Hh^2\varepsilon^{2h-2}\theta_C\left((\theta_C-1)\left(\frac{\varphi'(\Lambda_t)}{\varphi(\Lambda_t)}\right)^2
+\frac{\varphi''(\Lambda_t)}{\varphi(\Lambda_t)}\right)\bigg]\\
&+\frac{c^*_{Mt}}{\widehat{D}_t}\bigg[r_t+\frac{\sqrt{2H}\varepsilon^h n_{Mt}^*\sigma_t\kappa_tS_t}{W_{Mt}}+2Hh\varepsilon^{2h-1}\theta_M\frac{\varphi'(\Lambda_t)}{\varphi(\Lambda_t)}
\frac{n^*_{Mt}S_t\sigma_t}{W_{Mt}}+\frac{kx_t^{*2}}{2}\frac{D_t}{W_{Mt}}\\
&-\rho^{\frac{1}{\delta_M}}\varphi(\Lambda_t)^{\theta_M}+\theta_M\lambda_t\frac{\varphi'(\Lambda_t)}{\varphi(\Lambda_t)}
+Hh^2\varepsilon^{2h-2}\theta_M\left((\theta_M-1)\left(\frac{\varphi'(\Lambda_t)}{\varphi(\Lambda_t)}\right)^2
+\frac{\varphi''(\Lambda_t)}{\varphi(\Lambda_t)}\right)\bigg]\\
&+l_C\rho^{\frac{1}{\delta_C}}\varphi(\Lambda_t)^{\theta_C}+l_M\rho^{\frac{1}{\delta_M}}\varphi(\Lambda_t)^{\theta_M},\\
\sigma_D=&\frac{c^*_{Ct}}{\widehat{D}_t}\left(\frac{n_{Ct}^*\sigma_tS_t}{W_{Ct}}+\theta_C\frac{h}{\varepsilon}
\frac{\varphi'(\Lambda_t)}{\varphi(\Lambda_t)}\right)+\frac{c^*_{Mt}}{\widehat{D}_t}\left(\frac{n_{Mt}^*\sigma_tS_t}
{W_{Mt}}+\theta_M\frac{h}{\varepsilon}
\frac{\varphi'(\Lambda_t)}{\varphi(\Lambda_t)}\right).
\end{align*}
Then by substituting $y_t=\frac{c_{Mt}^*}{\widehat{D}_t},1-y_t=\frac{c_{Ct}^*}{\widehat{D}_t}$, (\ref{rato5}) and (\ref{rato2})-(\ref{rato3}), we derive the interest rate $r_t$ and the stock volatility $\sigma_t$ as (\ref{rt}) and (\ref{sigmat}) respectively.

It is from (\ref{ps_nM}) and (\ref{kt}) that
\[
\kappa_t=\sqrt{2H}\varepsilon^h\delta_M(1-n_{Ct}^*)\sigma_t\frac{S_t}{W_{Mt}},
\]
which together with (\ref{sigmat}) and (\ref{rato3}) gives (\ref{ktt}) of the Sharpe ratio $\kappa_t$.

We can rewrite $y_t$ as
\[y_t=\rho^{\frac{1}{\delta_M}}\varphi(\Lambda_t)^{\theta_M}\frac{W_{Mt}}{\widehat{D}_t}.\]
Then applying It\^{o}'s Lemma to both sides of
\[
 y_t\widehat{D}_t=\rho^{\frac{1}{\delta_M}}\varphi(\Lambda_t)^{\theta_M}W_{Mt}
\]
and matching the terms of $dt$ and $dw_t$ again, we just get $\mu_{yt}$ and $\sigma_{yt}$ as (\ref{muy}) and (\ref{sy}) respectively.

Finally, we turn to deriving $b_{it}^*$ and $n_{it}^*$ ($i\in\{C,M\}$). It is easy to obtain (\ref{bcc})-(\ref{nmt}) and (\ref{thnC3})-(\ref{thnC4}). In Region 2, since $x_t^*(n_{Ct,2}^*)=\frac{1-n_{Ct,2}^*}{k}$, we may rewrite (\ref{nC2}) as
\[
n_{Ct,2}^*=\frac{\mu_t-r_t+\left(1-x_t^*(n_{Ct,2}^*)\right)\frac{D_t}{S_t}
+\sigma_t\Lambda_t}{2H\varepsilon^{2h}\delta_C\frac{S_t}{W_{Ct}}\sigma_t^2
}=\frac{\kappa_t}{\sqrt{2H}\varepsilon^{h}\delta_C\frac{S_t}{W_{Ct}}\sigma_t},
\]
and then by substituting (\ref{sigmat}), (\ref{ktt}) and (\ref{rato2}) for  $\sigma_t$, $\kappa_t$ and $\frac{W_{Ct}}{S_t}$ we obtain
\[
n_{Ct,2}^*=\frac{1-y_t}{{\delta_C}\rho^{\frac{1}{\delta_C}}\varphi(\Lambda_t)^{\theta_C}}
\frac{{\delta_M}\rho^{\frac{1}{\delta_M}}\varphi(\Lambda_t)^{\theta_M}}{y_t}(1-n_{Ct,2}^*).
\]
Solving the above equation for $n_{Ct,2}^*$, we derive (\ref{thnC2}) for $n_{Ct,2}^*$. In Region 1, as $x_t^*(n_{Ct,1}^*)=n_{Ct,1}^*(1-p)$, we may rewrite (\ref{nC1}) as
\begin{align*}
n_{Ct,1}^*=&\frac{\mu_t-r_t+(1-x_t^*(n_{Ct,1}^*))\frac{D_t}{S_t}+\sigma_t\Lambda_t+(1-p)[1-n_{Ct,1}^*-k(1-p)n_{Ct,1}^*]\frac{D_t}{S_t}}
{2H\varepsilon^{2h}\delta_C\frac{S_t}{W_{Ct}}\sigma_t^2}\\
=&\frac{\kappa_t}{\sqrt{2H}\varepsilon^{h}\delta_C\frac{S_t}{W_{Ct}}\sigma_t}+(1-p)\left(1-\frac{n_{Ct,1}^*}{\frac{1}{1+(1-p)k}}\right)
\frac{\frac{D_t}{S_t}}{2H\varepsilon^{2h}\delta_C\frac{S_t}{W_{Ct}}\sigma_t^2}
\end{align*}
Then substituting (\ref{sigmat}), (\ref{ktt}), (\ref{thnC2}), (\ref{thnC3}), (\ref{rato4}) and (\ref{rato6}) for $\sigma_t$, $\kappa_t$, $n_{Ct,2}^*$, $n_{Ct,3}^*$, $\frac{{D}_t}{S_t}$ and $\frac{S_t}{W_{Ct}}$ respectively, we obtain
\begin{align*}
n_{Ct,1}^*=&(1-n_{Ct,1}^*){\frac{1-y_t}{{\delta_C}\rho^{\frac{1}{\delta_C}}\varphi(\Lambda_t)^{\theta_C}}}
\left/{\frac{y_t}{{\delta_M}\rho^{\frac{1}{\delta_M}}\varphi(\Lambda_t)^{\theta_M}}}\right.
+\left(1-\frac{n_{Ct,1}^*}{n_{Ct,3}^*}\right)
\frac{(1-p)(1-l_C-l_M)}{2H\varepsilon^{2h}}\frac{1-y_t}{{\delta_C}\rho^{\frac{1}{\delta_C}}\varphi(\Lambda_t)^{\theta_C}}\\
&\left.\times\left[n_{Ct,1}^*\rho^{\frac{1}{\delta_C}}\varphi(\Lambda_t)^{\theta_C}+(1-n_{Ct,1}^*)\rho^{\frac{1}{\delta_M}}\varphi(\Lambda_t)^{\theta_M}\right]^2
\right/\left[\sigma_D-\frac{h}{\varepsilon}
\frac{\varphi'(\Lambda_t)}{\varphi(\Lambda_t)}\left(\theta_C(1-y_t)+\theta_My_t\right)\right]^2,
\end{align*}
and finally derive (\ref{thnC1}).
\\
\\
\textbf{\large Proof of Corollary \ref{co1}}\\
The benchmark condition $p=1$ implies, by (\ref{protectC}), that $x_t\equiv 0$ and $f(x,D)\equiv 0$. Substituting $p=1$ and $x_t^{\mathfrak{B}}=0$ into Theorem \ref{th2}, it suffices to obtain $n_{Ct}^{\mathfrak{B}}$ to complete the proof. Indeed, by setting $p=1$ in (\ref{J}) and (\ref{thnC1}), we know that the condition $\frac{1}{1+(1-p)k}=1$ implies $n_{Ct}^{\mathfrak{B}}$ could be discussed equivalently in Region 1, and (\ref{conc}) is obtained easily.
\\
\\
\textbf{\large Proof of Theorem \ref{th4.1}}\\
To be brief, $C_q$ (or $C$) will denote a nonnegative constant depending only on $q$ whose value may change from line to line. For $n=0,1,\cdots,N-1$, define an approximation scheme used in definitions of It\^{o} integrals as
\begin{align}
\widetilde{Z}_{t_{n+1}}-\widetilde{Z}_{t_{n}}=&\left(\mu_D+\sigma_D\Lambda_{t_{n}}-H\varepsilon^{2h}\sigma_D^2\right)\Delta t+\sqrt{2H}\varepsilon^h\sigma_D \Delta w_{t_{n+1}},\quad \widetilde{Z}_{t_{0}}={Z}_{t_{0}};\label{s3}\\
\widetilde{\Lambda}_{t_{n+1}}=&\sum_{k=0}^n\sqrt{2H}h(t_{n+1}-t_k+\varepsilon)^{h-1}\Delta {w}_{t_{k+1}},\quad \widetilde{\Lambda}_{t_0}=0,\label{s4}\\
\widetilde{\lambda}_{t_{n+1}}=&\sum_{k=0}^n\sqrt{2H}h(h-1)(t_{n+1}-t_k+\varepsilon)^{h-2}\Delta {w}_{t_{k+1}},\quad \widetilde{\lambda}_{t_0}=0,\label{ss4}
\end{align}
where $\Delta w_{t_{n+1}}=w_{t_{n+1}}-w_{t_{n}}$. Set \[Y_{n+1}=(\widetilde{Z}_{t_{n+1}}-\widetilde{Z}_{t_{n}})-(Z_{t_{n+1}}-Z_{t_{n}})\]
for $n=0,1,\cdots,N-1$. In order to prove Theorem \ref{th4.1}, we firstly give several lemmas.
\begin{lemma}{\it
For all $q\geq 1$ and $n=0,1,\cdots,N-1$, there exists a nonnegative constant $C_q$ such that
\begin{equation}
\mathbb{E}|Y_{n+1}|^q\leq C_q (\Delta t)^{\frac{3}{2}q}.
\end{equation}
}
\end{lemma}
\begin{proof}
By definition (\ref{s3}), we may rewrite $\widetilde{Z}_{t_{n+1}}-\widetilde{Z}_{t_{n}}$ as
\[
\widetilde{Z}_{t_{n+1}}-\widetilde{Z}_{t_{n}}=\int_{t_{n}}^{t_{n+1}}\left(\mu_D+\sigma_D\Lambda_{t_{n}}-H\varepsilon^{2h}\sigma_D^2\right)dt
+\int_{t_{n}}^{t_{n+1}}\sqrt{2H}\varepsilon^h\sigma_D d w_t
\]
Then it is seen that
\[
Y_{n+1}=\int_{t_{n}}^{t_{n+1}}\sigma_D(\Lambda_{t_{n}}-\Lambda_{t})dt.
\]

For $q=1$,
\begin{equation}
\mathbb{E}|Y_{n+1}|\leq \sigma_D\int_{t_{n}}^{t_{n+1}}\mathbb{E}\left|\Lambda_{t_{n}}-\Lambda_{t}\right|dt;\label{y1}
\end{equation}
and for $q>1$,
\begin{align}
\mathbb{E}|Y_{n+1}|^q\leq &\sigma_D^q\mathbb{E}\left[\int_{t_{n}}^{t_{n+1}}\left|\Lambda_{t_{n}}-\Lambda_{t}\right|dt\right]^q\nonumber\\
\leq&\sigma_D^q\mathbb{E}\left[\left(\int_{t_{n}}^{t_{n+1}}\left|\Lambda_{t_{n}}-\Lambda_{t}\right|^qdt\right)^{\frac{1}{q}}
\left(\int_{t_{n}}^{t_{n+1}}1^{q_0}dt\right)^{\frac{1}{q_0}}\right]^q\nonumber\\
=&\sigma_D^q(\Delta t)^{\frac{q}{q_0}}\int_{t_{n}}^{t_{n+1}}\mathbb{E}\left|\Lambda_{t_{n}}-\Lambda_{t}\right|^qdt,\label{yp}
\end{align}
where $\frac{1}{q}+\frac{1}{q_0}=1$ and we use H\"{o}lder inequality.

Noticing the definition of $\Lambda_{t}$, we know that $\Lambda_{t_{n}}-\Lambda_{t}\,(t\in [t_{n},t_{n+1}])$ possesses a normal distribution with the mean $0$ and some variance $\sigma_n^2$ for all $n$ (here the constant random variables are considered to be normal random variables with variances $0$). If we can prove that there exists a nonnegative constant $C$ such that
\begin{equation}\label{si}
\sigma_n^2\leq C\Delta t,
\end{equation}
then it is from (\ref{y1})-(\ref{yp}) and moment properties of normal distribution that
\begin{equation}
\mathbb{E}|Y_{n+1}|^q\leq\left\{{\begin{array}{*{20}{l}}
{C_1\int_{t_{n}}^{t_{n+1}}(\Delta t)^{\frac{1}{2}}dt,}& {q=1,}\\
{C_q(\Delta t)^{\frac{q}{q_0}}\int_{t_{n}}^{t_{n+1}}(\Delta t)^{\frac{q}{2}}dt,}&{ q>1,}
\end{array}}\right.=C_q(\Delta t)^{\frac{3}{2}q},
\end{equation}
which just prove the lemma.

Indeed, for each $n$ and $t\in [t_{n},t_{n+1}]$, we have
\begin{align}\label{I0}
\sigma_n^2=&\mathbb{E}(\Lambda_{t_{n}}-\Lambda_{t})^2\nonumber\\
=&C\mathbb{E}\left[\int_0^t\left(\mathbf{1}_{s\leq t_n}(t_n-s+\varepsilon)^{h-1}-(t-s+\varepsilon)^{h-1}\right)dw_s\right]^2\nonumber\\
=&C\int_0^t\left(\mathbf{1}_{s\leq t_n}(t_n-s+\varepsilon)^{h-1}-(t-s+\varepsilon)^{h-1}\right)^2ds\nonumber\\
=&C\left[\int_0^{t_n}\left((t_n-s+\varepsilon)^{h-1}-(t-s+\varepsilon)^{h-1}\right)^2ds+\int_{t_n}^t(t-s+\varepsilon)^{2(h-1)}ds\right]
\nonumber\\
=&C(I_1+I_2),
\end{align}
where
\begin{align*}
I_1=&\int_{t_n}^t(t-s+\varepsilon)^{2(h-1)}ds,\\
I_2=&\int_0^{t_n}\left((t_n-s+\varepsilon)^{h-1}-(t-s+\varepsilon)^{h-1}\right)^2ds.
\end{align*}
It is easy to see by $2(h-1)\in(-3,-1)$ that
\begin{equation}\label{I1}
I_1\leq\int_{t_n}^t\varepsilon^{2(h-1)}ds\leq C\Delta t.
\end{equation}
Since $\Delta t<T$ and $h-2\in(-\frac{5}{2},-\frac{3}{2})$, we have
\begin{align*}
&|(t_n-s+\varepsilon)^{h-1}-(t-s+\varepsilon)^{h-1}|\\
=&|(h-1)\xi^{h-2}(t-t_n)|\quad (\xi\in(t_n-s+\varepsilon,t-s+\varepsilon))\\
\leq &(1-h)\varepsilon^{h-2}|t-t_n|\\
\leq & C\Delta t,
\end{align*}
which clearly entails
\begin{equation}\label{I2}
I_2\leq \int_0^{t_n}(C\Delta t)^2ds\leq C (\Delta t)^2\leq C\Delta t.
\end{equation}
Substituting (\ref{I1})-(\ref{I2}) into (\ref{I0}), we obtain (\ref{si}), which completes the proof.
\end{proof}

\begin{lemma}\label{lemma2}{\it
Define
\[V_N=\sum_{n=0}^{N-1}|Y_{n+1}|. \]
Then for all $\varsigma>0$, there exists a nonnegative random variable $K_{1\varsigma}$ (depending only on $\varsigma$ with $\mathbb{E}K_{1\varsigma}^q<+\infty$ for all $q\geq 1$) such that
\[
V_N(\omega)\leq K_{1\varsigma}(\omega)(\Delta t)^{\frac{1}{6}-\varsigma}
\]
for almost all $\omega\in\Omega$.
}
\end{lemma}
\begin{proof}
We first recall that for any nonnegative $a_1,a_2,\cdots,a_{N}$,
\begin{align*}
\left(\sum_{n=0}^{N-1}a_{n+1}\right)^q\leq &\left(N\cdot\max\limits_{n=0,1,\cdots,N-1}a_{n+1}\right)^q\leq N^q\left(\sum_{n=0}^{N-1}a_{n+1}^q\right),\quad q\geq 1.
\end{align*}
In the case $q\geq 3$,
\begin{align*}
\mathbb{E}V_N^q\leq N^q\left(\sum_{n=0}^{N-1}\mathbb{E}|Y_{n+1}|^q\right)\leq N^q\left(\sum_{n=0}^{N-1}C_q (\Delta t)^{\frac{3}{2}q}\right)\leq C_q(\Delta t)^{\frac{q}{2}-1},
\end{align*}
which implies by $\Delta t\leq 1$ that
\begin{equation}\label{v3}
\left[\mathbb{E}V_N^q\right]^{\frac{1}{q}}\leq C_q(\Delta t)^{\frac{1}{2}-\frac{1}{q}}\leq C_q(\Delta t)^{\frac{1}{6}}.
\end{equation}
In the case $1\leq q<3$,
\begin{align*}
\mathbb{E}V_N^q\leq \left[\mathbb{E}(V_N^q)^{\frac{3}{q}}\right]^{\frac{q}{3}}=\left[\mathbb{E}V_N^3\right]^{\frac{q}{3}}\leq C_q(\Delta t)^{\frac{q}{6}},
\end{align*}
which gives
\begin{equation}\label{v1}
\left[\mathbb{E}V_N^q\right]^{\frac{1}{q}}\leq C_q(\Delta t)^{\frac{1}{6}}.
\end{equation}
Hence it is from (\ref{v3}) and (\ref{v1}) that for all $q\geq 1$
\begin{equation}\label{v}
\left[\mathbb{E}V_N^q\right]^{\frac{1}{q}}\leq C_q(\Delta t)^{\frac{1}{6}}.
\end{equation}

Consequently, by applying Lemma 1 in \cite{Jentzen} or Lemma 3.4 in \cite{Jentzen2} to (\ref{v}), we obtain the statement of the lemma.
\end{proof}

\begin{lemma}\label{lemma3}{\it
Define
\begin{align*}
V_{1N}=\max\limits_{n=0,1,\cdots,N-1}\left|\Lambda_{t_{n+1}}-\widetilde{\Lambda}_{t_{n+1}}\right|,\\
V_{2N}=\max\limits_{n=0,1,\cdots,N-1}\left|\lambda_{t_{n+1}}-\widetilde{\lambda}_{t_{n+1}}\right|.
\end{align*}
Then for all $\varsigma>0$, there exists a nonnegative random variable $K_{2\varsigma}$ (depending only on $\varsigma$ with $\mathbb{E}K_{2\varsigma}^q<+\infty$ for all $q\geq 1$) such that
\[
\max\{V_{1N}(\omega),V_{2N}(\omega)\}\leq K_{2\varsigma}(\omega)(\Delta t)^{\frac{1}{6}-\varsigma}
\]
for almost all $\omega\in\Omega$.
}
\end{lemma}
\begin{proof}
For $k,n\in\{0,1,\cdots,N-1\}$, define
\[
X_{k+1}=\int_{t_k}^{t_{k+1}}\left((t_{n+1}-s+\varepsilon)^{h-1}-(t_{n+1}-t_k+\varepsilon)^{h-1}\right)dw_s,
\]
and it is seen by definitions of $\Lambda_{t_{n+1}}$ and $\widetilde{\Lambda}_{t_{n+1}}$ that
\begin{align*}
&\Lambda_{t_{n+1}}-\widetilde{\Lambda}_{t_{n+1}}\\
=&\sqrt{2H}h\sum_{k=0}^n\int_{t_k}^{t_{k+1}}(t_{n+1}-s+\varepsilon)^{h-1}dw_s
-\sqrt{2H}h\sum_{k=0}^n\int_{t_k}^{t_{k+1}}(t_{n+1}-t_k+\varepsilon)^{h-1}dw_s\\
=&\sqrt{2H}h\sum_{k=0}^nX_{k+1}
\end{align*}
and
\begin{equation}
V_{1N}\leq\sqrt{2H}|h|\sum_{k=0}^{N-1}|X_{k+1}|.
\end{equation}

Since $X_{k+1}$ possesses a normal distribution with mean $0$ and variance satisfying
\begin{align*}
\mathbb{E}X_{k+1}^2=&\int_{t_k}^{t_{k+1}}\left((t_{n+1}-s+\varepsilon)^{h-1}-(t_{n+1}-t_k+\varepsilon)^{h-1}\right)^2ds\\
=&\int_{t_k}^{t_{k+1}}(h-1)^2\xi^{2(h-2)}(s-t_k)^2ds\quad (\xi\in(t_{n+1}-s+\varepsilon,t_{n+1}-t_k+\varepsilon))\\
\leq&\int_{t_k}^{t_{k+1}}(h-1)^2\varepsilon^{2(h-2)}(\Delta t)^2ds\\
=&(h-1)^2\varepsilon^{2(h-2)}(\Delta t)^3,
\end{align*}
we get by moment properties of the normal distribution that for all $q\geq 1$ there exists a nonnegative constant $C_q$ such that
\[
\mathbb{E}|X_{k+1}|^q\leq C_q (\Delta t)^{\frac{3}{2}q}.
\]
Now by a similar proof of (\ref{v}), it is seen that for all $q\geq 1$
\begin{equation}\label{U}
\left[\mathbb{E}V_{1N}^q\right]^{\frac{1}{q}}\leq C_q(\Delta t)^{\frac{1}{6}},
\end{equation}
which implies by Lemma 1 in \cite{Jentzen} or Lemma 3.4 in \cite{Jentzen2} again that for all $\varsigma>0$, there exists a nonnegative random variable $L_{1\varsigma}$ depending only on $\varsigma$ with $\mathbb{E}L_{1\varsigma}^q<+\infty$ for all $q\geq 1$ such that
\[
V_{1N}(\omega)\leq L_{1\varsigma}(\omega)(\Delta t)^{\frac{1}{6}-\varsigma}
\]
for almost all $\omega\in\Omega$.

In a similar way, it is easy to show that for all $\varsigma>0$, there exists a nonnegative random variable $L_{2\varsigma}$ depending only on $\varsigma$ with $\mathbb{E}L_{2\varsigma}^q<+\infty$ for all $q\geq 1$ such that
\[
V_{2N}(\omega)\leq L_{2\varsigma}(\omega)(\Delta t)^{\frac{1}{6}-\varsigma}
\]
for almost all $\omega\in\Omega$. Then by setting
\[K_{2\varsigma}=\max\{L_{1\varsigma},L_{2\varsigma}\},\]
we complete the proof.
\end{proof}
Now we can prove Theorem \ref{th4.1}. For almost all $\omega\in\Omega$ (which we shall omit in our following proof for simplicity) and $n=0,1,\cdots,N-1$, it is seen from (\ref{s1}) and (\ref{s3}) that
\begin{align}\label{www}
\left|\Delta{w}_{t_{n+1}}-\Delta \widehat{w}_{t_{n+1}}\right|\leq \frac{|Y_{n+1}|+\sigma_D\Delta t|\Lambda_{t_{n}}-\widehat{\Lambda}_{t_{n}}|}{\sqrt{2H}\varepsilon^h\sigma_D}.
\end{align}
Then we obtain
\begin{align*}
&|\Lambda_{t_{n+1}}-\widehat{\Lambda}_{t_{n+1}}|\\
\leq &|\Lambda_{t_{n+1}}-\widetilde{\Lambda}_{t_{n+1}}|+|\widetilde{\Lambda}_{t_{n+1}}-\widehat{\Lambda}_{t_{n+1}}|\\
\leq &V_{1N}+\sum_{k=0}^{n}\sqrt{2H}|h|\varepsilon^{h-1}\left|\Delta{w}_{t_{k+1}}-\Delta \widehat{w}_{t_{k+1}}\right|\,((\ref{s2})\,\text{and}\,(\ref{s4}))\\
\leq& (\Delta t)^{\frac{1}{6}-\varsigma}K_{2\varsigma}+\frac{|h|}{\varepsilon\sigma_D}\sum_{k=0}^n(|Y_{k+1}|+\sigma_D\Delta t|\Lambda_{t_{k}}-\widehat{\Lambda}_{t_{k}}|)\,(\text{Lemma}\, \ref{lemma3}\;\text{and}\;(\ref{www}))\\
\leq&(\Delta t)^{\frac{1}{6}-\varsigma}K_{2\varsigma}+\frac{|h|}{\varepsilon\sigma_D}(\Delta t)^{\frac{1}{6}-\varsigma}K_{1\varsigma}+\frac{|h|\Delta t}{\varepsilon}\sum_{k=0}^n|\Lambda_{t_{k}}-\widehat{\Lambda}_{t_{k}}|(\text{Lemma}\, \ref{lemma2})\\
=&(\Delta t)^{\frac{1}{6}-\varsigma}K_{3\varsigma}+\frac{|h|\Delta t}{\varepsilon}\sum_{k=0}^n|\Lambda_{t_{k}}-\widehat{\Lambda}_{t_{k}}|,
\end{align*}
where
\[K_{3\varsigma}=K_{2\varsigma}+\frac{|h|}{\varepsilon\sigma_D}K_{1\varsigma}.\]
Using the discrete Gronwall lemma (see Lemma 1.4.2 in \cite{Quarteroni}), we have
\begin{align}
|\Lambda_{t_{n+1}}-\widehat{\Lambda}_{t_{n+1}}|\leq& (\Delta t)^{\frac{1}{6}-\varsigma}K_{3\varsigma}\exp\left\{\sum_{k=0}^n\frac{|h|\Delta t}{\varepsilon}\right\}\nonumber\\
\leq&(\Delta t)^{\frac{1}{6}-\varsigma}K_{3\varsigma}\exp\left\{\frac{|h|T}{\varepsilon}\right\}.\label{ww}
\end{align}
Making use of Lemma \ref{lemma3}, (\ref{www}) and (\ref{ww}), we have
\begin{align}
&|\lambda_{t_{n+1}}-\widehat{\lambda}_{t_{n+1}}|\nonumber\\
\leq&|\lambda_{t_{n+1}}-\widetilde{\lambda}_{t_{n+1}}|+|\widetilde{\lambda}_{t_{n+1}}-\widehat{\lambda}_{t_{n+1}}|\nonumber\\
\leq&(\Delta t)^{\frac{1}{6}-\varsigma}K_{2\varsigma}+\sum_{k=0}^{n}\sqrt{2H}|h(h-1)|\varepsilon^{h-2}\left|\Delta{w}_{t_{k+1}}-\Delta \widehat{w}_{t_{k+1}}\right|\nonumber\\
\leq&(\Delta t)^{\frac{1}{6}-\varsigma}K_{2\varsigma}+\frac{|h(h-1)|}{\varepsilon^2\sigma_D}(\Delta t)^{\frac{1}{6}-\varsigma}K_{1\varsigma}+\frac{|h(h-1)|\Delta t}{\varepsilon^2}\sum_{k=0}^n|\Lambda_{t_{k}}-\widehat{\Lambda}_{t_{k}}|\nonumber\\
\leq&(\Delta t)^{\frac{1}{6}-\varsigma}K_{4\varsigma},
\end{align}
where
\[K_{4\varsigma}=K_{2\varsigma}+\frac{|h(h-1)|}{\varepsilon^2\sigma_D}K_{1\varsigma}+\frac{|h(h-1)|T }{\varepsilon^2}\exp\left\{\frac{|h|T}{\varepsilon}\right\}K_{3\varsigma}.\]

By setting
\[
K_{\varsigma}=\max\left\{K_{3\varsigma}\exp\left\{\frac{|h|T}{\varepsilon}\right\},K_{4\varsigma}\right\},
\]
it is telling that $\mathbb{E}K_\varsigma^q<+\infty$ for all $q\geq 1$ and
\begin{align}
|\Lambda_{t_{n+1}}-\widehat{\Lambda}_{t_{n+1}}|\leq K_\varsigma(\Delta t)^{\frac{1}{6}-\varsigma},\label{ff1}\\
|\lambda_{t_{n+1}}-\widehat{\lambda}_{t_{n+1}}|\leq K_\varsigma(\Delta t)^{\frac{1}{6}-\varsigma}.\label{ff2}
\end{align}
Finally noticing right sides of (\ref{ff1})-(\ref{ff2}) are not related to $n$, we just prove the theorem.
\\
\\
\textbf{\large Lemma \ref{lastlemma} used in Remark \ref{remark_intro}}\\
\begin{lemma}\label{lastlemma}{\it
For $n=0,1,\cdots,N-1$, define a process $Q$ as
\[Q_{t_{n+1}}=Q_{t_n}+\left(\mu_D+\sigma_D\widetilde{\Lambda}_{t_{n}}-H\varepsilon^{2h}\sigma_D^2\right)\Delta t+\sqrt{2H}\varepsilon^h\sigma_D \Delta w_{t_{n+1}},\quad Q_{t_{0}}={Z}_{t_{0}},\]
where $\widetilde{\Lambda}_{t_{n}}$ is given by (\ref{s4}). Then for all $\varsigma>0$, there exists a nonnegative random variable $K_{A\varsigma}$ (depending only on $\varsigma$ with $\mathbb{E}K_{A\varsigma}^q<+\infty$ for all $q\geq 1$) such that
\[
\max_{n=0,1,\cdots,N-1}|Q_{t_{n+1}}(\omega)-Z_{t_{n+1}}(\omega)|\leq K_{A\varsigma}(\omega)(\Delta t)^{\frac{1}{6}-\varsigma}
\]
for almost all $\omega\in\Omega$.
}
\end{lemma}
\begin{proof}
We can rewrite $Q_{t_{n+1}}$ as
\[
Q_{t_{n+1}}=Q_{t_n}+\int_{t_n}^{t_{n+1}}\left(\mu_D+\sigma_D\widetilde{\Lambda}_{t_{n}}-H\varepsilon^{2h}\sigma_D^2\right)d t+\int_{t_n}^{t_{n+1}}\sqrt{2H}\varepsilon^h\sigma_D dw_t,
\]
which gives
\begin{align*}
&|Q_{t_{n+1}}-Z_{t_{n+1}}|\\
\leq&|Q_{t_{n}}-Z_{t_{n}}|+\left|\int_{t_n}^{t_{n+1}}\sigma_D(\widetilde{\Lambda}_{t_{n}}-\Lambda_t)dt\right|\\
\leq&|Q_{t_{n}}-Z_{t_{n}}|+|Y_{{n+1}}|+\sigma_D\Delta t|\widetilde{\Lambda}_{t_{n}}-\Lambda_{t_{n}}|\\
\leq&|Q_{t_{n}}-Z_{t_{n}}|+|Y_{{n+1}}|+\sigma_D\Delta tV_{1N}.
\end{align*}
By induction, it follows from Lemmas \ref{lemma2} and \ref{lemma3} that
\[
|Q_{t_{n+1}}-Z_{t_{n+1}}|\leq V_N+\sigma_DTV_{1N}
\leq (K_{1\varsigma}+\sigma_DTK_{2\varsigma})(\Delta t)^{\frac{1}{6}-\varsigma}
\]
for all $\varsigma>0$. We just complete the proof by setting $K_{A\varsigma}=K_{1\varsigma}+\sigma_DTK_{2\varsigma}$.
\end{proof}

\end{document}